\documentclass[a4paper, 12pt]{article} 
\usepackage[centertags]{amsmath}
\usepackage{amsmath}
\usepackage[dutch,british]{babel}
\usepackage{amsfonts}
\usepackage{amssymb} 
\usepackage{amsthm}
\usepackage{newlfont}
\usepackage{color}
\usepackage{subfig}
 \usepackage{bbm}
 \usepackage{float}
 \usepackage{import}

\usepackage{stmaryrd}
\usepackage{mathrsfs}
\usepackage{pstricks,pst-node}
\usepackage{palatino}  %palatcm - a missing package
\usepackage{fancyhdr}
\usepackage{graphicx}
\usepackage{fancybox}
\usepackage{multibox}
\usepackage{geometry}
 \usepackage{psfrag}
 
 \usepackage{multirow}

\theoremstyle{plain}
  \newtheorem{theorem}{Theorem}[section]
      \newtheorem{assumption}{Assumption}
  
  \newtheorem{proposition}[theorem]{Proposition}
  \newtheorem{lemma}[theorem]{Lemma}
  \newtheorem{remark}[theorem]{Remark}
\theoremstyle{definition}
  \newtheorem{definition}{Definition}[section]

\theoremstyle{remark}

\numberwithin{equation}{section}

\DeclareMathOperator{\Tr}{Tr}

 \DeclareMathOperator{\supp}{Supp}
\renewcommand{\Re}{\mathrm{Re}\, }
\renewcommand{\Im}{\mathrm{Im}\,}

\newcommand\otimesal{\mathop{\hbox{\raise 1.6 ex
  \hbox{$\scriptscriptstyle\mathrm{al}$}
\kern -0.92 em \hbox{$\otimes$}}}}
\newcommand\oplusal{\mathop{\hbox{\raise 1.6 ex
  \hbox{$\scriptscriptstyle\mathrm{al}$}
\kern -0.92 em \hbox{$\oplus$}}}}
\newcommand\Gammal{\hbox{\raise 1.7 ex
\hbox{$\scriptscriptstyle\mathrm{al}$}\kern -0.50 em $\Gamma$}}
\renewcommand\i{\mathrm{i}}

\let\al=\alpha \let\be=\beta  \let\ep=\epsilon

\let\ve=\varepsilon  \let\ga=\gamma 
\let\ka=\kappa \let\la=\lambda \let\om=\omega 
\let\si=\sigma

 \let\Ga=\Gamma  \let\Om=\Omega
  \let\Si=\Sigma

\newcommand{\caA}{{\mathcal A}}
\newcommand{\caB}{{\mathcal B}}
\newcommand{\caC}{{\mathcal C}}
\newcommand{\caD}{{\mathcal D}}
\newcommand{\caE}{{\mathcal E}}

\newcommand{\caG}{{\mathcal G}}

\newcommand{\caJ}{{\mathcal J}}

\newcommand{\caL}{{\mathcal L}}
\newcommand{\caM}{{\mathcal M}}

\newcommand{\caO}{{\mathcal O}}

\newcommand{\caR}{{\mathcal R}}
\newcommand{\caS}{{\mathcal S}}
\newcommand{\caT}{{\mathcal T}}

\newcommand{\caV}{{\mathcal V}}
\newcommand{\caW}{{\mathcal W}}

\newcommand{\scrB}{{\mathscr B}}

\newcommand{\scrD}{{\mathscr D}}
\newcommand{\scrE}{{\mathscr E}}

\newcommand{\scrG}{{\mathscr G}}
\newcommand{\scrH}{{\mathscr H}}

\newcommand{\scrR}{{\mathscr R}}
\newcommand{\scrS}{{\mathscr S}}
\newcommand{\scrT}{{\mathscr T}}

\newcommand{\bbC}{{\mathbb C}}

\newcommand{\bbE}{{\mathbb E}}

\newcommand{\bbN}{{\mathbb N}}

\newcommand{\bbQ}{{\mathbb Q}}
\newcommand{\bbR}{{\mathbb R}}
\newcommand{\bbS}{{\mathbb S}}

\newcommand{\opunit}{\text{1}\kern-0.22em\text{l}}

\newcommand{\frh}{{\mathfrak h}}

\newcommand{\frt}{{\mathfrak t}}

\newcommand{\frB}{{\mathfrak B}}

\newcommand{\frG}{{\mathfrak G}}

\newcommand{\frS}{{\mathfrak S}}

\newcommand{\frW}{{\mathfrak W}}

\newcommand{\bsI}{{\boldsymbol I}}

\newcommand{\e}{{\mathrm e}}

\renewcommand{\d}{{\mathrm d}}
\newcommand{\sys}{{\mathrm S}}
\newcommand{\res}{{\mathrm F}}

\renewcommand{\sp}{\sigma}
\newcommand{\Ran}{\mathrm{Ran}}

\newcommand{\Dom}{\mathrm{Dom}}
\newcommand{\spann}{\mathrm{Span}}
\newcommand{\beq}{ \begin{equation} }
\newcommand{\eeq}{ \end{equation} }
\newcommand{\bet}{ \begin{theorem} }
\newcommand{\eet}{ \end{theorem} }

\newcommand{\initial}{0}
\newcommand{\final}{n+1}

\newcommand{\lone}{\mathbbm{1}}

\newcommand{\baq}{\begin{eqnarray}}
\newcommand{\eaq}{\end{eqnarray}}

\renewcommand{\supp}{\mathrm{Supp}}

 \newcounter{smallarabics}
\newenvironment{arabicenumerate}
{\begin{list}{{\normalfont\textrm{\arabic{smallarabics})}}}
  {\usecounter{smallarabics}\setlength{\itemindent}{0cm}
  \setlength{\leftmargin}{5ex}\setlength{\labelwidth}{4ex}
  \setlength{\topsep}{0.75\parsep}\setlength{\partopsep}{0ex}
   \setlength{\itemsep}{0ex}}}
{\end{list}}

\newcounter{smallroman}

\newcommand{\ben}{\begin{arabicenumerate}}
\newcommand{\een}{\end{arabicenumerate}}

\newcommand{\norm}{ \|}
\newcommand{\str}{ |}

\newcommand{\initialresfinite}{P_{\Om}}

\newcommand{\realinitial}{\ltimes}
\newcommand{\realfinal}{\rtimes}
\newcommand{\adjoint}{\mathrm{ad}}

\newcommand{\uv}{\underline{v}}
\newcommand{\uu}{\underline{u}}
\newcommand{\uw}{\underline{w}}

\newcommand{\dist}{d}
\newcommand{\distance}{\mathrm{dist}}

\newcommand{\weird}{\diamond}
\newcommand{\normw}{\norm_{\weird}}

\newcommand{\inter}{\mathrm{I}}
\newcommand{\gap}{\mathrm{g}}

\newcommand{\opprod}{\otimes}
\newcommand{\indicator}{1}

\newcommand{\commpictone}{\scriptsize{\text{Pairs within $T$-blocks resummed}}}
\newcommand{\commpicttwo}{\scriptsize{\text{Pairs determining $G^c_A$ for given $A$ resummed}}}
\newcommand{\oneone}{{\scriptscriptstyle u_1}}
\newcommand{\onetwo}{{\scriptscriptstyle v_1}}
\newcommand{\twoone}{{\scriptscriptstyle u_2}}
\newcommand{\twotwo}{{\scriptscriptstyle v_2}}
\newcommand{\threeone}{{\scriptscriptstyle u_3}}
\newcommand{\threetwo}{{\scriptscriptstyle v_3}}

\begin{document}

\begin{center}
\large{ \bf{Approach to ground state and time-independent  photon bound  for massless spin-boson models}}
 \\
\vspace{15pt} \normalsize

{\bf   W.  De Roeck\footnote{
email: {\tt
 w.deroeck@thphys.uni-heidelberg.de}}  }\\
\vspace{10pt} 
{\it   Institut f\"ur Theoretische Physik  \\ Universit\"at Heidelberg \\
Philosophenweg 16,  \\
D69120 Heidelberg,  Germany 
} \\

\vspace{15pt}

{\bf   A. Kupiainen\footnote{
email: {\tt    antti.kupiainen@helsinki.fi  }}  }\\
\vspace{10pt} 
{\it   Department of Mathematics \\
 University of Helsinki \\ 
P.O. Box 68, FIN-00014,  Finland 
} \\

\vspace{15pt}

\end{center}

\vspace{20pt} \footnotesize \noindent {\bf Abstract:}  It is widely believed that an atom interacting with the electromagnetic field (with total initial energy well-below the ionization threshold) relaxes to its ground state while its excess energy is emitted as radiation.
Hence, for large times, the  state of the atom+field system should consist of  the atom in its ground state, and a few free photons that travel off to spatial infinity.   Mathematically, this picture is captured by the notion of \emph{asymptotic completeness}. Despite some recent progress on the spectral theory of such systems, a proof of relaxation to the ground state and asymptotic completeness  was/is still missing, except in some special cases (massive photons, small perturbations of harmonic potentials).  In this paper, we partially fill this gap by proving  relaxation to an invariant state in the case where the atom is modelled by a finite-level system.  If the coupling to the field is sufficiently infrared-regular so that the coupled system admits a ground state, then this invariant state necessarily corresponds to the ground state.  Assuming slightly more infrared regularity, we show that the number of emitted photons remains bounded in time.  We hope that these results bring a proof of asymptotic completeness within reach.

  \section{Model and result}

\subsection{Introduction}\label{secintroduction}

This paper fits into a broader project of rigorously controlling interacting Hamiltonian systems that exhibit irreversible behavior.   From a more concrete point of view, the problem treated here is inspired by nonrelativistic QED which,  in the last two decades, has proven to be a fruitful testing ground for mathematical techniques in quantum field theory (for an overview, see the book \cite{spohnbook})
Sacrificing  precision for the time being, the setup is as follows. We consider a model of an atom interacting with the electromagnetic field (or, a scalar field, as we will assume in this paper for the sake of simplicity).   The field is described by the Hamiltonian $H_\res$ on Hilbert space $\scrH_\res$, it is non-interacting and describes freely propagating scalar bosons with a linear dispersion law (the polarization of the photons does not play a crucial role in the physical problem, hence we omit it from our model).  The joint dynamical system (atom $+$ field) is described by a Hamiltonian $H$ on a Hilbert space $\scrH$ and the field is coupled in a nontrivial way to the atom.  Furthermore, we assume that the atom cannot be ionized, i.e.\  in the absence of coupling its spectrum is discrete. 
 We wish to  address  the long time behavior of the system.

\subsubsection{Approach to a stationary state}
Let $O$ be a local observable in $\scrB(\scrH)$ (i.e.\ bounded operators on $\scrH$) that should be thought of as probing the atom and the field in its spatial vicinity.  Let $\Psi_0 \in \scrH, \norm \Psi_0 \norm=1$ be the initial state vector ($0$ refers to time $t=0$) and $\Psi_t := \e^{-\i t H} \Psi_0$ is the time-evolved state vector. 
We consider the expectation value  ( $\langle \cdot, \cdot \rangle $ is the scalar product)
\beq
\langle O \rangle_t := \langle \Psi_t,  O \Psi_t \rangle   \label{eq: definition time evolved expectation value}
\eeq
As explained in the abstract, we substantiate the claim that this expression converges to an asymptotic value, as $t \to \infty$. i.e.\   $\langle O \rangle_t \to \langle O \rangle_{\infty} $ and that $\langle O \rangle_{\infty}$ is independent of $\Psi_0$.
 If the Hamiltonian $H$ admits a ground state $\Psi_{\mathrm{gs}}$, i.e.\ $E_{\mathrm{gs}}= \inf \si(H)$ is an eigenvalue, then this  implies that $ \langle O \rangle_{\infty} = \langle \Psi_{\mathrm{gs}},  O \Psi_{\mathrm{gs}} \rangle   $ since one can choose $\Psi_0=\Psi_{\mathrm{gs}}$.
   In particular, this means that  the eigenvalue $E_{\mathrm{gs}}$ is simple and that $H$ cannot have any other eigenvalues.  It is this last claim that has been established up to now in great generality. In fact,  one can even prove that  $H$ has absolutely continuous spectrum, apart from the simple eigenvalue $E_{\mathrm{gs}}$, see \cite{bachfrohlichsigalqed,frohlichgriesemersigal,georgescugerardmoeller,chenfaupinfrohlichsigal}.
 However, this seems not sufficient to prove that \eqref{eq: definition time evolved expectation value} converges as $ t \to \infty$,  because $\e^{-\i t H}$ appears twice in this expression (!).  Let us immediately add that we find it not at all inconceivable that there is some easy way around this problem, allowing to apply the techniques used for the spectral analysis of $H$ to determine the asymptotics of \eqref{eq: definition time evolved expectation value}.
However, this is not the strategy of the present paper.  We prove the convergence to an asymptotic value, $\langle O \rangle_t \to \langle O \rangle_{\infty} $ not by spectral considerations, but by exhibiting explicitly the irreversible  density matrix evolution  $\str \Psi_0 \rangle \langle \Psi_0 \str \to \str \Psi_t \rangle \langle \Psi_t \str$. 
 If $H$ is too infrared-singular to admit a ground state, then the convergence  $\langle O \rangle_t \to \langle O \rangle_{\infty} $ can still hold provided that the observable $O$ 'does not see' the low-energetic photons.  In that case the  asymptotic value $ \langle O \rangle_{\infty} $ is a state (a positive, normalized functional in $O$) that is not of the form $\langle O \rangle_{\infty} = \langle \Psi_{\mathrm{gs}},  O \Psi_{\mathrm{gs}} \rangle$. From the point of view of our technique, this case is no different from the infrared-regular case.   
 
 Finally, we mention that the problem of `return to equilibrium' at positive temperature has been studied with much more success since in that case, the problem can indeed be reduced to the study of the spectrum of an operator - the so-called \emph{standard Liouvillian} -  acting on an appropriate Hilbert space, see \cite{deroeckkupiainen} for references. For our technique, there is no difference between zero and  positive temperature, and the present result on approach to a stationary state was, up to some irrelevant details, in fact already contained in \cite{deroeckkupiainen}.

\subsubsection{Scattering theory}

Let us again assume that $H$ admits a ground state $\Psi_{\mathrm{gs}}$ (and no other eigenstates). The intuition is that the evolved state vector  $\Psi_t$ should, at large times $t$, look like the ground state with a few free photons that travel off to infinity.

 To make this intuition  more precise, one introduces an identification operator $J_{\mathrm{id}}: \scrH_\res \to \scrH$, 
% \beq
% J_{\mathrm{id}}   a^{*}(f_m) \ldots a^{*}(f_1) \Om   =   a^{*}(f_m) \ldots a^{*}(f_1)  \Psi_{\mathrm{gs}}
% \eeq
% for $f_i \in L^d(\bbR^d)$ and $a^{}$
% 
  such that $J_{\mathrm{id}}  \Om =  \Psi_{\mathrm{gs}}  $ with $\Omega$ the field vacuum and in general $J_{\mathrm{id}}$ maps photon state vectors $\Psi_{F}$ into state vectors that we think of as 'ground state together with a free photon wavepacket $\Psi_{F}$'. One defines the asymptotic wave operators $W^+,W^-$ by 
\beq
W^{\pm} \Psi_{F} := \mathop{\lim}\limits_{t \to  \pm \infty}        \e^{ \i t (H-E_{\mathrm{gs}})} J_{\mathrm{id}}  \e^{- \i t H_\res} \Psi_{F}  \label{eq: def wave operators}
\eeq
for  $\Psi_{F}$ in a dense domain.  Asymptotic completeness (AC) asserts that the operators  $W^{\pm}:  \scrH_\res \to \scrH$ exist and are unitary. The main unproven aspect of this statement is $\Ran W^{\pm} = \scrH$.  Indeed, what can happen  in the presence of (massless) photons is that $\Psi_t= \e^{-\i t H} \Psi_0$ will contain ever more (as $t \to \infty$) photons with ever smaller energies.  But if this is the case then $\Psi_0=\e^{\i t H} \Psi_t$ can clearly not  equal $W^{+} \Psi_{F} $   for any $\Psi_{F}  \in \scrH_\res$.  
For this reason it is important to bound the number of photons of the state vector  $\Psi_t$.    
Indeed, if one eliminates the soft photons, either by making them massive or by introducing a sharp infrared cutoff in the coupling, then all questions can be answered and in particular AC follows, see \cite{frohlichgriesemerschlein, derezinskigerardmassive}.
Our result provides a strong exponential bound  of the form $\langle \Psi_t, \e^{\ka N}  \Psi_t \rangle \leq C$ with $\ka $ a small positive constant and $C$ independent of time.  We believe  that with this bound as an input, a proof of asymptotic completeness is  within reach (see e.g.\ \cite{gerardscatteringmasslessnelson}) and we hope to pursue this in a subsequent paper.  Apart from  the case of massive bosons,  asymptotic completeness  can be established if the particle is a harmonic oscillator and the coupling to the field is linear, so that the full Hamiltonian can be explicitly diagonalized \cite{arairigorous}. Small perturbations (due to small anharmonicities in the particle potential) can be handled by an  expansion introduced in \cite{maassenthermal}, see \cite{spohnasymptoticcompleteness}.  Of course, if one considers models not described by quantum field theory, for example; $N$-particle quantum mechanics, more results are available and we refer to \cite{frohlichgriesemerschlein} for more references and background.

\subsection{Setup}\label{secsetup}

Let $\scrH_\sys$ be a Hilbert space (modeling the small system) with a self-adjoint  Hamiltonian $H_\sys$. 
The field is given by the one-particle dispersion relation $\str q \str$ and the Hamiltonian of the whole field  is given by 
\beq
H_\res:= \int_{\bbR^d} \d q \str q \str a^*_q a_q
\eeq
acting on the bosonic (symmetric) Fock space $\scrH_\res = \Ga(\frh)$ with $\frh= L^2(\bbR^d)$ the one-particle space.
Here $a_q^*, a_q$ are the creation/annihilation operators (actually, operator-valued distributions) of a mode with momentum $q \in \bbR^d$ satisfying the canonical commutation relations $[a_q, a^*_{q'}] = \delta (q-q')$.  We refer to e.g.\ \cite{derezinski1} for a review of these notions and precise definitions.
The Hilbert space of the total system consisting of small system and field, is  $\scrH= \scrH_\sys \otimes \scrH_\res$, and we simply write  $H_\sys$ and  $H_\res$ for the operators $H_{\sys} \otimes 1$ and $1 \otimes H_\res$ acting on $\scrH$.
For simplicity, the  coupling between field and the small system  
  is assumed to be linear in the creation and annihilation operators and of the form $\la H_\inter$, with $\la \in \bbR$ a (small) coupling constant and
 \beq
H_{\inter}=    D \otimes  \int_{\bbR^d} \d  q    \left( \phi(q) a^*_q+  \overline{\phi(q)}a_q \right),   \label{def: interaction ham}
\eeq
for some  $\phi \in \frh$ and Hermitian matrix  $ D=D^* \in \scrB(\scrH_\sys)$.
Our technique works equally well if one considers a finite sum of such interaction terms, or if one adds a sufficiently small quadratic interaction term, but we prefer to keep the setup as elementary as possible. 
The formal total Hamiltonian of the system is hence 
\beq \label{def: total hamiltonian}
H:= H_\sys +  H_\res+  \la H_{\inter}, \qquad  \textrm{on} \,\, \scrH.
\eeq
To construct $H$ rigorously, we assume throughout that
\beq  \label{ass: semibounded hamiltonian}
\langle \phi,  (1+\frac{1 }{ \str q \str}) \phi \rangle_{\frh} =   \int_{\bbR^d} \d  q  (1+\frac{1 }{ \str q \str}) \str \phi(q) \str^2  < \infty. 
\eeq
which yields, by a standard estimate:
\beq \label{eq: infinitesimal perturbation}
 \norm H_{\inter} \Psi  \norm^{2} \leq   2 \norm D \norm \langle \phi, (1+1 / \str q \str) \phi \rangle_{\frh}\,    \langle \Psi,  H_\res \Psi \rangle  \leq C  \langle \Psi,  H_\res \Psi \rangle, \qquad  \text{for} \,  \Psi \in \Dom(H_\res)  
 \eeq
It follows that $H_{\inter}$ is an infinitesimal perturbation of $H_\res$ and by the Kato-Rellich theorem, the Hamiltonian $H$ is self-adjoint on $\Dom (H_\sys+ H_\res)$.

The following assumption is the key ingredient of our analysis, as it expresses that correlations of the free field vanish in time sufficiently fast. 
 \renewcommand{\theassumption}{\Alph{assumption}}
\begin{assumption} [Decay of correlation functions] \label{ass: decay correlation functions}
\beq
 \int_{\bbR_{+}} \d t  \,   (1+ t )^{\al} \str h(t) \str< \infty, \qquad \textrm{with} \,\,   h(t):=  \int_{\bbR^d} \d q \,   \e^{-\i t \str q \str} \str \phi(q)\str^2,     \eeq
 for some $\al>0$. 
\end{assumption}
This assumption implies some infrared regularity of the model. In particular, if it is satisfied with $\al \geq 1$, then the Hamiltonian $H$ has a ground state\footnote{In fact, weaker conditions suffice for the existence of a ground state. For example, discontinuities in  the form factor $\phi$ can invalidate our results (see Section \ref{sec: quadratic ham}), but the existence of the ground state depends solely on the behavior of $\phi$ near $0$.}, as one can establish by, for example, the techniques in \cite{griesemerliebloss,haslerherbstgroudstates,bachfrohlichsigalspectralanalysis, gerardgroundstate, abdessalamspinboson}.  In Section \ref{sec: ifnrared regularity}, we give a condition on the form factor that is sufficient for Assumption \ref{ass: decay correlation functions} to hold. 

The next assumption is meant to exclude situations in which the atom is poorly coupled to the field. In particular, if the function $\phi(q)$ vanishes identically, then one cannot expect relaxation to the ground state, and this should surely be excluded. 
We assume that $\si(H_\sys)$, the spectrum of $H_\sys$, is nondegenerate, and let $P_e$ for $ e \in \sigma( H_\sys)$ be the corresponding (one-dimensional) spectral projectors. We introduce the nonnegative numbers
\beq
j(e,e') :=     \Tr[P_e DP_{e'}DP_e]  \,  \hat h(e-e'), \label{eqexpression jump rates}
\eeq
where $\hat h(\ve) =  \int_{-\infty}^{\infty} \d t  \,     \e^{\i t \ve} h(t)$   is well-defined  by virtue of Assumption \ref{ass: decay correlation functions}. It can be written in a more intuitive fashion as $2 \pi\int_{\bbR^{d}} \d  q    \,    \delta( \str q\str-(e-e'))    \left\str \phi(q)  \right\str^2 $, which one recognizes as the textbook Fermi Golden Rule expression for a scattering rate.  One deduces that 
 $j(e,e')  =0$ whenever $e'\geq e$. Physically, this expresses that the field is in the vacuum state and it can only absorb (and not emit) energy.

\begin{assumption}[Fermi Golden Rule]\label{ass: Fermi Golden Rule}
We assume that  the spectrum of $H_\sys$ is non-degenerate (all eigenvalues are simple) and we let $e_0 := \min \sigma( H_\sys)$ (atomic ground state energy). Most importantly, we assume that for any  eigenvalue $e \in \sigma( H_\sys), e \neq e_0$, there is a sequence  $e(i), i=1, \ldots, n$ of eigenvalues such that 
\beq
e= e(1) >  e(2) >\ldots >  e(n)= e_0,  \qquad  \textrm{and} \qquad \forall i =1,\ldots, n-1:     j(e(i), e(i+1)) >0 
\eeq
with  $j(\cdot,\cdot)$ as defined above.
\end{assumption}
The numbers  $j(e,e')$ should be viewed as 'jump rates': We define the one-dimensional spectral projector  $P_{\Om}= \str \Om\rangle \langle \Om \str \in \scrB_1(\scrH_\res)$ with $\Om$ the vacuum vector in the Fock space $\scrH_\res$.
 If the joint atom-field system  is described by the density matrix $\rho_0 = P_e  \otimes P_{\Om}$ at time $t=0$, then formal perturbation theory (Fermi Golden Rule) suggests that the probability to find the atom in state $e' \neq e$ at a later time $t>0$, is 
 \beq  \Tr [P_{e'} \e^{-\i t H}\rho_0\e^{\i t H}] =  j(e,e') (\la^2t) + \caO((\la^2t)^2)   \label{eq: fermi golden rule rates} \eeq
The rigorous version of this formula, given in Proposition \ref{prop: weak coupling}, is a crucial ingredient of our analysis. 

\subsection{Initial states and observables} \label{secinitial states}

We now define the class of initial states $\rho_0$ and observables $O$ that we consider. 
For $\psi \in \frh$, let ${\caW}(\psi) \in \scrB(\scrH_\res)$ be the Weyl operator
\beq
{\caW}(\psi) =   \e^{\i \Phi(\psi)}, \qquad  \Phi(\psi) :=   \int d q   \left( \psi(q) a^*_q +\overline{\psi(q)} a_q  \right)
\eeq
We use $\realinitial, \realfinal$ as labels to denote objects characterizing the initial state (`left boundary') and observable (`right boundary').
We pick  $\psi_{\realinitial}, \psi_{\realfinal} \in  \frh$, an atom observable $O_\sys \in \scrB(\scrH_\sys)$ and a density matrix $ \rho_{\sys,0} \in \scrB_1(\scrH_\sys)$, i.e.\ such that $\Tr\rho_{\sys,0}=1$ and $\rho_{\sys,0} \geq 0$. Then we put
\beq
O := O_\sys \otimes {\caW}(\psi_{\realfinal}), \qquad   \rho_{0} := \rho_{\sys,0} \otimes   {\caW}(\psi_{\realinitial}) P_{\Om}\caW^{*}(\psi_{\realinitial}).  \label{eq: choice o and rho}
\eeq
If $\rho_{\sys,0}= \str \psi_{\sys,0} \rangle \langle \psi_{\sys,0} \str $ for some $\psi_{\sys,0} \in \scrH_\sys$, then $\rho_0= \str \Psi_{0} \rangle \langle \Psi_{0} \str $ with $\Psi_0= \psi_{\sys,0} \otimes {\caW}(\psi_{\realinitial})\Omega $ and this is assumed for simplicity in the next section (for notational convenience, we use the general case in later sections). 
We need to assume some regularity properties on $\psi_{\realinitial}, \psi_{\realfinal}$: 
\begin{assumption}[Regularity of initial states and observables]  \label{ass: initial states and observables}
\baq
 \int_{\bbR_{+}} \d t  \,   (1+\str t \str )^{\al} \str     h_{\ltimes}(t) \str< \infty, \qquad \textrm{with} \,\,    &     h_{\ltimes}(t) :=    \langle\phi,  \e^{-\i \str q \str t} \psi_{\realinitial} \rangle_\frh    \label{ass: regularity initial} \\[2mm]
 \int_{\bbR_{+}} \d t  \,   (1+\str t \str )^{\al} \str     h_{\rtimes}(t) \str< \infty, \qquad \textrm{with} \,\,    &    h_{\rtimes}(t):=  \langle\phi,  \e^{\i \str q \str t} \psi_{\realfinal} \rangle_\frh  \label{ass: regularity final}   \\[2mm]
  \sup_{t \geq 0} \str h_{\Join}(t)\str    (1+t)^{\al} < \infty,  \qquad \textrm{with} \,\,     &     h_{\Join}(t)  :=   \langle\psi_{\realinitial},   \e^{\i \str q \str t}  \psi_{\realfinal} \rangle_\frh  \label{ass: regularity join}
\eaq
for $\al>0$. 
\end{assumption}
\setcounter{theorem}{0}
Throughout our paper, we always assume that Assumptions \ref{ass: decay correlation functions} and \ref{ass: initial states and observables} are satisfied with the same parameter $\al>0$. This is done for the sake of simplicity, though it slightly weakens the result (see Remark \ref{rem: alphas}  below).

\subsection{Results}

Write $\Psi_t := \e^{-\i t H}\Psi_0$ for some initial vector $\Psi_0 \in \scrH$ and $H$ as defined in \eqref{def: total hamiltonian}.  We define the Weyl algebra $\frW_{ \al}, \al>0$ to be the $C^{*}$-algebra generated by 'atomic' observables $A \otimes \lone$ and Weyl-operators $ \lone \otimes
{\caW}(\psi_{\realfinal}) $ with $\psi_{\realfinal} \in \frh$ satisfying \eqref{ass: regularity final}.

\bet \label{thm: steady state}
Assume that Assumption \ref{ass:  decay correlation functions} and Assumption \ref{ass: Fermi Golden Rule}  are satisfied.   Then, there is a $\la_0>0$ such that, for any coupling strength $ 0 < \str\la \str \leq \la_0$, the following holds true:
\ben

\item  There is a bounded linear functional $O \mapsto \langle O \rangle_{\infty}$ on $\frW_{ \al}$ such that 
\beq
\lim_{t\to \infty}  \langle \Psi_t, O \Psi_t \rangle   = \langle O \rangle_{\infty}  \label{eq: approach to ground state statement}
\eeq 
for any initial vector $\Psi_0 \in \scrH$ with $ \norm \Psi_0 \norm=1$ and $O \in \frW_{\al}$. 
\item   Let $\Psi^0_{\mathrm{gs}}= \psi_{e_{0}} \otimes \Om_{},  \norm \Psi^0_{\mathrm{gs}} \norm=1$ be the normalized ground state  of the uncoupled Hamiltonian $H_\sys+H_\res$ and let $O $ be of the form  \eqref{eq: choice o and rho} with $\psi_\realfinal$ satisfying \eqref{ass: regularity final},  then
\beq   \label{eq: steady state is perturbative}
 \langle O \rangle_{\infty} -        \langle \Psi^0_{\mathrm{gs}}, O  \Psi^0_{\mathrm{gs}} \rangle   = \caO(\str\la\str^{ \min{ (2\al, 1})}), \qquad    \la \to 0
\eeq
\item   If $O$ and $\Psi_0$ (that is: the $\psi_{\realinitial}, \psi_{\realfinal}$ that determine them) satisfy the three bounds of Assumption \ref{ass: initial states and observables}, then 
\beq
 \left\str  \langle O \rangle_{\infty} -  \langle \Psi_t, O \Psi_t \rangle \right\str \leq  \caO(t^{-\al}), \qquad  t \to \infty
 \eeq
\een
\eet
\begin{remark}
One could be tempted to interpret the functional $\langle \cdot\rangle_{\infty}$ as the expectation in the ground state of the coupled system, but this is not correct since, for $\al <1$, the coupled system does in general not have a ground state in $\scrH$ (although a ground state does exist in the Hilbert space corresponding to a different representation of the operator algebra).   On the other hand, if the system does admit a ground state $\Psi_{\mathrm{gs}} \in \scrH$, then by choosing $\Psi_0 = \Psi_{\mathrm{gs}}$,  \eqref{eq: approach to ground state statement} immediately implies that $\langle  O  \rangle_{\infty} = \langle \Psi_{\mathrm{gs}} , O \Psi_{\mathrm{gs}} \rangle    $. 
\end{remark}

Our second result bounds the number of emitted bosons. Let $N$ be number operator on the Fock space $\scrH_\res$. 

\bet  \label{thm: photon bound}
Assume that  Assumption \ref{ass:  decay correlation functions} for some $\al>0$ and Assumption \ref{ass: Fermi Golden Rule}  are satisfied. 
Then, there are $\la'_0, \ka'_0>0$ such that, for any coupling strength $\la$ with $ 0 < \str\la \str \leq \la'_0$, complex number $\ka$ with $\str \ka \str \leq \ka'_0$, and initial vector $\Psi_0$ with $\psi_{\realinitial}$ satisfying the bound  \eqref{ass: regularity initial} in Assumption \ref{ass: initial states and observables}, we have
\beq
   \left    \str \langle \Psi_t,  (\lone \otimes \e^{\ka N}) \Psi_t \rangle \right \str  \leq   \breve C \exp{\left( C \str t\str^{(1-\min(\al,1))}  \right)},\qquad  t \geq 0  \label{eq: bound number operator}
\eeq
where the constant $ \breve C$ depends on $\Psi_0$, but $C$ does not (and none of them depends on $\la,\ka$ or $t$).   In particular, if  $\al \geq 1$, then the LHS is bounded uniformly in time. 
  \eet

\begin{remark} \label{rem: alphas}
As indicated below Assumption \ref{ass: initial states and observables}, we prefer to keep one constant $\al$ throughout the paper. 
Let us describe the possible improvement of Theorem \ref{thm: photon bound}  if we were to drop this constraint. 
Assume again that Assumption \ref{ass:  decay correlation functions} holds for some $\al>0$ and assume the  the bound  \eqref{ass: regularity initial} in Assumption \ref{ass: initial states and observables} is satisfied for some $\al_{\realinitial} >0$, then the photon number bound \eqref{eq: bound number operator} still holds with $\al$ determined by Assumption \ref{ass:  decay correlation functions}, regardless of the value of $\al_{\realinitial}$. This is clear from the inspection of the last part of the proof of Theorem \ref{thm: photon bound} in Section \ref{sec: finite size}.
\end{remark}

\subsection{Discussion of the results}

\subsubsection{Quadratic Hamiltonians} \label{sec: quadratic ham}
The easiest way to understand our results and the different conditions involved, is to compare them to an integrable model where the same questions can be asked.  Consider the formal Hamiltonian
\beq  \label{def: vanhovehamiltonian}
H =  \int \d q  \str q \str a^*_q a_q  +   \int \d q \left(\phi(q) a^*_q + \overline{\phi(q)} a_q \right),   \qquad   \phi \in \frh
\eeq
which fits into our framework by taking $\scrH_\sys= \bbC$ (in that case $H_\sys$ is an irrelevant number).
By completing the square (which, in this context, is a special case of  a 'Bogoliubov transformation') we can rewrite it as
\beq
H =  E_{gs} +  \int \d q  {\str q \str} b^*_q b_q , \qquad   b_q :=  a_q + \frac{ \phi(q)}{{\str q \str}}, \qquad  E_{gs}  :=  -  \int \d q \frac{\str \phi(q)\str^2}{{\str q \str}}
\eeq
If   $\phi \in \frh$ and $ \phi /\sqrt{{\str q \str}}   \in  \frh$,  then the term linear in $a/a^*$ is an infinitesimal perturbation of the quadratic term (the first term in \eqref{def: vanhovehamiltonian}). The operator $H$ is self adjoint on the domain of the quadratic term and it is bounded below by $E_{gs} > - \infty$.  If  moreover  $\phi/ {\str q \str}  \in  \frh$,
then $H$ has a normalizable ground state, given by 
\beq
\Psi_{gs} :=  \e^{- \norm \phi  /{\str q \str}  \norm^2} \caW( \phi  /{\str q \str} )   \Om,
\eeq 
with the Weyl operator $\caW( \psi ), \psi \in  \frh $ as defined in Section \ref{secinitial states}.  We refer to  e.g.\ \cite{derezinskivanhove} for an extended and rigorous discussion of quadratic Hamiltonians.

Let us look into ergodic properties of the evolution  generated by a quadratic Hamiltonian. Let $\psi_{\realfinal} \in L^2$, and consider the observable $O= {\caW}(\psi_{\realfinal})$. Then by explicit calculation
\beq
\langle \Om,   \e^{\i t H}  O  \e^{-\i t H} \Om \rangle =    \e^{- \norm \psi_{\realfinal} \norm^2}   \e^{2\i  \Re \langle \psi_{\realfinal}, (\e^{\i t {\str q \str}}-1) \frac{ \phi}{{\str q \str}}  \rangle } \label{eqref: asymptotic value vanhove}
\eeq
Recall the correlation function $h_{\rtimes}(t)=\langle \psi_{\realfinal}, \e^{\i t {\str q \str}} \phi \rangle$ from Assumption \ref{ass: initial states and observables}.
Clearly,  if $h_{\rtimes} \in L^1(\bbR,\d t)$, then the RHS of \eqref{eqref: asymptotic value vanhove} converges as $t \to \infty$. In particular, this can be true even when $\phi/{\str q \str} \not \in  \frh$, that is, if $H$ has no ground state. One can easily convince oneself that the $t \to \infty$-asymptotics of  \eqref{eqref: asymptotic value vanhove} does not change if we consider a general initial state $\rho_0 =\str \Psi_0 \rangle \langle \Psi_0 \str $ with $\Psi_0= {\caW}(\psi_{\realinitial})\Om, \norm \Psi_0 \norm=1$ and $\psi_{\realinitial}$ such that $h_{\ltimes}(t) = \langle \phi,\e^{-\i t {\str q \str}}\psi_{\realinitial} \rangle  $ is integrable and $h_{\Join}(t)= \langle \psi_{\realinitial},\e^{\i t {\str q \str}} \psi_{\realfinal} \rangle$ vanishes at infinity. 
Next, we study the number of emitted photons 
\baq
\langle \Om,   \e^{\i t H}  \e^{\ka N}  \e^{-\i t H} \Om \rangle  &= &     \langle \Om,     \exp{\left(\ka  \int \d k(a_k+ \varphi_t(k))^* (a_k+ \varphi_t(k))  \right) }  \Om \rangle, \qquad      \varphi_t :=  (1-\e^{\i t {\str q \str}}) \frac{ \phi}{{\str q \str}} \\[2mm]
&= &   \exp{ \left(\norm \varphi_t \norm^2  (\e^{\ka}-1) \right)   }   =  \exp{  \left( (\e^{\ka}-1)   \int_0^t \d s  \int_0^t   \d s'  h(s-s')  \right) }  
\eaq
It is clear that this expression remains bounded if (and only if) $\phi/{\str q \str} \in L^2$, hence if $H$ has a ground state.  Moreover, we see that the rate of growth of the LHS of \eqref{eq: bound number operator} in Theorem \ref{thm: photon bound} corresponds roughly to estimating $h(s-s')$ by $\str h(s-s') \str$ in the integral. 

%\subsection{Related work}\label{sec: related work}
%

\subsection{Plan of the proof} \label{sec: plan of the proof}

Our two  results, relaxation to the ground state and the photon bound, are very similar from the technical point of view, even if their physical meaning possibly is not. For this reason, we focus exclusively on the relaxation to the ground state in the present section, and we devote a few words to the photon bound at the end.

The proof relies on the following philosophy. The original problem is formulated as a perturbation (with small parameter $\la$) of an integrable Hamiltonian $H_\sys + H_\res$ whose dynamics does not have the phenomenon that we want to exhibit: it does not relax into the ground state.  This can already be seen by remarking that the atom $\sys$ is not coupled to the field $\res$ and, as the former is finite-dimensional, its dynamics is oscillatory.  However, the Fermi Golden Rule \eqref{eq: fermi golden rule rates} provides us with a picture that does capture the dissipative behavior: The fact that the state of the atom changes by jumps between eigenstates of the Hamiltonian $H_\sys$ suggests the following approximation.  
\beq
\rho_t  \approx    \e^{-\i t \adjoint(H_\sys) + \la^2 t M }  \rho_{\sys,0} \otimes P_{\Omega}  \label{eq: pictorial markov approx}
\eeq
where $M$ is the generator of a dissipative dynamics that we can loosely describe as the Markov jump process on eigenstates of  $H_\sys$ with jump rates given by  \eqref{eq: fermi golden rule rates} and exponential decay of the off-diagonal (in $H_\sys$-basis) part of the density matrix.     The approximation becomes exact (at least as far as the $\sys$-state  is concerned), as $\la \to 0, t \to \infty$ such that  $t\la^{2}$ is held fixed.  This was already advocated by Van Hove \cite{vanhove} and it was made precise by Davies \cite{davies1}.   We state it explicitly in Proposition \ref{prop: weak coupling} and we review the proof in Appendix \ref{secweak coupling limit}.

The underlying physical reason why \eqref{eq: pictorial markov approx} is a good approximation is
% that  ansatz that the state of the field $\res$ is assumed to be $P_{\Omega}$ for all times is based on  the fact
that $P_{\Omega}$ is invariant under the free $\res$-dynamics and any disturbance in the field caused by $\sys$ is carried away (disperses) to spatial infinity quickly, such that it is irrelevant for the further evolution, and one can pretend that the state of the field remains $P_{\Omega}$.  The dispersive property is a consequence of the temporal decay of field correlations for the uncoupled dynamics, which  is our Assumption \ref{ass: decay correlation functions}.  It  is therefore plausible that the evolution of $\sys$ is Markovian on time scales longer than the time necessary for a field excitation to disperse away 

The approximate dynamics in \eqref{eq: pictorial markov approx} exhibits relaxation to the projection onto the (uncoupled) ground state  $\psi_{e_0} \otimes \Om$ provided that sufficiently many jump rates are nonzero; this is captured by the Fermi Golden Rule Assumption \ref{ass: Fermi Golden Rule}.

In a nutshell, our strategy is to use the dynamics \eqref{eq: pictorial markov approx} as a zero-order term of our expansion for the full dynamics. Note that our expansion is not simply in powers of the coupling constant $\la$; the exponent in \eqref{eq: pictorial markov approx} clearly has zeroth and second order contributions in $\la$. The reason we refer to \eqref{eq: pictorial markov approx} as zeroth order is that all other contributions to the dynamics are small compared to this term, or rather, to its dissipative effect.

 Because of the jumps described by $M$, the dynamics \eqref{eq: pictorial markov approx} is stochastic and  hence our task reduces  to controlling a small perturbation (the real dynamics at finite but small $\la$) of a stochastic evolution. This is quite a tractable problem that can be handled by analytic perturbation theory of isolated eigenvalues and a cluster expansion.  Similar expansions were developed e.g.\ in  \cite{maesnetocnyspacetime, bricmontkupiainencoupledmaps} and, very closely to the setup of the present paper, in \cite{deroeck, deroeckkupiainen}.    

%We proceed as follows:  We define the reduced evolution $T$ acting on $\scrB_1(\scrH_\sys)$ by 
%$$  T_\rho_\sys =  \e^{\i t } $$

The main result of this expansion is that we manage to represent $  \Tr \rho_{t} O$ as a one-dimensional polymer gas (this dimension corresponds to time). Then the problem of showing that there is a well-defined and unique asymptotic state is analogous to the problem of proving decay of correlations in the one-dimensional gas. 

More precisely, the polymer representation is 
 \beq
\Tr (\rho_{t} O) = k_{\realinitial} k_{\realfinal} \sum_{\caA} \prod_{A \in \caA} v(A), \qquad  \textrm{for} \,   t= n/\la^2  \label{eq: first polymer rep}
 \eeq
where  the sum is over collections $\caA$ of sets $A \subset \{0,1, \ldots, n, n+1 \}$ such that for any two sets $A_1, A_2$, $\distance(A_1,A_2) >1$ where we write $\distance(A_1,A_2)= \min_{i \in A_1, j\in A_2} \str i-j\str$.   The sets $A$ are called polymers, the numbers $v(A)$ are polymer weights.  Moreover, only the $v(A)$ 
 with $0  \in A$, $(n+1) \in A$ depend on the initial state, resp.\ the observable.  To a good approximation (the error made is not important for the present discussion), 
 \beq 
 k_{\realinitial} \sim  \Tr \rho_0    , \qquad k_{\realfinal} \sim \Tr  [(P_{e_0} \otimes P_{\Om}) O]
 \eeq
 Our goal is to prove that 
\beq
\lim_{t \to \infty} \Tr (\rho_{t} O) = \Tr (\rho_0)  \langle O \rangle_{\infty}  
\eeq
where $ \langle O \rangle_{\infty}$ does not  depend on $\rho_0$.
Setting all $v(A)=0$ essentially amounts to pretending that $\rho_t = P_{e_0} \otimes P_{\Om}$ for all $t$. In that case $\Tr (\rho_{t} O)= k_{\realinitial} k_{\realfinal}$ and the expectation value of the observable is independent of the initial state, apart from the trivial normalization factor $k_{\realinitial}$.  The polymers $A$ contain corrections to this picture.  These corrections originate from the fluctuations of the  Markovian dynamics generated by $M$ and  from the corrections to the Markovian behavior \eqref{eq: pictorial markov approx}.  Pictorially, let $A = \{ \tau_1, \tau_2, \ldots \tau_m\}$
 with $\tau_{i } \leq \tau_{i+1}$, then   $v(A)$ describes  correlated deviations from the assumption that $\rho_t =  P_{e_0} \otimes P_{\Om}$  in the time-interval   $\Dom(\tau_{1}), \ldots, \Dom(\tau_{m-1}) $, where $\Dom(\tau_m)\equiv  (1/ \la^2)[\tau_m-1,\tau_m]$.  
 The detailed construction of the polymer representation  \eqref{eq: first polymer rep} is carried through in Section \ref{sec: polymer rep}. 

To prove decay of correlations, we use a standard cluster expansion, which we review in Appendix \ref{appsec: cluster expansions}.  A possible  condition for the applicability of the cluster expansion method is the "Kotecky-Preiss" criterion. Applied to our model, it demands that,
\beq
\sum_{A: \distance(A,A') \leq 1}  v(A) \e^{a(A)} \leq    a(A'),  \label{eq: kp intro}
\eeq
 with $A,A' \subset \{1, \ldots,n \}$ and  $a(\cdot)$ an $n$-independent positive function on polymers. 
 We formulate \eqref{eq: kp intro} (together with some other statements) in Lemma \ref{lem: bound on scalar polymers}.   In our case $a(A) \equiv  \ep C \str A \str$ where $\ep =  \str \la \str^{2\min{(\al,1)}}$ can be seen as a renormalized coupling constant.

To prove that the  Kotecky-Preiss criterion \eqref{eq: kp intro} is satisfied, we use a  Dyson (or Duhamel) expansion, it relies on the smallness of the coupling constant $\la$ and the decay of field correlations. This is intuitive;  if the  correlations of the free field decay roughly as  $t^{-(1+\al)}$ (cfr.\ Assumption \ref{ass: decay correlation functions}) then one could conjecture, for example, that  for $A=\{ \tau_1, \tau_2\}$
also the weight $v(A)$  decays as $ (\distance(\Dom (\tau_1),  \Dom (\tau_2)   ) )^{-(1+\al)} $ as $\tau_2 -\tau_1 \to \infty$.  This picture turns out to be essentially correct. In fact, we get $\str v(\{ \tau_1, \tau_2\}) \str \sim  C \ep ( \tau_2-\tau_1)^{-(1+\al)}$.  
This analysis (proof of the Kotecky-Preiss criterion) is accomplished in Section \ref{secexpansions}. 

In the concluding Section \ref{sec: analysis of polymer models}, we prove our results. The main point of this section is to pinpoint how the cluster expansion gives rise to decay of correlations. Such reasoning is completely standard in high-temperature expansions of statistical physics, see e.g.\ \cite{simon}.

To prove the photon bound, we use an analogous approach but this time we develop a polymer representation like \eqref{eq: first polymer rep} for the quantity $\Tr (\rho_t \e^{\ka N})$.  In fact, in our proofs,  we provide one general polymer representation for $\Tr(O \e^{(\ka/2) N}\rho_t  \e^{(\ka/2) N} )$ and then we set $\ka=0$ to  study $\Tr(\rho_t O)$ and  $O= \lone$ to study $\Tr (\rho_t \e^{\ka N})$.

In the rest of the paper, we assume that Assumptions \ref{ass: decay correlation functions}, \ref{ass: Fermi Golden Rule} and \ref{ass: initial states and observables} are satisfied for $\al>0$ and all the Lemmata will depend on this parameter $\al$.

\section*{Acknowledgements}
 W.D.R.\  is grateful to the people who explained him  scattering in quantum field theory. 
In particular, he profited a lot from discussions and collaboration with  J.\ Derezi{\'n}ski over the past years, and from exchanges with W.\ Dybalksi, J.\ Fr{\"o}hlich,  M.\ Griesemer and B.\ Schlein. 
More generally, we would like to thank J.\ Fr{\"o}hlich and I.\ M.\ Sigal for bringing these problems to the attention of the math-phys community, starting more than a decade ago. 

We  thank the European Research Council and the Academy of Finland for financial support.

\section{Polymer Representation} \label{sec: polymer rep}

In this section, we complete the first important step of our proof, namely we rewrite all quantities of interest through a polymer representation.  First, let us discretize time by introducing a mesoscopic time scale $\la^{-2}$, where $\la>0$ is the coupling strength.  That is, we consider times  of the form $t= n/\la^2 $ with $n \in \bbN$ (the discretization will be easily removed at the end of the argument).  The main quantity that we study is 
\beq
Z_n  = Z_n(O, \rho_0, \ka): =  \Tr\left[ O  \e^{ (\ka/2)N}    \e^{-\i (n/\la^2) H }    \rho_0     \e^{\i (n/\la^2) H}   \e^{(\ka/2) N} \, \right]  
    \label{eq: def z}
\eeq
where   $N$ is the number operator  on the Fock space $\scrH_\res$ and  $\ka \in \bbC$ is a sufficiently small  parameter.  The operators $\rho_0$ and $O$ are the initial states, respectively observable constructed as in Section \ref{secinitial states}.
%At the end of our analysis, we will allow for different choices of the mesoscopic time scale, namely $\ell \la^{-2}$ instead of $\la^{-2}$ with $1/2 \leq\ell\leq1$, and this will allow us to recover full information for all times, not only integer multiples of $\la^{-2}$. 

As announced, we develop a polymer representation for $Z_n$. In fact, we will first construct a polymer representation with \emph{operator-valued polymer weights}. In this representation, the polymers correspond to deviations from Markovian behavior.  This is described in Section \ref{sec: operator valued polymer model}. Bounds on the operator-valued polymer weights are stated in \ref{sec: bounds on operator valued polymers} but their proof is deferred to Section \ref{secdiscretization}.  

Then we define the 'true' (scalar) polymer representation for $Z_n$ by adding the excitations of the Markovian approximation to the already defined operator-valued polymers. This is done in Section \ref{sec: scalar polymer model}.   Why this leads to scalar polymers will be explained there.  Finally, we need to provide bounds on the scalar polymers in order to satisfy the Kotecky-Preiss criterion. These bounds follow  from the bounds on the operator valued polymer weights and this is described in Section \ref{sec: bounds on scalar polymers}.   In all of the following sections, we prefer to treat 'bulk' and 'boundary'- polymers separately, where the terms 'bulk' and 'boundary' refer to the time-dimension.   The analysis of the bulk polymers is the more subtle piece of work but the treatment of the boundary polymers demands additional notation.  Therefore we treat the former first, and then indicate the (in all cases quite minimal) changes necessary for the latter.

\subsection{Operator valued polymer model} \label{sec: operator valued polymer model}

\subsubsection{Definition of the deformed dynamics} \label{sec: definition of the deformed dynamics}

Starting from \eqref{eq: def z}, we would like to move the deformation parameter $\ka$ into the dynamics $\e^{-\i t H}$. Given an operator  $X$, we introduce the left, right and two-sided multiplication, acting on operators $\rho$
 \beq
{\caL}(X) \rho:= X\rho, \qquad {\caR}(X) \rho:=  \rho X, \qquad \caM (X)\rho :=   X \rho X.
 \eeq
Defining the Liouvillian $ L=  \adjoint (H)={\caL}(H)-{\caR}(H)$
with $H$ as in \eqref{def: total hamiltonian}, we may rewrite \eqref{eq: def z}
as
\beq
Z_n  =  \Tr\left[ O  \caM(  \e^{(\ka/2) N} )   \e^{-\i (n/\la^2) L }    \rho_0     \, \right]  .
    \label{eq: def z1}
\eeq
We  formally define the \emph{deformed} (not self-adjoint) Hamiltonian 
 \beq
 H_{\ka} =  \e^{(\ka/2) N} H    \e^{-(\ka/2) N}
 \eeq
and  the \emph{deformed} Liouvillian
 \beq
L_{\ka} := {\caL}(H_{\ka}) -{\caR}(H_{-\ka})  =    \caM(   \e^{(\ka/2) N}) \,   L \,  \caM(   \e^{-(\ka/2) N}).
\eeq
Then, formally
\beq\label{kappaL}
\caM(   \e^{(\ka/2) N})  \e^{-\i t L}   \caM (   \e^{-(\ka/2) N}) =    \e^{-\i t L_\ka}
\eeq
which would allow to rewrite
\beq
Z_n  =  \Tr\left[ O   \e^{-\i (n/\la^2) L_\ka }   \caM(  \e^{(\ka/2) N} )   \rho_0     \, \right]  .
    \label{eq: def z2}
\eeq
To make these manipulations precise note first that
relying on the relative boundedness of $\e^{(\ka/2) N} H_\inter \e^{-(\ka/2) N}$ w.r.t.\ $H_\sys+H_\res$, which follows from \eqref{ass: semibounded hamiltonian},  it is easy to construct  $H_\ka$ and $L_\ka$ as unbounded closed operators and to show that they form  an analytic family (in $\ka$) of class A in the sense of Kato \cite{katoperturbation}. Then \eqref{eq: def z2}  is a simple consequence of the functional calculus in case $\ka \in \i \bbR$. 
The upcoming Lemma \ref{lem: fake functional calculus} gives a constructive meaning to it for
arbitrary $\ka$ and  provides an expansion that we will use in practice. 
To state this lemma, we need a few additional definitions. 
%
%Since $N$ commutes with both $H_\sys$ and $H_\res$, we have 
%\beq
%H_{\ka}= H_\sys+ H_\res+ H_{\inter,\ka}, \qquad \text{with} \quad
% H_{\inter,\ka} =   \e^{(\ka/2)N}H_{\inter}  \e^{-(\ka/2)N} \eeq  and we also define
%%\beq
%%H_{\inter}^{\ka}(s) : =     \e^{\i s (H_\sys + H_\res)} H_{\inter}^{\ka}   \e^{-\i s (H_\sys + H_\res)}
%%\eeq
%\baq
%H_{\inter,\ka}(s) & := &   \e^{\i s (H_\sys + H_\res)} H_{\inter, \ka}   \e^{-\i s (H_\sys + H_\res)}  \\
%& =  &  D(s) \otimes    \int d q  \left(\e^{\i t \str q \str + \ka/2}  a^*_q  +  \e^{\i t \str q \str - \ka/2}  a_q \right), \qquad D(s):= \e^{\i s H_\sys} D  \e^{-\i s H_\sys} 
%\eaq
First,  some additional Liouvillians: 
\beq
L_\sys := \adjoint(H_\sys),  \qquad  L_\res := \adjoint(H_\res),  \qquad    L_\inter  :=  \lambda   \adjoint(H_\inter)
\eeq
and 
\beq
  L_{\inter,\ka}(s) : =       \caM(   \e^{(\ka/2) N})   \e^{\i s (L_\sys+ L_\res)}    L_{\inter}  \e^{-\i s (L_\sys+ L_\res)}    \caM(   \e^{-(\ka/2) N}) . \label{LIdef} 
%&=&    \caL \left( \e^{\i s L_\sys} (D)  \otimes    \int d q    \left(\e^{\i t \str q \str + \ka/2} \phi(q)  a^*_q  +  \e^{\i t \str q \str - \ka/2}  \bar \phi(q)  a_q \right) \right)  \nonumber    \\ 
%&   &- &   \caR \left( \e^{\i s L_\sys} (D)  \otimes    \int d q  \left(\e^{\i t \str q \str - \ka/2}   \phi(q) a^*_q  +  \e^{\i t %\str q \str + \ka/2} \bar  \phi(q)   a_q \right) \right).  
\eeq
Defining
\baq  
 \Phi_{\ka}(\psi,s)&:=&    \e^{(\ka/2)N}  \e^{\i s (H_\sys+ H_\res)}  \Phi(\psi)  \e^{-\i s (H_\sys+ H_\res)}  \e^{-(\ka/2)N} \nonumber\\
& =& \int \d q \,     (\e^{\i s \str q \str + \ka/2} \psi(q)  a^*_q  +  \e^{-\i s \str q \str - \ka/2}  \overline{ \psi(q)}  a_q )
\label{Phikappa}
\eaq
we have
\baq
  L_{\inter,\ka}(s) =      \la  \caL \left(  \e^{\i s L_\sys} (D)  \otimes \Phi_{\ka}(\phi,s) \right)
  -  \caR \left( \e^{\i s L_\sys} (D)  \otimes\Phi_{-\ka}(\phi,s) \right).  
\eaq
For $\Psi \in \scrH$, we decompose $\Psi= \sum_{m \in \bbN} \Psi_m$ where $\Psi_m=   1_{[N=m]}  \Psi$, i.e.\ $\Psi_m \in \scrH_\sys \otimes  \scrH_{\res,m}$ with $\scrH_{\res,m} \subset \scrH_{\res} $ the $m$-boson sector, see e.g.\ \cite{derezinski1} for more details. Define the dense subspace $\scrD_1(\scrH)  \subset \scrB_1(\scrH)$ to be the space consisting of finite linear combinations  of rank-one operators $\str \Psi \rangle \langle \Psi' \str$ satisfying 
\beq
   \forall m \in \bbN:   \max{(\norm \Psi_{m}  \norm, \norm \Psi'_{m}  \norm)} \leq  \frac{C^{m}}{\sqrt{m!}   },   \qquad  \text{for some}\,\,    C>0
   \eeq

\begin{lemma}\label{lem: fake functional calculus} The LHS of equation  \eqref{kappaL} defines an operator 
$$
 \e^{-\i t L_\ka}:\scrD_1(\scrH)\to\scrD_1(\scrH) $$ 
 for all 
 $\ka \in \bbC$ and $t\in \bbR$. Given $\rho\in\scrD_1(\scrH)$ the map $\ka \mapsto \e^{-\i t L_\ka}    \rho$ is holomorphic
 from  $\bbC $ to $\scrB_1(\scrH_\sys)$. Moreover on this domain
\beq  \label{eq: group property}
 \e^{-\i (t_1+t_2) L_\ka}= \e^{-\i t_1 L_\ka}  \e^{-\i t_2 L_\ka}
  \eeq
and
\beq
 \e^{-\i t L_\ka}    \rho  =   \e^{-\i t (L_\sys+L_\res)}   \sum _{m \in \bbN} (-\i )^m  \mathop{\int}\limits_{0< t_1 < \ldots < t_m< t} \d t_1 \ldots \d t_m \,      L_{\inter, \ka} (t_m)  \ldots  L_{\inter, \ka} (t_2) L_{\inter, \ka} (t_1)   \rho   \label{eq: duhamel on superspace}
  \eeq
  where the $m=0$-term on the RHS is understood to equal $\rho$ and  sums and integrals  converge absolutely. Finally, the RHS of eq. \eqref{eq: def z2} is well defined and  \eqref{eq: def z2} holds.

\end{lemma}
\begin{proof}
Starting from $L_\ka =  L_\sys+L_\res+ \la  L_{\inter,\ka}$, and iterating the Duhamel formula
\beq
   \e^{-\i   t  L_{\ka} } \rho =     \e^{-\i   t  (L_\res+L_\sys) } \rho -\i   \int_0^t \d s \,   \e^{-\i  (t-s)  L_\ka }   L_{\inter,\ka}   \e^{-\i   s  (L_\res+L_\sys) }   \rho
\eeq
we formally arrive at \eqref{eq: duhamel on superspace}.   Hence, the only nontrivial claim in the lemma is the absolute convergence of the series on the RHS of \eqref{eq: duhamel on superspace} and the fact that it belongs to $\scrD_1(\scrH)$. We refer to \cite{derezinskideroeck2} for an explicit proof, which relies exclusively on the well-known estimate 
\beq
\left\norm \int \d q \psi(q)  a^{\#}_q   \Psi_m \right\norm_{\scrH}  \leq    \sqrt{m+1}  \,    \norm \Psi_m \norm_{\scrH}  \int \d q \str\psi(q)\str^2, \qquad \text{for}\,  a^{\#}=a,a^*
\eeq 
\end{proof}

Hence, in what follows, we freely use the operators $\e^{-\i t L_\ka}$ and the group property \eqref{eq: group property}. 

\subsubsection{Splitting of the dynamics} \label{sec: splitting of the dynamics}

We define  the reduced dynamics $Q_t: \scrB_1(\scrH_\sys) \to \scrB_1(\scrH_\sys)$ of the atom;
 \beq
 Q_{t} \rho_{\sys}  :=  \Tr_\res [  \e^{- \i t  L_{\ka}} (\rho_{\sys} \otimes P_{\Om}) ]   \label{eq: first encounter reduced q}
 \eeq
 where $\Tr_{\res}: \caB_1(\scrH) \to  \caB_1(\scrH_\sys) $ is the partial trace and the well-definedness of the RHS follows from Lemma \ref{lem: fake functional calculus}.   A large part of our analysis serves to prove that $ Q_{t} $ tends to a one-dimensional projection as $t \to \infty$. 

We start by rewriting $Z_n$ in \eqref{eq: def z2}.  
%
%For an operator $A \in \scrB(\scrH)$, we define  $   \mathrm{Ad}( A)  \in \scrB(\scrB_1(\scrH)) $ by setting
%\beq
%    \mathrm{Ad}( A)   \rho :=    A  \rho  A^*
%\eeq
Recall $L_\res= \adjoint(H_\res)$  and introduce  operators $U_{\tau}$ with $ \tau \in \bbN$.
     \baq
 U_{\tau}  &=&       \e^{\i (\tau/\la^2)  L_\res}   \e^{-\i  (1/\la^2)   L_{\ka}}    \e^{-\i  ((\tau-1)/\la^2)   L_\res} 
 \eaq
 The motivation for this definition is that the product of $U_\tau$ telescopes into
 \beq   \label{def: Z}
    U_n   \ldots U_1  =      \e^{\i  (n/\la^2)   L_\res}   \e^{-\i  (n/\la^2)  L_{\ka}}.  
 \eeq
In particular, if we choose $\rho_0= \rho_{\sys,0} \otimes P_{\Om}$ and $O = O_\sys \otimes \lone $ then 
eq. \eqref{eq: def z2} becomes
 \beq
 Z_n =  \Tr_{\sys} [ O_\sys   Q_{n/\la^2} \rho_{\sys,0}  ]
%  , \qquad     \Tr_\res [      U_n   \ldots U_1 (\rho_{\sys,0} \otimes P_{\Om}) ] 
    \label{eq: first encounter simple z}
 \eeq
 and hence in this case the study of $Z_n$ reduces to the study of $Q_t$. 
  The main idea of our approach is that, at least qualitatively, the main contribution to $   Q_{n/\la^2} $ can be inferred by approximating $ Q_{n/\la^2} $ by $ (  Q_{1/\la^2})^n$. We rename  $T:= Q_{1/\la^2}$ and we define the 'excitation operators' 
 \beq
 B_\tau =   U_\tau  - T \otimes \lone  
 \eeq
 Our task is to  understand how the behavior of $T^n$ is modified by the  excitation operators $B_{\tau}$. To quantify the influence of the latter, we now develop a formalism.   
 \subsubsection{Correlation functions of excitations}  \label{sec: correlation functions of excitations}
We abbreviate
\beq
\scrR_\sys =  \scrB(\scrB_1(\scrH_S)), \qquad    \scrR_\res =  \scrB(\scrB_1(\scrH_\res))
\eeq
%such that $U_\tau, B_\tau$ are elements of $\scrR_\sys  \otimes  \scrR_\res$ and $T$ is an element of $\scrR_\sys$. 
Define, for $W,W' \in \scrR_\sys  \otimes  \scrR_\res$ the object
$$W\otimes_\sys  W'\in \scrR_\sys  \otimes  \scrR_\sys \otimes \scrR_\res $$
 as an operator product in $\res$-part and tensor
 product in $\sys$-part. Concretely, let $W=W_\sys \otimes W_\res$ and $W'=W'_\sys \otimes W'_\res$.
 Then 
$$W\otimes_\sys W':= W_\sys  \otimes W'_\sys  \otimes W_\res W'_\res.$$
and we extend this by linearity to arbitrary $W, W'$. 
Iterating this construction we define for $W_i \in \scrR_\sys  \otimes  \scrR_\res  $, $i=1,\dots,m$
$$ W_m\otimes_\sys \dots\ldots \otimes_\sys W_2 \otimes_\sys W_1 \in (\scrR_\caS)^{ \otimes^m}\otimes \scrR_\res .$$
(Note that no analysis problems arise since $\scrR_\sys$ is finite-dimensional.)
We define the `expectation' 
$$\bbE:(\scrR_\sys)^{ \otimes^m}\otimes  \, \scrR_\res
\,\,\rightarrow \,\,(\scrR_\sys)^{ \otimes^m}$$ as
$$
\bbE (W) J:={\Tr}_{\res}  [W (J \otimes P_{\Om})], \qquad  J \in  (\scrB_1(\scrH_\sys))^{ \otimes^m}
$$
Obviously, the action of $\bbE$ is extended to unbounded $W$ satisfying $W((\scrB_1(\scrH_\sys))^{ \otimes^m}  \otimes P_{\Om}) \in \scrB_1(\scrH_\sys^{ \otimes^m}  \otimes \scrH_\res)$. In particular, by (an obvious generalization of) Lemma \ref{lem: fake functional calculus}, we can consider $W= U_{\tau_m}\otimes_\sys \ldots \otimes_\sys U_{\tau_1}$ for any $m$-tuple of times $\tau_1, \ldots, \tau_m$.  For example, we
  rewrite the definition of $T$ as $T = \bbE(U_{\tau})$ (this is with $m=1$).  Let $A=\{\tau_1,\tau_2,\dots, \tau_m\}\subset \bbN$ with the convention that
$\tau_i<\tau_{i+1}$ and define the `time-ordered {correlation function}'
\beq
G_A :=    \bbE \left( B_{\tau_m}\otimes_{\sys}  B_{\tau_{n-1}}\otimes_{\sys}\dots\otimes_{\sys} B_{\tau_1} \right)   \in   (\scrR_S)^{ \otimes^m}   \label{eqfirst definition of correlation function}
\eeq
Note that $G_A  =0$ when the set $A$ is a singleton, as follows directly from $B_\tau = U_\tau- \bbE(U_\tau) \Rightarrow \bbE(B_{\tau})=0$, and $G_A = G_{A+\tau}$ because $ \e^{-\i t L_\res} \initialresfinite=\initialresfinite$.

It will be convenient to label the $\scrR_\sys$'s and to drop the subscript $\sys$ (since we will rarely need $\scrR_\res$). Therefore,
let $\scrR_\tau$, $\tau \in \bbN$ be copies of $\scrR_\sys$ and let $A \subset \bbN$ be as above. We define $\scrR_A$ by
$$
\scrR_A:=\scrR_{\tau_m}\otimes \scrR_{\tau_{m-1}}\otimes\dots \otimes \scrR_{\tau_{1}}.
$$
Obviously,  $\scrR_A$ is naturally isomorphic to $\scrR^{\otimes^m}$ by identifying the right-most tensor factor to $\scrR_1$, the next one to $\scrR_2$, etc. We denote this isomorphism from $\scrR^{\otimes^m}$ to $\scrR_A$ by $\bsI_A$ and we will from now on write $G_A$ to denote $\bsI_A[G_A]$ since $G_A$ acting on the unlabeled space $\scrR^{\otimes^m}$ will not be used.  Similarly we will most often abbreviate $\bsI_{\tau}[ V ]$ ($V \in \scrR$) by $ V_\tau$, which will  lead to a slight abuse of notation, see below and in Section \ref{sec: initial state and observable}.
Consider a collection $\caA$ of disjoint sets $A$ , then each of the spaces $\scrR_{A \in \caA}$ is naturally embedded into $\scrR_{\supp \caA}$, where $\supp \caA= \cup_{A \in \caA} A$.  Consequently, given a  collection of operators $K_A \in \scrR_A, A \in \caA$, we can define 
\beq
\otimes_{A \in \caA}   K_A      \in \scrR_{\supp \caA}
\eeq
In particular, we have $\scrR_A = \otimes_{\tau \in A} \scrR_\tau$.    
In the literature on quantum lattice systems, where similar constructions are necessary, one usually identifies $\scrR_\tau$ with the subspace $\ldots \otimes \lone \otimes \lone \otimes \scrR \otimes \lone \otimes \lone \otimes \ldots $ of the infinite tensor product, such that, for example,  $\otimes_{A \in \caA}   K_A$ is simply written as $\prod_{A} K_A$.  However, we chose to have the tensor products explicit in the notation.

%
%It will be convenient to label the $\scrR_\sys$'s and to drop the subscript $\sys$ (since we will rarely need $\scrR_\res$). Therefore, consider the infinite tensor product $\scrR^{\otimes_\bbN}$ with the natural embedding $\scrR_\tau, \tau \in \bbN$ as a subspace of  $\scrR^{\otimes_\bbN}$. 
%
%and let $A \subset \bbN$ be as above. We define $\scrR_A$ by
%$$
%\scrR_A:=\scrR_{\tau_m}\otimes \scrR_{\tau_{m-1}}\otimes\dots \otimes \scrR_{\tau_{1}}.
%$$
%Obviously,  $\scrR_A$ is naturally isomorphic to $\scrR^{\otimes^m}$ by identifying the right-most tensor factor to $\scrR_1$, the next one to $\scrR_2$, etc. We denote this isomorphism from $\scrR^{\otimes^m}$ to $\scrR_A$ by $\bsI_A$ and we will from now on write $G_A$ to denote $\bsI_A[G_A]$ since $G_A$ acting on the unlabeled space $\scrR^{\otimes^m}$ will not be used.  Similarly we will most often abbreviate $\bsI_{\tau}[ V ]$ ( $V \in \scrR$) by $ V_\tau$, which will lead to a slight abuse of notation, see Section \ref{sec: initial state and observable}.
%
%Consider a collection $\caA$ of disjoint sets $A$ and a  collection of operators $K_A \in \scrR_A, A \in \caA$, then the product 
%\beq
%\prod_{A \in \caA}   K_A      \in \scrR_{\supp \caA}
%\eeq
 
We  define the "contraction operator" $\caT: \scrR_A\to \scrR$, by first giving its action on elementary tensors. Consider a family of operators $V_{\tau}  \in \scrR$,  and set 
\beq
\caT \left[  \mathop{\otimes}\limits_{\tau \in A} V_{\tau}     \right] =V_{\tau_m} V_{\tau_{m-1}}\dots V_{\tau_{1}}, \qquad  \textrm{where}  \qquad  \tau_m > \tau_{m-1} > \ldots  > \tau_1  \label{eq: def contraction}
\eeq
and then extend linearly to the whole of $\scrR_A$.  On the LHS, we abbreviated $\bsI_\tau[V_\tau] $ by $V_\tau $.  

By expanding $U_{\tau}=T\otimes 1+B_\tau$ for every $\tau \in I_n:=  \{1, \ldots,n \} $ in the expression for the reduced dynamics \eqref{eq: first encounter reduced q}, we arrive at
\beq
Q_{n/\la^2} =       \sum_{A  \subset I_n  }   \caT   \left[   G_A \mathop{\opprod}\limits_{ \tau \in I_n \setminus A} T_\tau       \right].   \label{eqZ from correlation functions} 
\eeq
Note that, for each operator appearing in the tensor product, we have specified the space $\scrR_A$ on which it acts.  In contrast, 
the order in which we write the factors inside the $\caT[\cdot]$  does not have any meaning. 

\begin{remark}\label{rem: category}
The above construction with tensor product spaces $\scrR_A$ and the contraction operator $\caT$ does of course not depend on the fact that the spaces are indexed by elements of $\bbN$. The only requirement is that the index set is ordered. In particular, in Section \ref{secexpansions}, we will use the same formalism, but now with copies of $\scrR$ indexed by (a finite number of) times $t_i \in \bbR^+$.  
\end{remark}

\subsubsection{Polymer expansion for cumulants}   \label{sec: connected correlation functions}

The cumulants or "connected correlation functions", denoted by $G^c_A$, are defined to be operators in $\scrR_A$ satisfying $G^c_{\{\tau\}} = G_{\{\tau\}}=0$ and
\beq
G_{A'} =   
   \sum_{\scriptsize{\left.\begin{array}{c}   \textrm{partitions}\,  \caA \, \textrm{of}\, A'    \end{array} \right.  }}  \left( \mathop{\opprod}\limits_{A \in \caA}  G^c_A  \right)   \label{def: cumulants}
\eeq
The tensor product in this formula makes sense since $ \scrR_{A'} = \otimes_{A \in \caA} \scrR_A$ whenever $\caA$ is a partition of $A'$ (cfr.\ previous section).
Note that this definition of connected correlation functions reduces to the usual probabilistic definition when all operators that appear are numbers and the tensor product can be replaced by multiplication. 
Just as in the probabilistic case, the relations \eqref{def: cumulants} for all sets $A'$ fix the operators $G^c_A$ uniquely since the formula \eqref{def: cumulants} can be inverted. 

Plugging \eqref{def: cumulants} into \eqref{eqZ from correlation functions}, we get
\beq
Q_{n/\la^2} =       \sum_{\caA \in \frB^0_n }   \caT   \left[  \mathop{\opprod}\limits_{A \in \caA}  G^c_A   \,  \mathop{\opprod}\limits_{ \tau \in I_n \setminus \supp \caA} T_{\tau}    \right]  \label{eqZ from connected correlation functions} 
\eeq
where $\frB^0_n$ is the set of disjoint collections of subsets of $I_n$. 
The formula \eqref{eqZ from connected correlation functions} 
 is the starting point of our analysis.    
Note that in the case where $ \caA$ contains at least one set $A$ that is not a discrete interval (a consecutive set of integers), there is no obvious way to write the RHS of \eqref{eqZ from connected correlation functions} as an operator product. This was the main motivation for introducing the formalism with tensor products and the contraction $\caT$. 
Following a standard terminology in statistical mechanics we call the sets $A\subset I_n$ {\it polymers} and the function $ G^c_A $
{\it weight} of the polymer. Unlike in statistical mechanics this weight is operator-valued and our
objective is to manipulate this expansion to arrive to standard scalar valued weights.

\subsubsection{Notation for combinatorics}   \label{sec: notation for combinatorics}
We gather some notation (partially already introduced above) that will be used throughout. 
We let
$ I_n =\{1, \ldots, n \} $
which is interpreted as the set of discrete times, often denoted by $\tau, \tau'. \ldots$.   It will be convenient to add two "boundary elements" to this set, which are represented by the times $\initial$ and $\final$, hence we set
\beq
\breve{I}_{n} =   I_n   \cup \{\initial, \final \}
\eeq

The set $\frB_n$ is the set of collections of subsets of $I_n$, i.e.\  $\frB_n= 2^{(2^{I_n})}$. Elements of $\frB_n$ will mostly be denoted by $\caA=\{A_1, A_2, \ldots, A_m \}$ with $A_{i} \subset I_n$.  We define two important subsets of $\frB_n$: $\frB^{0}_n$ is the set of  collections of mutually disjoint sets in $I_n$, and $\frB^1_n$ is the set of collections of sets such that the distance between any of them is greater than $1$; 
\beq
\frB^{j}_n :=  \{ \caA  \subset 2^{I_n} \,  \big\str \,   \forall A, A' \in \caA:   (A \neq A' \Rightarrow   \distance (A, A')  > j ) \}, \qquad   j=0,1
\eeq
where $\distance (A, A')  =\min_{\tau \in A, \tau'\in A'} \str \tau-\tau'\str$.
Similarly, we define $\breve{\frB}_{n}, \breve{\frB}^0_{n}, \breve{\frB}^1_{n}$ starting from $\breve{I}_{n} $ instead of $I_n$.
Intervals in $ I_n$ and  $\breve I_n$ are sets of consecutive numbers. 
For any $A$, we say that the intervals $J_1\ldots, J_m$ are the maximal intervals in $A$ iff.\   $\cup_j J_j =A$ and the collection $\{ J_1, \ldots, J_m\} $ is in $  \frB^{1}_n$ or $\breve  \frB^{1}_n$.

We define in general the support of $\caA  \in \breve\frB_n$ as 
\beq
\supp \caA = \cup_{A \in \caA} A
\eeq
and the 'span' of sets as
\beq
\spann A:= \{\min A, \ldots, \max A \}, \qquad   \spann \caA = \spann (\supp\caA) 
\eeq
that is  $\spann A, \spann \caA$ is the smallest interval that contains $A,\supp \caA$, respectively. 
The size of $\spann A$, $\supp \caA$ is called the diameter of $A$, $\caA$, denoted by 
\beq
\dist(A) := 1+  \max A -\min A, \qquad \dist(\caA)= 1+  \max \supp\caA -\min \supp\caA
\eeq

%We will also need \emph{discrete intervals} $J$, these are sets 
%

\subsubsection{Initial state and observable} \label{sec: initial state and observable}

To deal with the observable $O$ and initial state $\rho_0$ in a convenient way we also define the operators $U_{\initial}, U_{\final}$
  \baq
 U_{\initial}  \rho & := &   \e^{(\ka/2)N}{\caW}(\psi_{\realinitial}) \rho  \caW^*(\psi_{\realinitial})   \e^{(\ka/2)N},  \\[1mm]  U_{\final} \rho & :=&   \caW(\e^{\i (n/ \la^{2})\str q\str }\psi_{\realfinal})  \rho
  \eaq  
where we wrote $\caW(\psi)$ instead of $\lone \otimes \caW(\psi)$. Note that $U_{\final}$ depends on the total macroscopic time $n$, which is a notational drawback of our formalism.
We introduce a modified reduced dynamics  $\breve Q_{n}$,  as
\beq\label{modifieddynamics}
  \breve{Q}_{n}  \rho_{\sys}      :=    \Tr_{\res} \left[    U_{\final}   \e^{\i  (n/\la^2) L_\res}   \e^{-\i (n/\la^2)  L_{\ka}} U_{\initial} (\rho_\sys \otimes P_{\Om} )    \right]
\eeq
such that we have 
\beq
Z_n(O, \rho_0, \ka)= \Tr_\sys [O_\sys  \breve Q_{n}    \rho_{\sys,0}]   \label{eq: z as reduced breve}
 \eeq
which generalizes \eqref{eq: first encounter simple z} and reduces to the latter when $\psi_{\realinitial}=\psi_{\realfinal}=0$. To check \eqref{eq: z as reduced breve},  note that  $\caW(\e^{\i t \str q\str }\psi_{\realfinal} )= \e^{\i t H_\res } {\caW}(\psi_{\realfinal}) \e^{-\i t H_\res } $. 

It is straightforward to extend the formalism of Sections \ref{sec: correlation functions of excitations}-\ref{sec: connected correlation functions}  to incorporate the times $\initial, \final$. We define
\baq
 T_{\initial}&:=&   \bbE(U_{\initial})=  \left\langle  \e^{(\ka/2)N}   {\caW}(\psi_{\realinitial}) \Om,   \e^{(\ka/2)N} {\caW}(\psi_{\realinitial}) \Om \right\rangle  \lone \label{Tinit}\\
  T_{\final}&:=&   \bbE(U_{\final})=\left\langle   \Om,  {\caW}(\psi_{\realfinal})   \Om\right\rangle  \lone \label{Tfinal}
 \eaq
 and set
 \beq
  B_{\tau} :=     U_{\tau} -T_{\tau}, \qquad  \tau=\initial,\final.
\eeq

Next, we define copies $\scrR_{\initial}, \scrR_{ \final}$ of $\scrR$, and also $\scrR_A$ for $A \cap \{\initial, \final \} \neq \emptyset$, completely analogous to the construction in Section  \ref{sec: correlation functions of excitations}, such that the definition of (connected) correlation functions $G^{c}_A$ extends to all  $A \subset \breve I_n$.
 In our expansion, we will need $\bsI_{\initial}[T_{\initial} ], \bsI_{\final}[T_{\final} ]$ which we write simply as  $T_{\initial}, T_{\final}$. 
 With these definitions,  the representation \eqref{eqZ from connected correlation functions} is generalized to 
 \beq
\breve Q_{n} =       \sum_{\caA \in \breve \frB^0_n }   \caT   \left[  \mathop{\opprod}\limits_{A \in \caA}  G^c_A   \,  \mathop{\opprod}\limits_{ \tau \in \breve I_n \setminus \supp \caA} T_{\tau}    \right]  \label{eqZ from connected correlation functions boundary} 
\eeq

\subsubsection{Norms }  \label{sec: norms}

Let us introduce some conventions. For $S$ acting on $\scrH_\sys$, we write 
\beq
\norm S \norm := \sup_{\psi \in \scrH_\sys, \norm \psi \norm =1}  \norm S \psi \norm 
\eeq
and for $E$ acting on $\scrB(\scrH_\sys)$, we write
\beq
\norm E \norm  :=    \sup_{\rho \in \scrB_1({\scrH_\sys}), \norm \rho  \norm_1=1}   \norm  E(\rho)  \norm_1, \qquad   \textrm{with}  \, \,   \norm \rho  \norm_1= \Tr \str \rho \str,
\eeq
i.e.\ the natural operator norm on $\scrB(\scrB_1(\scrH_\sys))$. 

For $E \in \scrR_A, \str A \str >1$, we exploit that $E$ can be  written  as a finite sum of elementary tensors  $ $, i.e.\
\beq
E = \sum_\nu     E_\nu,   \qquad E_{\nu}=\otimes_{\tau \in A} E_{\nu,\tau}, \qquad  E_{\nu,\tau} \in \scrR_\tau,
\eeq
to  define
\beq \label{def: weird norm}
\norm E \normw:= \inf_{\{ E_\nu \}}  \sum_{\nu} \prod_{\tau \in A} \norm E_{\nu, \tau} \norm  
\eeq
where the infimum ranges over all such elementary tensor-representations of $E$.
This norm  is useful because of  the following properties (trivial from the definition):  
\ben
\item 
For any family of operators  $K_{A \in \caA}$ with  $K_A \in \scrR_A$ and $\caA$ a collection of disjoint sets, i.e.\ $\caA \in \breve\frB^0_{n}$,  we have
\beq
\left \norm    \mathop{\opprod}\limits_{A \in \caA} K_A   \right \normw  \leq    
\mathop{\prod}\limits_{A \in \caA}  \norm K_A \norm_{\weird}    
\label{eqbound w norm}
\eeq
\item For any  $K_A \in \scrR_A$, 
\beq
\left \norm   \caT \left[  K_A  \right] \right \normw  \leq    
\left \norm   K_A   \right \normw 
\eeq
\een

\subsection{Bounds on operator-valued polymers}  \label{sec: bounds on operator valued polymers}

We aim to set up a perturbative scheme where the $G^c_{A}$ will describe small corrections to  $T$, and hence we must specify in what sense the operators $G^c_{A}$ are small.  

We will often need to distinguish between \emph{bulk} polymers, i.e.\ subsets of $I_n$, and boundary polymers, i.e.\  subsets $A \subset \breve I_n$ for which $A \cap \{ \initial, \final\} \neq \emptyset$.  
Lemma \ref{lem: bound on primary polymers}  gathers the necessary properties of operator valued polymers. 
To relate these properties to assumptions on the original model, we introduce the 'renormalized coupling constants'
\beq
\ep :=   \str \la \str^{2\min{(\al,1)}}, \qquad   \breve \ep :=   \max(\str \la \str,\ep)  \label{def: renormalized couplings}
\eeq
We use in general  $C,c $ to denote constants that can depend on all model parameters except the coupling constant $\la$, conjugation parameter $\ka$, the macroscopic time $n$ and the initial state and observable. By $\breve C, \breve c$ we denote constants that can also depend on the initial state and the observable (but not on $\la, \ka$ or the macroscopic time $n$).   
For operators $J,J' \in \scrB(\scrH_\sys)$, we use the Hilbert-Schmidt scalar product $\langle J, J' \rangle= \Tr [ J^* J'] $.
Recall the diameter of $A$, $\dist(A)$,  and write $\dist(A)^{\al} := (\dist(A))^{\al}$.

\begin{lemma}\label{lem: bound on primary polymers}
For sufficiently small $\str \la \str, \str \ka \str$, but $\la \neq 0$, the following hold uniformly in $\la,\ka$ (i.e.\ $\gap_T$ and all constants can be chosen independent of $\la,\ka$)
\ben
\item 
The  operator $T$  has  a simple eigenvalue $\e^{\theta_T}$ with $\theta_T= \theta_T(\la, \ka)$ and  a gap $ \gap_T >0$, in the sense that 
\beq \label{eq: iterates of t}
\norm T^m - \e^{m\theta_T}R  \norm \leq C_T  \e^{ m(\str\theta_T\str-\gap_T)}, \qquad  m \in \bbN
\eeq
where $R=R(\la, \ka)$ is the (one-dimensional) spectral projector associated to the eigenvalue $\e^{\theta_T}$. 
 Moreover, $\theta_T(\la, \ka=0)=0$ and $\lim_{\la\to 0} \theta_T(\la, \ka)=0$. $R$ is of the form
 \beq
R =  \str \tilde \eta \rangle  \langle \eta  \str.  \label{eqr explicitly}
\eeq
where, for  $\ka=0$, we can choose $\tilde \eta=1$ and $\eta$ a density matrix.  For $\ka \in \bbR$, we can choose $\eta, \tilde \eta$ positive.
\item  Let $\ep,\breve\ep$ be as in  \eqref{def: renormalized couplings}.  The bulk polymers satisfy   \beq    
\mathop{\max}\limits_{\tau \in I_n}  \sum_{A \subset I_n : A \ni \tau }  (C\ep)^{ { -(\str A \str-1) } }   \dist(A)^{\al}   \norm G^{c}_A  \norm_{\weird}   \leq   1   \label{eqbound operator polymers}
\eeq
For boundary polymers, take $\tau=\initial,\final$, then,
\beq   
\sum_{A \subset \breve I_n: A \ni \tau, A \neq \{ \initial,\final\} }  (C\breve\ep)^{  -\str A \cap I_n\str}   \dist(A)^{\al}    \norm G^{c}_A  \norm_{\weird}   \leq   \breve C   \label{eqbound operator polymers boundary}
\eeq
For $A= \{ \initial,\final\}$ (excluded from the sum above), we have $\dist(A)^{\al}    \norm G^{c}_A  \norm_{\weird} \leq \breve C$.
\een

\end{lemma}

The proof of this lemma is in Section \ref{secdiscretization}.

\subsection{Scalar polymers}  \label{sec: scalar polymer model}

The representation \eqref{eqZ from connected correlation functions} 
 evokes the picture of a leading dynamics $T$ interrupted by excitations, indexed by the sets $ A \in \caA$.  One could call this representation a polymer expansion, but it is not yet what we want because the values  of the polymers, i.e.\  $G^c_A$,  are operators.  To make them scalar, we exploit the dissipativity of the model, namely the fact (see Lemma \ref{lem: bound on primary polymers}) that the reduced dynamics $T$ has a maximal (in modulus), simple eigenvalue. Recalling
that the corresponding spectral projection is denoted by $R$ we write
\beq\label{Tdeco}
T= R T  +   R_{\perp} T, \qquad   RT  =  \e^{\theta_T}R
\eeq
where $ R_{\perp}:=1-R$ and we have $R T=TR$. Hence, in particular
\beq\label{TtauR}
R T   R_{\perp} T=R_\perp T   R T=0.
\eeq
We insert this decomposition into the expansion
\eqref{eqZ from connected correlation functions}: Let  $\caJ(\caA)=\{J_1,\dots, J_k\}$
be the family of maximal intervals in  $I_n \setminus \supp \caA$.
The relation  \eqref{TtauR} implies that each tensor product
$ \mathop{\opprod}\limits_{ \tau \in J_j} T_{\tau} $ in \eqref{eqZ from connected correlation functions} may be replaced by $ \mathop{\opprod}\limits_{ \tau \in J_j} (RT)_{\tau} 
+\mathop{\opprod}\limits_{ \tau \in J_j} (R_\perp T)_{\tau} $. Thus \eqref{eqZ from connected correlation functions}
becomes
 \beq
 Q_{n/\la^2} =       \sum_{\caA \in  \frB^0_n }  \sum_{\caJ\subset\caJ(\caA)} \caT   \left[  \mathop{\opprod}\limits_{A \in \caA}  G^c_A   \,  
 \mathop{\opprod}\limits_{ \tau \in  \supp \caJ} (R_\perp T)_{\tau} \mathop{\opprod}\limits_{ \tau \in I_n \setminus (\supp \caJ\cup \supp \caA)} (RT)_{\tau} 
   \right].  \label{eqZ from connected correlation functions1} 
\eeq
We will now produce a new family of polymers $\caA'=\caA'(\caA,J)$ by "fusing" some of the sets in the family $\caA\cup\caJ$.  Let $\Gamma=\Gamma(\caA,\caJ)$ be the graph with vertex set   $\caV(\Gamma)=\caA\cup\caJ$ and
edges 
$ \{S, S'\}$ whenever the sets $S$ and $S'$ are adjacent i.e. $\mathrm{dist}(S,S')=1$ (in particular this
implies that at least one of them is in  $\caA$). For each 
 connected component $\ga$ of $\Gamma$ write  $\caV(\gamma)=\caA_\ga\cup\caJ_\ga$,
set 
$A'_\ga:=\supp \caA_\ga\cup\supp \caJ_\ga$ and let $\caA'$ be the family of $A'_\ga$. 
 Defining
\beq
V((\caA_\ga,\caJ_\ga)) :=     \mathop{\opprod}\limits_{A \in \caA_\ga} G^c_{A} \mathop{\opprod}\limits_{\tau \in\supp \caJ_\ga}  (R_{\perp}T)_{\tau}, \qquad        \label{eqdef polymer weight big v with collections}
 \eeq
 the identity \eqref{eqZ from connected correlation functions1}  becomes
 \beq
Q_{n/\la^2} =       \sum_{\caA \in  \frB^0_n }  \sum_{\caJ\subset\caJ(\caA)} \caT   \left[  \mathop{\opprod}\limits_{\ga}  V((\caA_\ga,\caJ_\ga))   \,  
\mathop{\opprod}\limits_{ \tau \in I_n \setminus  \supp \caA'(\caA,\caJ)} (RT)_{\tau} 
   \right].  \label{eqZ from connected correlation functions boundary2} 
\eeq
The next step is to write \eqref{eqZ from connected correlation functions boundary2} as a sum
of families $\caA'$, i.e. to fix $\caA'$ and sum over $\caA$ and $\caJ$. 

%Let $\frS_n$
%be the set of pairs
%$\caS=
%(\caA,\caJ)$
%where $\caA$ is a family of mutually disjoint subsets of $ I_n$ and $\caJ$ a family of mutually disjoint
%intervals in $ I_n$ disjoint from the sets in $\caA$. We may also view $\caS$ as a family
% of mutually disjoint subsets $S$ of $ I_n$ together with a label "A-set" or "J-set". 

\begin{definition} \label{def: scalar weights}  Let $\caA$ be a family of mutually disjoint subsets of $I_n$ and let $\caJ$ be a family of mutually disjoint and non-adjacent intervals in $I_n$ and disjoint from the sets in $\caA$, i.e.\ $\supp \caA \cap \supp \caJ = \emptyset$. 
We say that  $\caS=( \caA, \caJ) $ is a \emph{fusion}   iff. 
\ben
%\item  All sets $S \in \caS$ are mutually disjoint (that is: disjoint as subsets of $I_n$)
\item       $\distance( I_n \setminus \supp \caS, \supp \caJ  )>1$ where
 $\supp \caS:=\supp \caA\cup\supp \caJ$. (Pictorially; the intervals $J$ are in the "interior" of $\supp\caS$.)
\item  The graph $\Gamma(\caA,\caJ)$ is connected.   
\een
The set of fusions is denoted by $ \frS_n^f $.% Likewise we define $ \frS_n^f ,\frS_n$ by replacing $\breve I_n$ by $ I_n$.
\end{definition}
Defining
\beq
 {\scriptsize \Si}V (A') :=    \sum_{ \caS=(\caA,\caJ) \in \frS^f_n: \supp \caS =  A' }  V((\caA,\caJ))        \label{eqdef polymer weight big v}
\eeq
we obtain the representation for $ Q_{n/\la^2}$  in terms of fusions as (we drop the prime from
$A'$ and $\caA'$)
\beq
 Q_{n/\la^2} = \sum_{\caA \in \frB^1_{n}}   \caT   \left[     \mathop{\opprod}\limits_{ \tau \in I_n \setminus \supp \caA} (RT)_{\tau} \,  \mathop{\opprod}\limits_{A \in \caA}  
 {\scriptsize \Si}V(A)    \right]   . \label{eq: q in terms of primed a}
\eeq
Note that by construction $ \caA \in \frB^1_{n}$ since the sets $A'_\ga$ above are non-adjacent i.e. their mutual distances are at least $2$. Hence, for any $\caA$ in the formula above, all $\tau$ that are adjacent to $\supp\caA$ carry the projector $RT$.  
A pictorial way to phrase this is that any $\caA$ in \eqref{eq: q in terms of primed a}  is surrounded by projections $R$.  We exploit this by defining
\beq
\hat v(A') :=       \caT \left[  {\scriptsize \Si}V(A')  \bigotimes_{\tau \in I_n \setminus A'}  R_{\tau}  \right], \qquad  \hat v(A') \in \scrR        \label{eqdef polymer weight v}
\eeq
Note that $\hat v(A')$ is a multiple of $R$ unless $A' \ni 1$ and/or $A' \ni n$. To eliminate these boundary effects, we consider  $ R Q_{n/\la^2} R$ instead of $  Q_{n/\la^2} $, for the time being.  Then \eqref{eq: q in terms of primed a} implies
\beq
 R Q_{n/\la^2} R =  \sum_{\caA \in \frB^{1}_{n}}   (RT)^{(n-\str \supp \caA \str) }    \prod_{A \in \caA}     (R  \hat v(A)  R)    \label{eqpolymer with tildes}  \eeq
 where the product on the RHS is commutative since the projection $R$ is one-dimensional.
It is convenient to extract the contribution due to the  $RT$, by defining (recall $\eta,\tilde\eta$ from \eqref{eqr explicitly})
  \beq
v(A)   =     \e^{- \str A \str \theta_T}  \langle \tilde\eta,\hat v(A) \eta \rangle 
\eeq
Indeed, since $\langle \tilde \eta,  \eta \rangle =1$ (because $RR=R$), we get
\beq
 R Q_{n/\la^2} R = R e^{n \theta_T}    \sum_{\caA \in \frB^1_{n}}    \prod_{A \in \caA}       v(A)      \label{eq: toy polymer}
\eeq
and the RHS is indeed a 'scalar polymer representation', i.e.\ the weights  $v(A)$ are numbers,  and we can study it with the help of a cluster expansion.
Note for later use that
\beq \label{eq: z empty}
\Tr R Q_{n/\la^2}R = Z_n(\tilde\eta\otimes \lone,\eta \otimes P_{\Om},  \ka). 
\eeq
Our real object of interest however is the partition function \eqref{eq: z as reduced breve}
expressed in terms of the operator
\eqref{eqZ from connected correlation functions boundary}. We proceed with the latter as with 
$Q_{n/\la^2} $. Decompose as in  \eqref{Tdeco}
\beq
T_\tau= R T_\tau  +   R_{\perp} T_\tau, \qquad   \tau\in\{\initial,\final\}
\eeq
where $R$ is the same operator as in  \eqref{Tdeco}.
Since 
$T_{\initial}, T_{\final}$ are proportional to the identity operators, we have again $R T_\tau=T_\tau R$.
As a consequence, for all $\tau,\tau'\in\breve I_n$ 
\beq\label{TtauR2}
R T_\tau   R_{\perp} T_{\tau'}=R_\perp T_\tau   R T_{\tau'}=0.
\eeq
We end up with the expansion like \eqref{eq: q in terms of primed a} for $\breve Q_{n} $
with the expected difference that
we have to replace $I_n$ by  $\breve I_n$ and $ \frB_n^1$ by $\breve  \frB_n^1$ (also in the definition of fusions). 

For $Z_n$ of \eqref{eq: z as reduced breve} we get the scalar valued polymer expansion as in
\eqref{eq: toy polymer} with the change that the weights $v(A)$  for $A\cap \{\initial,\final\}\neq\emptyset
$ have to be slightly modified because the weight of the polymer is influenced by the initial state and the observable. 
We keep the definition of $\hat v(A)$ given in \eqref{eqdef polymer weight v}.  Then, recall that $R= \str \eta \rangle \langle \tilde \eta \str$ and define
\beq
 v(A)  := \e^{-\str A\cap I_n \str  \theta_T }  k^{-1}_A \left\{ \begin{array}{lr}         \langle \tilde\eta,  \hat v(A)  \rho_{\sys,0}  \rangle & \qquad      \initial \in A,  \final \notin A   \\[3mm]
 \langle O_\sys,   \hat v(A) \eta \rangle &    \final \in A,   \initial \notin A   \\[3mm]
 \langle  O_\sys,  \hat v(A)  \rho_{\sys,0}  \rangle &     \{\initial, \final \}  \in A
      \end{array} \right.    \label{def: boundary values}
\eeq
where 
\beq
k_A:=  \left(\indicator_{[\initial \in A]} k_{\realinitial} + \indicator_{[\initial \not\in A]}    \right)  \left(\indicator_{[\final \in A]} k_{\realfinal}+ \indicator_{[\final \not\in A]}    \right) 
\eeq
and 
\baq
k_{\realinitial}   &:= &             \langle \tilde\eta,  T_{\initial}\rho_{\sys,0} \rangle
= \Tr_\sys( \tilde\eta, \rho_{\sys,0} )
\langle  \e^{(\ka/2)N}   {\caW}(\psi_{\realinitial}) \Om,   \e^{(\ka/2)N} {\caW}(\psi_{\realinitial}) \Om \rangle\label{kinitial}\\
% (=  \Tr [( \tilde \eta \otimes \lone )    \e^{\ka N}\rho_0  ]) \\[1mm]
 k_{\realfinal}  &:= &     \langle O_\sys, T_{\final} \eta  \rangle 
 =  \Tr_\sys(O_\sys \eta) \Tr_\res ({\caW}(\psi_{\realfinal})P_\Om)\label{kfinal}
\eaq
With these definitions, we arrive at
\beq
Z_{n} =    \e^{ n \theta_T  }  k_{\realinitial} k_{\realfinal} \sum_{\caA \in \breve\frB^1_n}       \prod_{A \in \caA}  v(A)    \label{eqscalar polymer model boundaries}
\eeq
%
% if we replace $T \to T_{\tau}$. Indeed, since  $T_{\initial}, T_{\final}$ are multiples of identity, they commute with $R$ and in particular $(T_\tau R)(T_\tau R_{\perp}) =0$, which is a necessary for the second equality in \eqref{eqlong calculation scalar}.  From \eqref{eqlong calculation scalar}, we get to \eqref{eqscalar polymer model boundaries} by using the above modified definition of $ v(A)$. We verify by inspection  and for the boundary polymers one verifies by inspection that  \eqref{eqscalar polymer model boundaries} holds. 
which is our representation of $Z_n$ as a polymer model. 
A few remarks:
\ben
\item The case where the factors $k_\realinitial, k_\realfinal$ vanish, will be of no concern, as we will explain later, following Lemma \ref{lem: convergence invariant state}.
\item
Up to now, we did not indicate the macroscopic time $n$ in the notation. Let us, only in the next lines, indicate this dependence by a superscript. Consider $A$ such that  $\max A -1 <n< n'$, then 
\beq \label{eq: independence macro}
v^{(n)}(A) = v^{(n')}(A).
\eeq 
Furthermore, $ v^{(n)} (A+\tau)=v^{(n)}(A)$ for $\tau \in \bbN$, as long as $0 <\min A,  \max A +\tau -1\leq n$.  If we allow polymers to contain 
the final time $n+1$ but not the initial time $0$, i.e.\ $0 < \min A, \max A -1 \leq n < n'$, then we have 
\beq \label{eq: independence final}
v^{(n)}(A) = v^{(n')}(A+(n'-n)).
\eeq 
These properties follow from the property $G_A=G_{A+\tau}$ and the expression for the 'final-time' operators $U_{\final}, U_{n'+1}$.
\item  If we choose $\rho_0= \eta \otimes P_{\Omega}$, then $v(A)=0$ whenever $ 0 \in A$. Indeed, Indeed, in that case we can substitute $\lone$ for $U_0$, hence $B_0=0$, and we have $R_{\perp} \rho_{\sys,0}=0$ because $\rho_{\sys,0}=\eta$.
By analogous reasoning, we check that if we choose $O=\tilde \eta \otimes \lone$, then $v(A)=0$ whenever $ n+1 \in A$. 
Note that these two observations are consistent with (\ref{eq: toy polymer},\ref{eq: z empty}) which tell us that one can omit the boundary polymers.
\een
We finish
this section by a pictorial illustration of the construction.
Each term in the representation of $Q_{n/\la^2}$ in \eqref{eqZ from connected correlation functions} can be represented by a picture as in \eqref{fig: splittingintoteng}. The horizontal axis is the time-axis which has been divided into intervals of the form $\la^{-2}[\tau-1,\tau]$ by the discretization procedure. Each term in \eqref{eqZ from connected correlation functions} is specified choosing a collection $\caA$, this is indicated on our picture by connecting the $\tau$ that belong to the same $A \in \caA$ by straight diagonal lines above the axis.  The times $\tau \in \supp \caA$ are drawn with the symbol \includegraphics[width = 0.8cm, height=0.2cm]{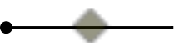} and the $\tau \not \in \supp \caA$ are drawn with a \includegraphics[width = 0.8cm, height=0.2cm]{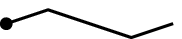}.  
%The fact that times are represented by 'blocks'  is motivated by their microscopic origin, since they in fact correspond to the time-intervals, see Section \ref{secdiscretization}.

 As to the operator value of this picture, each \includegraphics[width = 0.8cm, height=0.2cm]{stukjeT.eps} corresponds to $T$, and each connected component of \includegraphics[width = 0.8cm, height=0.2cm]{stukjeG.eps}
corresponds to an operator $G^c_A$. 

\begin{figure}[h!] 
\vspace{0.5cm}
\begin{center}
\def\svgwidth{\columnwidth}
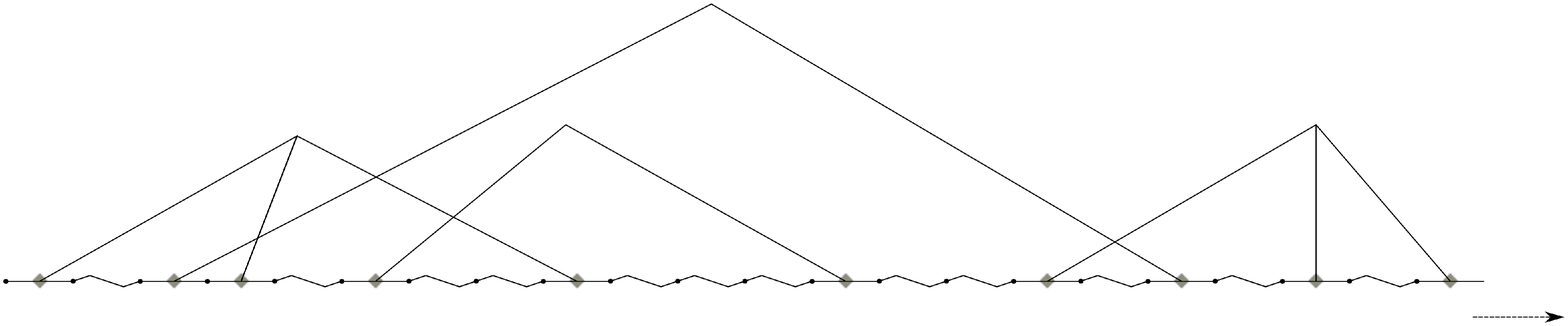
\caption{ \footnotesize{Representation of a term in the sum \eqref{eqZ from connected correlation functions} }   \label{fig: splittingintoteng} }
\end{center}
\end{figure}

We now generate new pictures by splitting
\beq
T= TR+ TR_{\perp} = \left.\begin{array}{ccccc}  \includegraphics[width = 0.8cm, height=0.2cm]{stukjeT.eps}  &=&  \includegraphics[width = 0.8cm, height=0.2cm]{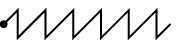}   & + &  \includegraphics[width = 0.8cm, height=0.2cm]{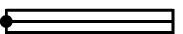}   \end{array}\right.
\eeq
In this way we obtain $2^{\ell}$ pictures from Figure \ref{fig: splittingintoteng}, with $\ell$ the number of \includegraphics[width = 0.8cm, height=0.2cm]{stukjeT.eps}.  One of these is the upper picture in Figure \ref{fig: splittingoftingraph}. 
 Note that splittings in which one \includegraphics[width = 0.8cm, height=0.2cm]{stukjeTR.eps}  is adjacent to an \includegraphics[width = 0.8cm, height=0.2cm]{stukjeTRperp.eps} do not contribute because $R R_{\perp}=0$. Therefore, we are in fact choosing the splitting for each maximal interval in $I_n \setminus \supp \caA$. 
In a next step, we  do not distinguish between excitations that originate from $B_{\tau}$ or $TR_{\perp}$ and we write simply 

 \includegraphics[width = 0.8cm, height=0.2cm]{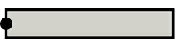}  for  both \includegraphics[width = 0.8cm, height=0.2cm]{stukjeTRperp.eps}  and    \includegraphics[width = 0.8cm, height=0.2cm]{stukjeG.eps}. 
 Moreover, we fuse adjacent  \includegraphics[width = 0.8cm, height=0.2cm]{stukjeGplusTRperp.eps} to form new polymers $A'$, the lower picture in Figure \ref{fig: splittingoftingraph} shows the result of fusing the upper picture.   As a result, these new polymers are surrounded by  \includegraphics[width = 0.8cm, height=0.2cm]{stukjeTR.eps}, and since those correspond to one-dimensional projectors, the operator-valued contribution of a new polymer to the total sum is independent of the presence of any other polymers. This is why the new representation is {scalar}.

\begin{figure}[h!] 
\vspace{0.5cm}
\begin{center}
\def\svgwidth{\columnwidth}
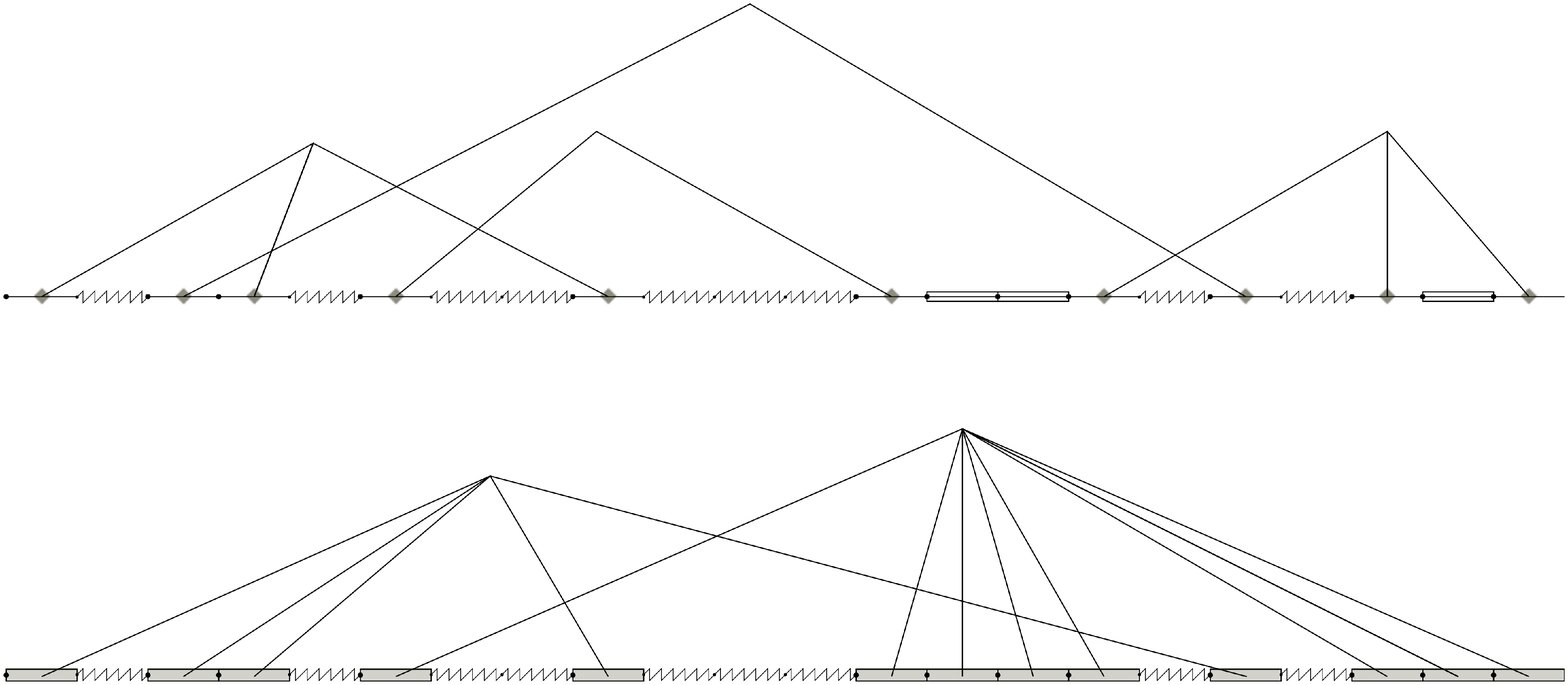
\caption{  \footnotesize{In the upper picture, we split $T=TR+TR_{\perp}$ for all $\tau \not \in \cup_{i=1}^4 A_{i}$ and we choose one of the two terms for each  $\tau$. In the example displayed, we choose $TR_\perp$ for $\tau_{1,2,3} = 14,15,21 $  (those that correspond to the triple lines in the upper picture), and $TR$ for all other $\tau$'s. 
We fuse blocks marked with $TR_{\perp}$ with a polymer  whenever they are adjacent to that polymer and we fuse two polymers whenever they are adjacent (possibly after having been fused with $TR_{\perp}$-blocks).  In the example above, $A'_{1}=A_{1}\cup A_{2}, A'_{2}=A_{3}\cup A_{4}  \cup \{ \tau_1, \tau_2, \tau_3 \}$.}
 \label{fig: splittingoftingraph} 
}
\end{center}
\end{figure}

\subsection{Bounds on scalar polymers}  \label{sec: bounds on scalar polymers}

To obtain  bounds on the new, scalar, polymers, we use the smallness of the operators
$G^c_A$, expressed by the renormalized couplings $\ep,\breve \ep$, and their summability in $A$ when one of the elements of $A$ is held fixed. The  
 $TR_\perp$ factors, which glue together several  $G^c_A$-operators into one scalar polymer, 
are not small (they are of order $1$) but whenever we line up $\ell$ of them, their weight decays exponentially in $\ell$ by the ergodicity of $T$ (see Statement 1) of Lemma  \ref{lem: bound on primary polymers}).  Save for $2$ special cases, dealt with in \eqref{eqbound scalar boundary polymers special} below, all  scalar polymers contain at least one operator $G^c_A$ so that one can always extract a factor $\ep$ or $\breve \ep$.

\begin{lemma}[bounds on polymer weights]       \label{lem: bound on scalar polymers}
For sufficiently small $\str \la\str\neq 0$, there is a $a_v>0$ such that
\ben
\item{
For bulk polymers, 
\beq
 \sum_{A \subset I_n:  A \ni  \tau }  \e^{a_v \str A \str }   \dist(A)^{\al}  \str v(A) \str \leq   C \ep,  \label{eqbound scalar polymers}
\eeq
}
\item{
For boundary polymers, take $\tau=\initial,\final$, then
\beq
\mathop{ \sum}\limits_{ \scriptsize{ \left. \begin{array}{c} A \subset \breve I_n:  A \ni  \tau   \\   A \neq \{ \initial, \final \}, A \neq \breve I_n    \end{array} \right.  } }
  \e^{a_v \str A \str }    \dist(A)^{\al}  \str v(A)\str  \leq   \breve C \breve\ep,     \label{eqbound scalar boundary polymers}
\eeq
and the boundary polymers that are excluded in the sum above satisfy 
\beq
 \str v(\breve I_n)  \str  \leq \breve C \e^{-2 a_v n} , \qquad    \str v(\{ \initial, \final \})  \str  \leq \breve C n^{-\al}   \label{eqbound scalar boundary polymers special}
\eeq
}
\een
\end{lemma}
\begin{proof}
We first prove \eqref{eqbound scalar polymers}. Choose $A' \in I_n$.
By the properties of the norm $\norm \cdot \normw$ stated in Section \ref{sec: norms} , we get 
\baq
\str v(A') \str  &\leq &  C  \e^{-\str A \str \Re \theta_T} \sum_{\caS=(\caA ,\caJ) \in \frS_n^f:  \supp \caS =  A'}  \norm R \normw^{p}    \prod_{A \in \caA} \norm G^c_A \normw    \prod_{J \in \caJ}    \norm (TR_{\perp})^{\str J\str} \normw    \label{eq: naive bound on v}
\eaq
where $p$ is the number of maximal intervals in $\spann A' \setminus A'$, and we will upper bound   $p< C \str \supp \caA \str $. 
For any  fusion $\caS= (\caA, \caJ)$ such that $\supp (\caA \cup \caJ)=  A'$, we have 
\beq
\dist(A')^{\al} \leq  \prod_{A \in \caA} \dist(A)^{\al}   \prod_{J \in \caJ}  (1 + \str J \str)^{\al}    \label{eq: naive bound on dist}
\eeq
which follows immediately from the connectedness of the graph $\Gamma(\caA,\caJ)$ in Definition \ref{def: scalar weights}. 

To continue, let us interpret pairs $\caS= (\caA, \caJ)$ as a collection of sets $S \in \caA \cup \caJ$  together with a label indicating whether they are `$A$-sets' ($\in \caA$) or `$J$-sets' ($\in \caJ$) (of course, only intervals can be $J$-sets).    
Let us gather such collections of labelled sets in the set $\frS_n$.
Furthermore, we consider the adjacency relation $\sim$  defined by
$
S \sim S' $ iff.\  $\dist(S,S')=1$ and at least one of the sets $S,S'$ is an $A$-set.   Then, it immediately follows that, for any fusion $\caS$, the graph $(\caS, \sim)$ is connected.  

We introduce (with $C_w, c_w >0$, to be fixed below)
 \beq
 w_{\ep}(S) :=  \left\{ \begin{array}{ll}      \dist(A)^{\al} (\ep C_w)^{-(\str A \str-1)}   \norm G^c_A \normw    & \textrm{if $S$ is an $A$-set}   \\[2mm]
 C_w^{-1}  \e^{-c_w \str J \str}    & \textrm{if $S$ is a $J$-set} 
 \end{array} \right.    \label{def: weps}
 \eeq
Then we can bound the LHS of \eqref{eqbound scalar polymers} as
 \beq  \label{eq: bound by w s}
\e^{a_v \str A' \str} \dist(A')^{\al}  \str v(A') \str  \leq  \ep C    \sum_{\caS \in \frS_n^f:  \supp \caS=  A'}  \prod_{S \in \caS} w_{\ep}(S)
 \eeq
 provided one chooses  $c_w+a_v < \gap_T $ and  $\ep$ small enough (depending on $C_w$ and $c_w$).  To check  \eqref{eq: bound by w s}, use  (\ref{eq: naive bound on v}, \ref{eq: naive bound on dist}), the bounds $\norm R \normw \leq C$ and $\str \caJ \str \leq  C  \str \supp \caA \str$ for any  fusion $(\caA,\caJ)$, the fact that any fusion contains at least one $A$-set, and the bound from Lemma \ref{lem: bound on primary polymers} 1);
 \beq
\e^{-\str J \str  \Re \theta_T}     \norm (TR_{\perp})^{\str J\str}  \norm   \leq C_T  \e^{-\str J \str \gap_T }.
  \eeq
   
The desired bound \eqref{eqbound scalar polymers} will trivially follow once we establish
 \beq
 \sum_{\caS \in \frS_n: \supp \caS \ni \tau}  \indicator_{[(\caS,\sim)  \, \text{connected} ]} \prod_{S \in \caS}   w_{\ep}(S)   \leq   C
 \label{eq: desired s bound}
 \eeq
 because $(\caS,\sim)$ is connected for any fusion $\caS$. 
To prove \eqref{eq: desired s bound}, we first  establish that, provided $C_w$ is sufficiently large, 
 \beq
 \sum_{S \sim S'} w_{\ep}(S)  \e^{a(S)}  \leq   a(S')  \label{eq: kp for s},   \qquad  a(S):=  \left\{ \begin{array}{ll}    \str A \str     & \textrm{if $S$ is an $A$-set}   \\[1mm]
1   & \textrm{if $S$ is a $J$-set} 
 \end{array} \right.       
 \eeq
Indeed, if $S$ is an $A$-set, one uses \eqref{eqbound operator polymers} and if it is a $J$-set, one sums the exponential. Relying on \eqref{eq: kp for s} (which is a 'Kotecky-Preiss' criterion) in the terminology of Appendix \ref{appsec: cluster expansions}, we now apply  the general combinatorial estimate Lemma \ref{lemma: combi trick abstract},   obtaining
\beq
\sum_{\caS \in \frS_n }  1_{[(\caS \cup \{ S'\}, \sim) \,   \mathrm{connected}] }  \prod_{S \in \caS} w_{\ep}(S)  \leq  \e^{a(S')}
\eeq
(we bounded the indicator $\indicator_{[\ldots]}$ by $k(\caS \cup \{ S'\})$ and we used  \eqref{eqbound on clusters containing something abstract}). 
Consequently, 
 \baq
 \sum_{
 \scriptsize{ \left.\begin{array}{c}     \caS  \in \frS_n: \supp \caS \ni \tau       \\      (\caS, \sim)  \, \text{connected   }     \end{array}\right.   } }
 \prod_{S \in \caS}   w_{\ep}(S)   & \leq &    
 \sum_{S': S' \ni \tau} w_\ep(S')
 \sum_{
  \scriptsize{ \left.\begin{array}{c}     \caS  \in \frS_n        \\     (\caS \cup \{ S'\}, \sim)   \, \text{connected   }     \end{array}\right.   } 
  }   \prod_{S \in \caS}   w_{\ep}(S)   \nonumber  \\[3mm]
  & \leq  &      \sum_{S': S' \ni \tau} w_\ep(S')  \e^{a(S')}  \leq  C  \label{eq: trees of cas}
 \eaq
 This proves \eqref{eq: desired s bound} and hence also \eqref{eqbound scalar polymers}. 
 
  The other claims of Lemma \ref{lem: bound on scalar polymers} are proven by analogous reasoning, using the bound on boundary polymers \eqref{eqbound operator polymers boundary}.   
Indeed, the general idea of the proof of \eqref{eqbound scalar polymers}  was to define $w_{\ep} $ in \eqref{def: weps} such that it satisfies \eqref{eq: desired s bound} and to extract a factor $\ep^{-1}$ from all terms in the sum on the LHS of \eqref{eq: desired s bound}, which was possible because all fusions $\caS$ contain at least one $A$-set.
To get \eqref{eqbound scalar boundary polymers}, we modify the definition of $w_\ep$ for  $S$ such that $S \cap \{\initial, \final \} \neq \emptyset$ to (note a new boundary-dependent constant $\breve C_w$)
 \beq
 w_{\ep}(S) :=  \left\{ \begin{array}{ll}      \dist(A)^{\al} (\breve C_w)^{-1}(\breve\ep C_w)^{-\str A \cap I_n \str}   \norm G^c_A \normw    & \textrm{if $S$ is an $A$-set}   \\[2mm]
 C_w^{-1}  \e^{-c_w \str J \str}    & \textrm{if $S$ is a $J$-set} 
 \end{array} \right.    \label{def: weps boundary}
 \eeq
and we use $\breve\frS_n, \breve\frS^f_n$ instead of $\frS_n,\frS^f_n$ by replacing in the definitions $I_n$ by $\breve I_n$. 
Then \eqref{eq: bound by w s} still holds with $\ep$ replaced by $\breve \ep$ provided that $\caS$ contains at least one $A$-set with $\str A \cap I_n \str >0$. This is the case unless $A' = \supp \caS$ is $\{ \initial,\final\}$ or $\breve I_n$, which are the two special cases in Lemma \ref{lem: bound on scalar polymers}. Next, we need to establish \eqref{eq: desired s bound} with $C$ on the RHS replaced by $\breve C$ and $\tau=\initial,\final$.
It is pedagogical to split this estimates according to which boundary times $\caS$ contains. Let us consider the case where $0 \in \supp \caS$, but $ \final \not \in \supp\caS$, in particular choose $\tau=\initial$.  Then, we write, analogous to \eqref{eq: trees of cas},
\baq
 \sum_{
 \scriptsize{ \left.\begin{array}{c}     \caS  \in \breve\frS_n: \final \not \in \supp \caS  \\      (\caS, \sim)  \, \text{connected   }, \, \initial \in \supp \caS     \end{array}\right.   } }
 \prod_{S \in \caS}   w_{\ep}(S)   & \leq &    
 \sum_{S': S' \ni \initial} w_\ep(S')
 \sum_{
  \scriptsize{ \left.\begin{array}{c}     \caS  \in \frS_n        \\     (\caS \cup \{ S'\}, \sim)   \, \text{connected   }     \end{array}\right.   } 
  }   \prod_{S \in \caS}   w_{\ep}(S)   \nonumber  
 \eaq
The sum over $ \caS  \in \frS_n$ concerns only bulk quantities, and it is therefore bounded by $\e^{a(S')}$, just as before. We then conclude by bounding 
\beq
  \sum_{S': S' \ni \initial} w_\ep(S')  \e^{a(S')} \leq \breve C
\eeq
as follows from \eqref{eqbound operator polymers boundary} when $S$ is an $A$-set, otherwise it is the same bound as before. 
The cases where $\caS$ contains $\final$ but not $\initial$, or $\final$ and $\initial$, are dealt with analogously, and we get  \eqref{eqbound scalar boundary polymers}.
Finally, we turn to \eqref{eqbound scalar boundary polymers special}. The claim about $A=\{\initial, \final \}$ follows directly from the bound on that polymer in Lemma \ref{lem: bound on primary polymers}.   The claim concerning $A=\breve I_{n}$ differs from \eqref{eqbound scalar boundary polymers}  because it is now possible that the fusion $\caS$ consists of a single $J$-set $J=\breve I_{n}$, in which case one cannot extract $\breve \ep$. Taking this into account, this claim follows as above. 
\end{proof}

%scriptsize{ \left.\begin{array}          \\          \end{array}\right.   }

\section{From Hamiltonian dynamics to polymer models}   \label{secdiscretization}

In Section \ref{sec: polymer rep},  the reasoning  was largely independent of the details of the underlying Hamiltonian model. Indeed, in  that section, we assumed  some properties of the operator-valued polymers in Lemma \ref{lem: bound on primary polymers} and we explored the consequences of these properties.  Now,  time has come to prove Lemma  \ref{lem: bound on primary polymers}. This lemma   discusses the ergodicity properties of the operator $T$ and  bounds on the correlation functions  $G^{c}_A$.  In both cases, the proof starts by expanding the microscopic evolution in a  Dyson (or Duhamel) expansion. This expansion is standard and has been used many times in a related context.  What might however seem odd at a first glance, is our complicated presentation of the perturbation series, involving tensor products of copies of the space $\scrR = \scrB(\scrB_1(\scrH_\sys))$ (already introduced in Section \ref{sec: correlation functions of excitations}).  We use this formalism since we hope it makes the important estimates more natural and transparent. 
 We first derive the expansion for $G_A, G^c_A$, see the expressions (\ref{eq: correlation micro}, \ref{eq: connected micro}). This part consists of  purely algebraic manipulations (strictly speaking their validity is only established in Section \ref{sec: bounds on correlation functions} where we show that some series is absolutely summable). 
 
The analysis part comes in  Section \ref{sec: proof of bound on operator polymers}, where we control $\norm G^{c}_A \normw$. 
For reasons of readability,  we first restrict ourselves to 'bulk' polymers. Afterwards, the analysis is repeated with minor adjustments to include the boundary polymers; this is done in Section  \ref{sec: boundary polymers}-\ref{sec: bounds on boundary polymers}. 

As mentioned above, we also need to establish  the exponential ergodicity of the operator $T$, which is done in Section \ref{sec: analysis of t}.  The crux of this argument is to relate $T$ to the Markovian approximation, which was already discussed in the introductory Section \ref{sec: plan of the proof}.

\subsection{Dyson series} \label{secexpansions}

Recall and rewrite the reduced dynamics $Q_{t}$, introduced in Section \ref{sec: splitting of the dynamics};
\baq
Q_{t} \rho_\sys &=&  \Tr_{\res} \left[    \e^{-\i   t  L_{\ka} }    (\rho_\sys \otimes \initialresfinite)     \right] = \bbE  \left[   \e^{-\i   t  L_{\ka} }  \right]   \rho_\sys=  \bbE  \left[ \e^{\i   t  L_\res }   \e^{-\i   t  L_{\ka} }  \right]   \rho_\sys
\eaq
By the Duhamel formula \eqref{eq: duhamel on superspace} in Lemma \ref{lem: fake functional calculus}, we get 
\baq \label{eq: first duhamel series}
\e^{\i t L_\sys} Q_{t} \rho_\sys&=&     \sum_{m \in \bbN}  (-1
)^{m}\mathop{\int}\limits_{0< t_1 < \ldots < t_{2m} <t} \d t_1 \ldots \d t_{2m} \, \,  \Tr_{\res}\left[  L_{\inter,\ka}(t_{2m})  \ldots  L_{\inter,\ka}({t_2})  L_{\inter,\ka}({t_1})  (\rho_\sys\otimes P_{\Om})\right] 
\eaq
where, the RHS is understood to be $1$ for $m=0$.
Note that the RHS contains  only terms with an even number of operators $L_{\inter,\ka}(s)$ because $\Tr_{\res}( \Phi P_\Om)=0$ whenever $\Phi$ is a monomial of odd degree in creation/annihilation operators.  
We  use the  Wick theorem  to evaluate the expression in \eqref{eq: first duhamel series}. To this purpose, we use the formalism developed in Section \ref{sec: correlation functions of excitations}, as anticipated in Remark \ref{rem: category}.  For any $t_i, i=1, \ldots,2m$, we define a copy $\scrR_{t_i}$ of the space $\scrR$ (we do not aim to define "continuous" tensor  products; all our formulas contain a finite number of $t_i$'s).  
In particular, we write
\beq
\bbE\left[   L_{\inter,\ka}(t_{2m})  \ldots  L_{\inter,\ka}({t_2})  L_{\inter,\ka}({t_1})  \right] =  \caT  \bbE\left[L_{\inter,\ka}(t_{2m})  \otimes_\sys\ldots  \otimes_\sys L_{\inter,\ka}({t_2}) \otimes_\sys  L_{\inter,\ka}({t_1})  \right] 
\eeq
where, on the RHS, the operator $\caT$ acts on operators in $\scrR_{\{t_1, \ldots, t_{2m}\}}= \otimes_{i=1}^{2m} \scrR_{t_i} $ and it contracts the operators such that those in $\scrR_{t_1}$ are on the right, then those on $\scrR_{t_2}$, etc.., as in  \eqref{eq: def contraction}.

Let $\{u,v \}$ be a pair of times with the convention that $u < v$. Then we define 
\beq\label{Kdefi}
K_{u,v}  : =   -% \la^2
 \bbE\left[ L_{\inter,\ka}({v})  \otimes_\sys  L_{\inter,\ka}({u})   \right], \qquad    K_{u,v}  \in  \scrR \otimes \scrR
 \eeq
which can be written more explicitly:
 \baq
 K_{u,v}   & = &  - \la^2   \e^{\ka}  h(v-u)  \left[  \caR(D(v)) \otimes     \caL(D(u)) \right]  -    \la^2\e^{\ka}  h(u-v) \left[ \caL(D(v)) \otimes    \caR(D(u)) \right]    \nonumber  \\[1mm]
   & +&      \la^2   h(v-u)  \left[ \caL(D(v)) \otimes   \caL(D(u)) \right]  +    \la^2 h(u-v)  \left[  \caR(D(v)) \otimes    \caR(D(u))  \right] 
   \label{eq: explicit for k}
 \eaq
 where $h(-s)= \overline{h(s)}$, defined in Assumption \ref{ass: decay correlation functions} and $D(s)= \e^{\i s H_\sys}D \e^{-\i s H_\sys}$.
 For later use, we note that  $\norm K_{u,v} \normw \leq   \la^2  C \str h(v-u)\str $ for sufficiently small $\ka$ (for example, $C= 4 \norm D \norm \e^{\ka}$).
 We view  the operator $ K_{u,v} $ as acting on the copies $\scrR_v \otimes \scrR_u$ (that is, we should in fact write $\bsI_{\{ u,v\}} [K_{u,v}]$ but, as for the operators $G^c_A$, we prefer to drop the embedding operators $\bsI [\cdot]$.)
Let us now apply the Wick theorem and expand any contribution to the integral in \eqref{eq: first duhamel series} in contractions of the creation and annihilation operators. This yields
\beq
 \caT  \bbE\left[  L_{\inter,\ka}(t_{2m})) \otimes_{\sys}  \ldots  \otimes_{\sys} L_{\inter,\ka}({t_1})) \right]  =  \sum_{\pi \in \textrm{Pair}(t_1,\ldots,t_{2m})} \caT\left[ \mathop{\otimes}\limits_{\{u,v \} \in \pi}  K_{u,v}  \right]    \label{eq: first condensed pair sum}
\eeq
where the sum on the RHS runs over pairings $\pi$, i.e.\ partitions of the times $t_1, \ldots, t_{2m}$ in $m$ pairs $\{u,v\}$ with the notational convention that $u < v$. 
This formula relies crucially on the fact that the interaction is linear in creation/annihilation operators (nonetheless, it is easy to extend the proof so as to cover an additional small quadratic interaction).
By plugging \eqref{eq: first condensed pair sum} into \eqref{eq: first duhamel series}, we obtain
\baq \label{eq: second duhamel series}
\e^{\i t L_\sys} Q_{t} &=&     \sum_{m \in \bbN}  \mathop{\int}\limits_{0< t_1 < \ldots < t_{2m} <t} d t_1 \ldots \d t_{2m} \, \,  \sum_{\pi \in \textrm{Pair}(t_1,\ldots,t_{2m})} \caT\left[ \mathop{\otimes}\limits_{\{u,v \} \in \pi}  K_{u,v}  \right] 
\eaq
where, for $m=0$, the RHS is understood as $1$.

\subsection{The evolution as an integral over time-pairs} \label{seccombinatorics}
 
  We will now rewrite  \eqref{eq: second duhamel series} in a convenient way. 
The integral over ordered $t_1, \ldots, t_{2m}$, together with the sum over pairings, $\pi$,  on the set of times, is represented as an integral/sum over unordered pairs $\{u_i, v_i\}$ with $u_i,v_i \in \bbR_{+}$  and $i=1, \ldots, m$,  such that 
\beq\label{mpairs}
u_i < v_i, \qquad    u_1 < \ldots  < u_m
\eeq
This is done as follows. For any pair $\{ t_r,t_s \} \in \pi$ with $r<s$,  hence $t_r < t_s$, we let $u_i=t_r, v_i =t_s$  where the index $i=1,\ldots,m$ is chosen such that the $u_i$ are ordered $u_1 < u_2 \ldots < u_m$. 
We write $w= \{u,v \}$ for a pair (with the convention that $u < v$) and $\uw$ for a finite, possibly empty, collection of pairs.  In the formulas below,  we treat $u,v$ as being implicitly defined by $w$. 
Given a Borel set $J \subset \bbR_{+}$,  let $\Si_{J}^m$ be the set of collections of $m$ pairs $\{u_i,v_i \}_{i=1, \ldots,m}$ with $u_i, v_i\in J$ and satisfying the conditions \eqref{mpairs}. 
Let $\mu_m(d\uw)$ be  the Lebesgue measure on the corresponding subset - a simplex - of $J^{2m}$. We can then
 rewrite
\baq \label{eq: second duhamel series1}
  \mathop{\int}\limits_{0< t_1 < \ldots < t_{2m} <t} d t_1 \ldots \d t_{2m} \, \,  \sum_{\pi \in \textrm{Pair}(t_1,\ldots,t_{2m})} \caT\left[ \mathop{\opprod}\limits_{\{u,v \} \in \pi}  K_{u,v}  \right] 
  = \mathop{\int}\limits_{\Si_{[0,t]}^m}  \mu_m(d\uw)    
   \caT\left[  \opprod_{w \in \uw} K_{w}    \right]
\eaq
where we abbreviated $K_w :=K_{u,v}$ for $w=\{u,v \}$. Let $\Si_{J}$ denote
the disjoint union  
$$\Si_{J}=\sqcup_{m=0}^\infty\Si_{J}^m, 
$$
i.e.\  $\uw\in\Si_{J}$ is
given by $\uw=(m,\uw^m), \ m\in{\bbN}, \ \uw^m\in\Si_{J}^m$
with the convention
$\Si_{J}^0=\{\emptyset\}$ and we write $\str \uw \str :=m $. Thus measurable functions $F$ on $\Si_{J}$ are
collections $\{F_m\}_{m\in\bbN}$ of measurable $F_m$ on $\Si_{J}^m$. We let
$ \mu(\d \uw)$ be 
the measure on $\Si_{J}$ given by
\beq\label{mu0def}
  \mathop{\int}\limits_{\Si_{J}} \mu(\d \underline{w}^{})F(\uw) :=
  \sum_{m \in \bbN}  \mathop{\int}\limits_{\Si_{J}^m}  \mu_m(d\uw) F_m(\uw) 
  \eeq
where we set $ \mu_0((0,\emptyset))=1$.  
Note that the elements of $\Si_J$ are naturally interpreted as sets, i.e.\
we write $w \in \uw$ to mean $w \in \uw^m$ for $m=\str \uw \str$ and $\uw=\uw' \cup \uw'' $ for the element in $\Si^{m}_J$ with $m=\str \uw'\str \cup \str \uw''\str$ with $w \in \uw$ whenever $w \in \uw'$ or $w \in \uw''$. 
% \beq\label{eq: product property}
%  \mathop{\int}\limits_{\Si_{J}} \mu(\d \underline{w}^{})F(\uw) =     \mathop{\int}\limits_{\mathrm{Ran}(I)}  \mu(\d \underline{w}')  \mu(\d \uw'') F(\uw' \cup \uw') 
%  \eeq
%\beq\label{eq: product property}
%  \mathop{\int}\limits_{\Si_{J}} \mu(\d \underline{w}^{})F(\uw) \leq     \mathop{\int}\limits_{\Si_{J}} \mu(\d \underline{w}^{})  \mathop{\int}\limits_{\Si_{J}} \mu(\d \uw') F(\uw \cup \uw') 
%  \eeq
With these conventions the Dyson expansion \eqref{eq: second duhamel series}  becomes
\beq  \label{eq: third duhamel series}
\e^{\i t L_\sys}  Q_t    =%\sum_{m \in \bbN}  \mathop{\int}\limits_{\Si_{[0,t]}^m}  \mu_m(d\uw)    \caT\left[  \bigotimes_{w \in \uw} K_{w}    \right]:= 
\mathop{\int}\limits_{\Si_{[0,t]}}  \mu(\d \uw)   \caT\left[  \mathop{\opprod}\limits_{w \in \uw} K_{w}    \right]
\eeq
where, for $\uw=\emptyset$ the integrand is defined to be $1$.

\subsection{Correlation functions and the Dyson series} \label{sectreegraph}
\subsubsection{The contraction operator $\caT_A$}

To each macroscopic time $ \tau \in I_n$, we now associate a {domain} of microscopic times,
\beq
\Dom  ( \tau) = 
  [  \la^{-2} (\tau-1),   \la^{-2}  \tau ] 
  \eeq
 To a set $A \subset I_n$ of macroscopic times, we then associate the  domain 
 \beq
 \Dom (A)= \bigcup_{\tau \in A} \Dom  ( \tau)
 \eeq
The contraction operator $\caT[\cdot]$ defined in Section \ref{secexpansions} contracts operators so as to produce an operator in $\scrR$. We now define a contraction operator $\caT_A$ that produces operators in $\scrR_A$.
Let us first consider a finite family of operators $V_{t_i} \in \scrR_{t_i} $ where the indexed times $t_i$ satisfy $t_i < t_{i+1}$ and $t_i \in \Dom(A)$. Then we set
\beq
\caT_{A} \left[  \mathop{\otimes}\limits_{i }  V_{t_i}   \right]   :=         \mathop{\otimes}\limits_{\tau \in A}  \bsI_{\tau} \left[  \caT \left[
\mathop{\otimes}\limits_{j:  t_j  \in  \Dom  ( \tau)} V_{t_j} \right] \right]
\eeq 
and we extend by linearity to the whole of $\otimes_i\scrR_{ t_i }$, obtaining   $\caT_{A}: \otimes_i\scrR_{ t_i } \mapsto \scrR_A$. In words, $\caT_{A}$ puts each operator into the right 'macroscopic' time-copy and contracts the operators within each macroscopic time-copy. 

\subsubsection{The graph $\caG(\uw)$}

 A set of pairs $\uw \in \Si_{  \Dom  (A)}$ determines a graph $\caG(\uw)$ on $I_n$ by the following prescription: the vertices $\tau< \tau'$ are connected by an edge iff.\ there is at least one pair $w=\{ u,v \}$ in $\uw$ such that 
\beq
u \in   \Dom  ( \tau)  \qquad  \textrm{and} \qquad     v \in   \Dom  ( \tau') 
\eeq 
We write $\supp(\caG(\uw))$ for the set of non-isolated vertices of $\caG(\uw)$, i.e.\ the vertices that belong to at least one edge.  If $\uw \in \Si_{  \Dom  (A)}$ than  $\supp(\caG(\uw))$ is  a subset of $A$.  In that case, we write $\caG_A(\uw)$ for the induced  subgraph on $A$.

Let us denote the free $\sys$-evolutions  
\beq
Y_\tau =   \bsI_{\tau}[\e^{\i (\tau-1) L_\sys}], \quad Y_A =  \mathop{\otimes}\limits_{\tau \in A}Y_\tau, \qquad \text{and}  \qquad   \widetilde Y_\tau =    \bsI_{\tau}[\e^{-\i \tau L_\sys}], \quad    \widetilde Y_A =  \mathop{\otimes}\limits_{\tau \in A}  \widetilde Y_\tau
\label{eq: y operators}
\eeq 
We are now ready to state the connection between the Dyson expansion and the correlation functions $G_A, G^c_A$:
For $A \subset  I_n$ with $\str A \str \geq 2$, 
\baq
  \widetilde Y_A G_A Y_A   & = &  \mathop{\int}\limits_{\Si_{\Dom(A)}}  \,  \mu (\d \uw)   \,   \,      \indicator_{[  \supp(\caG(\uw)) =A]} \,     \caT_{A}   \left[ \mathop{\opprod}\limits_{w \in \uw} K_w   \right]    \label{eq: correlation micro} \\[2mm]
 \widetilde Y_A G^c_A  Y_A   &=&   \mathop{\int}\limits_{\Si_{\Dom(A)}} \,  \mu (\d \uw)   \,   \,      \indicator_{[   \caG_A(\uw)  \, \textrm{connected}  ]}  \, 
   \caT_{A}  \left[  \mathop{\opprod}\limits_{w \in \uw} K_w    \right]. \label{eq: connected micro}  \eaq
 To check this, it is helpful to note first  (with $A=\{\tau_1, \ldots, \tau_m\}$)
 \beq
\widetilde Y_A \bbE [U_{\tau_m} \otimes_\sys \ldots \otimes_\sys U_{\tau_1} ]  Y_A =     \mathop{\int}\limits_{\Si_{\Dom(A)}}  \,  \mu (\d \uw)   \,   \,       \caT_{A}  \left[ \mathop{\opprod}\limits_{w \in \uw} K_w   \right]
 \eeq
 which is a straightforward generalization of \eqref{eq: third duhamel series}.  Then, one observes that all pairs $\{ u,v\}$ with  $u,v \in\Dom(\tau)$ make up $ \bsI_\tau[T] $, and hence the operators $G_A$ are integrals over all those $\uw$ such that for each $\tau \in A$, the graph $\caG(\uw)$ has an edge $\{ \tau, \tau' \}$ containing $\tau$. This leads to \eqref{eq: correlation micro}.  Then, one verifies that $G^c_A$ as given by \eqref{eq: connected micro} satisfies  \eqref{def: cumulants}.  Since \eqref{def: cumulants}  fixes $G^c_A$ uniquely,  \eqref{eq: connected micro} is thereby proven.   
 
 \subsubsection{From pairings to  $G^{c}_A$: Pictorial representation}
 We divide the time-axis into blocks  (intervals of the form $\la^{-2}[\tau-1, \tau] $ with $\tau \in \bbN$)). We draw pairings as arcs connecting two times on the axis, either solid \includegraphics[width = 0.8cm, height=0.5cm]{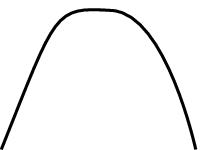}  or dotted \includegraphics[width = 0.8cm, height=0.5cm]{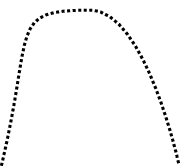}, the reason to draw some of them differently is addressed below.  The two times in a pairing, $ \{ u,v\}$, are denoted as 'legs'.
 This yields the upper picture in Figure \eqref{fig: splittingofmicro}.  

Then, we single out blocks which \textbf{do not} contain a leg of a pair whose other leg lies in another block.  Hence, such a block $\tau$ can contain no pairing at all, or one pairing, or two pairings, etc... In the language of the graph $\caG(\uw)$, this means that the time $\tau$ has no edge to any other time; it is an isolated vertex.   
Such 'isolated blocks' are now, in the middle picture, drawn as   \includegraphics[width = 0.8cm, height=0.2cm]{stukjeT.eps} and we omit the pairs on them (imagining that they have been resummed and \includegraphics[width = 0.8cm, height=0.2cm]{stukjeT.eps} stands for the sum). 
Blocks which are not isolated are drawn as \includegraphics[width = 0.8cm, height=0.2cm]{stukjeG.eps}, but we still indicate the pairs on them.  It is now clear that the pairings that were drawn with a dotted line in the above picture, are those inside an isolated vertex and only the other ones (drawn with a solid line) are reproduced in the middle picture. 

Finally, to arrive at the lower picture, we resum the pairings on the blocks \includegraphics[width = 0.8cm, height=0.2cm]{stukjeG.eps} corresponding to connected components of the graph $\caG(\uw)$. That, is whenever two blocks are connected by a pairing, we call them connected, and this induces a partition of the blocks \includegraphics[width = 0.8cm, height=0.2cm]{stukjeG.eps} into connected components.  Only this partition is indicated in the lower figure. Blocks belonging to the same (connected component) determine one polymer $A$.  Alternatively, the operator $G^c_A$ is determined by summing all sets of pairs $\uw$ that contribute to the connected component $A$.

\begin{figure}[h!] 
\vspace{0.5cm}
\begin{center}
\def\svgwidth{\columnwidth}
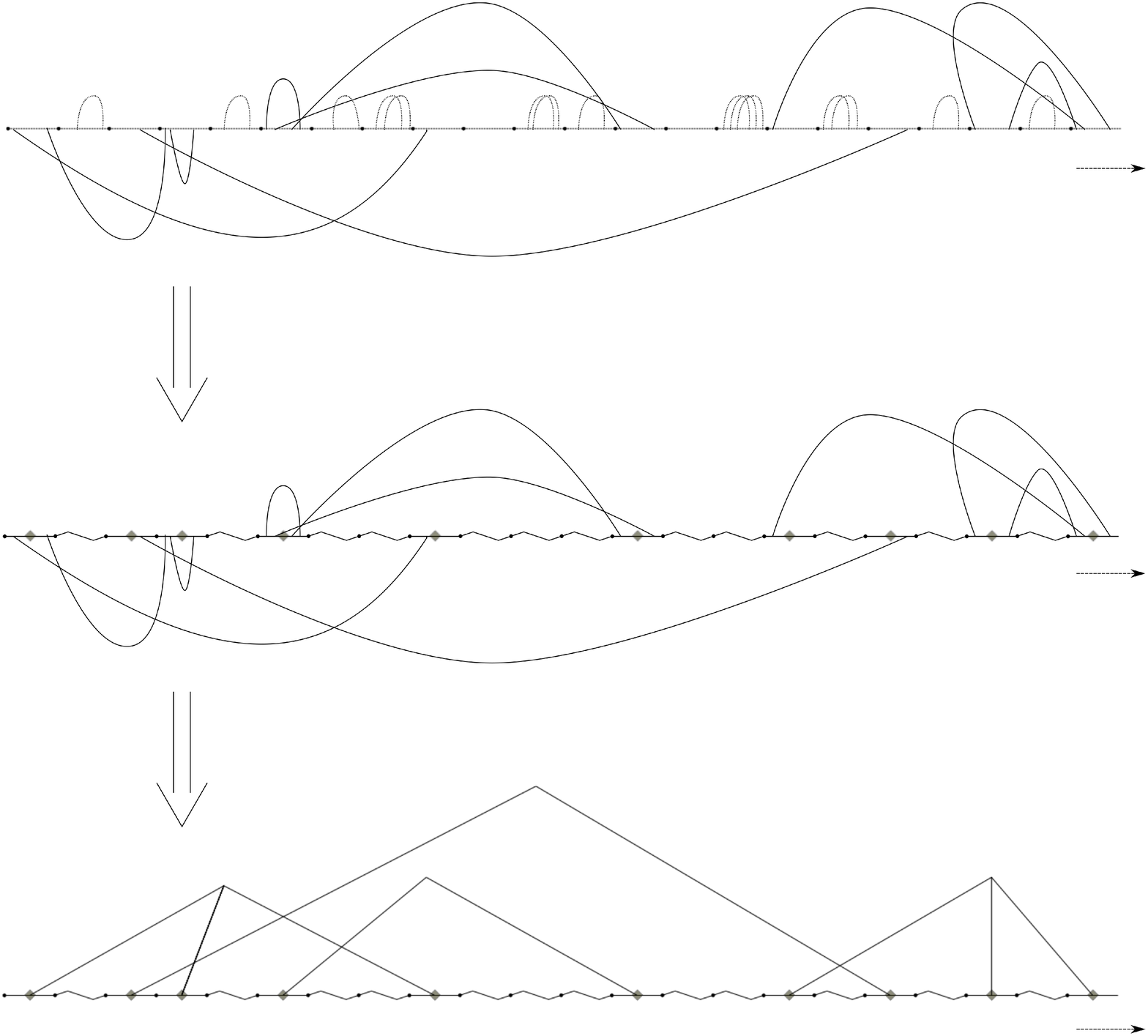
\caption{ \footnotesize{From the collection $\uw$ of pairs $\{u_i,v_i \}$ to a collection $\caA$ of polymers, $\caA=\{A_1, A_2, A_3,A_4  \}$.}
 \label{fig: splittingofmicro} 
}
\end{center}
\end{figure}

\subsubsection{Bounds on correlation functions} \label{sec: bounds on correlation functions}

Up to now, we did nothing more on the Dyson expansions than straightforward algebraic manipulations.  In what follows, we provide bounds. From the properties of the norm $\norm \cdot\normw$ discussed in \ref{sec: norms}, we get 
\beq \label{eqbasic bound on correlations bulk}
\norm G^c_A \normw = \norm   \widetilde Y_A G^c_A  Y_A \normw   \leq  \mathop{\int}\limits_{\Si_{\Dom(A)}}  \,  \mu (\d \uw)   \,    \prod_{w \in \uw} \norm K_{w} \normw
\eeq
To continue, we recall that  $\norm K_{w} \normw \leq   \la^2 C \str h(v-u)\str $ (Section \ref{secexpansions}) and
 that  $\int_0^\infty \d s (1+s)^{\al}\str h(s)\str < \infty$ by  Assumption \ref{ass: decay correlation functions}.  
 We say that `$\uw$ spans $A$ minimally' if the graph $\caG_A(\uw)$ is connected and no pair $w$ can be omitted from $\uw$ without losing this property.  In particular, if $\uw$ spans $A$ minimally, than  $\caG_A(\uw)$ is a spanning tree.  
\begin{lemma} \label{lem: a priori} Let  $\norm h \norm_1 := \int_0^\infty \d s \str h(s)\str$.   For $A \subset I_n$, 
\beq  \label{eqminimal spanning bound}
 \norm G^c_A  \normw   \leq    \e^{ C  \norm h \norm_1   \str A \str  } \int_{\Si_{ \Dom  (A)} }\mu(\d \uw)  \, \,   \indicator_{[ \underline{w} \, \textrm{spans} \, A\, \textrm{minimally}]}  \,      \left( \prod_{w \in \uw}  \la^2 C | h(v-u)|  \right) \,   
\eeq

\end{lemma}
%The following estimate is a  consequence  this integrability property and  \eqref{eq: connected micro} and   the bound \eqref{eqbasic bound on correlations bulk}.

\begin{proof}
We start with an appealing estimate that was the main motivation for encoding the pairings $\pi$ in the pair-sets $\uw$.
 For integrable $F$, 
\beq
 \mathop{\int}\limits_{  \Si_{ \Dom  (A)}}     \mu (\d \uw)\, \indicator_ {[\caG_A(\uw)  \, \textrm{connected} ]}
    \,  \str F(\uw) \str  \quad   \leq \quad   \mathop{\int}\limits_{ \Si_{ \Dom  (A)}  }   \,  \mu (\d \uw')  \,
 \indicator_{[ \uw'\,  \textrm{spans} \, A\, \textrm{minimally}  ]}   
      \,      \mathop{\int}\limits_{   \Si_{ \Dom  (A)} }   \,  \mu (\d \uw'')    \,      \str F(\uw' \cup \uw'') \str  \label{eq: minimal split off}
\eeq
  To realize why this holds true,  choose a spanning tree $\scrT$ for the connected graph $\caG_A(\uw)$, pick a minimal subset $\uw'$ of the collection $  \uw$ such that $\caG_A(\uw')=\scrT$ and use that
  \beq \label{eq: product property}
  \mu( \d (\uw'\cup \uw'')) =   \mu( \d\uw' ) \mu(d \uw''). 
  \eeq

  We apply the inequality  \eqref{eq: minimal split off} to \eqref{eqbasic bound on correlations bulk}  with $F(\uw) = \prod_{w \in \uw} \norm K_{w} \normw$ and we estimate the $\uw''$-integral as follows: 
  \baq
  &&  \mathop{\int}\limits_{\Si_{\Dom(A)}}  \,  \mu (\d \uw'')   \,    \prod_{w \in \uw''}  \la^2 C \str h(v-u) \str  \nonumber  \\
  &\leq & \sum_{m \in \bbN}  \mathop{\int}\limits_{u_1 < \ldots < u_{m}, u_i \in \Dom(A) } d \underline{u} \,   \left( \prod_{i=1}^m  \int_{u_i}^t \d v_i  \,  \la^2 C \str h(v_i-u_i) \str \right)  \nonumber  \\
  &  \leq &   \sum_{m \in \bbN}   \frac{ (\la^2  C \str \Dom(A)\str \norm h \norm_1)^m}{m!}   \leq   \e^{\la^2 C \str   \Dom(A)\str  \norm h \norm_1}
  \label{eq: unconstrained estimate}
  \eaq
  and we conclude, since  $\str \Dom(A)\str = \la^{-2} \str A \str $. 
\end{proof}
In what follows, we no longer trace explicitly the dependence on $h$ and we simply estimate $\norm h \norm_1 \leq C$.

\subsection{Proof of bound \eqref{eqbound operator polymers} in  Lemma \ref{lem: bound on primary polymers}}  \label{sec: proof of bound on operator polymers}
In this section, we prove \eqref{eqbound operator polymers} of Lemma \ref{lem: bound on primary polymers}.
We start from Lemma \ref{lem: a priori}.
For each  $\uw$ that spans $A$ minimally, $\caG_A(\uw)$ is a spanning  tree on $A$. Hence we can reorganize the bound \eqref{eqminimal spanning bound} by first integrating all $\underline{w}$ that determine the same spanning tree $\scrT$. This amounts to integrate, for each edge of the tree, all pairs $\{u,v\}$ that determine this edge. Hence we arrive at the bound 
\beq  \label{eqspanning trees bound}
 \norm G^c_A  \normw  \leq     \e^{C \str A \str}     \,   \sum_{\textrm{trees} \, \scrT  \, \textrm{spanning} \, A}  \,\, 
\prod_{  \{\tau,\tau' \} \in \caE(  \scrT) }    \hat e(\tau,\tau')   
\eeq
where $\caE(  \scrT)$ is the set of edges of the tree $\scrT$ and, for $\tau' >\tau$, 
\baq \label{def: edge factors}
 \hat{e}(\tau,\tau') &:=&   \mathop{\int}\limits_{  \la^{-2} [\tau-1, \tau]} \d u   \mathop{\int}\limits_{  \la^{-2}  [\tau'-1, \tau']}  \d v \,  \,  \la^2C \str h(v-u)\str, \eaq
and  $ \hat{e}(\tau',\tau)  :=  \hat{e}(\tau,\tau') $.  Let   $\Delta \tau \equiv \tau'-\tau$ and  $s \equiv v-u$, then 
\baq
 \hat{e}(\tau,\tau')  &\leq &   C      \mathop{\int}\limits_{  (\Delta \tau-1)/\la^2}^{(\Delta \tau+1)/\la^2}   \d s | h(s)|, \qquad       \Delta \tau >1  \\[1mm]
  \hat{e}(\tau,\tau')    & \leq&  \la^2C  \int_{0}^{2/\la^2} \d s   \, s \str h(s)\str , \qquad       \Delta \tau =1 
\eaq
and, since  $\int_0^\infty \d s (1+ \str s \str)^{\al} \str h(s)\str  \leq C$, this implies  the bound 
\beq  \label{eqkotecky preiss for pairings}
\sum_{\tau':  \tau' \neq \tau}    \hat{e}_{\al}(\tau,\tau')  \leq C \ep
\eeq
where $\ep=\str \la \str^{2 \min{(\alpha,1)}}$ (cfr. Lemma \ref{lem: a priori}) and
\beq 
  \hat{e}_{\al}(\tau,\tau')  :=  ( 1+  \str \tau'-\tau \str)^{\al} \hat{e}(\tau,\tau').
\eeq  

Starting from  \eqref{eqspanning trees bound},  we bound
\baq
  \sum_{A: A \ni \tau_0} (C\ep )^{{ -(\str A \str-1) }}\dist(A)^{\al} \norm G_A^c \normw & \leq &  \sum_{A: A \ni \tau_0}  (C\ep )^{{ -(\str A \str-1) }}  \,    \sum_{\textrm{trees} \, \scrT  \, \textrm{spanning} \, A}  \, \,
\prod_{ \{\tau,\tau' \} \in \caE(\scrT) }  \hat e_{\al}(\tau,\tau')  \nonumber  \\[2mm] 
& \leq &    \sum_{\textrm{trees} \, \scrT:  \tau_0 \in \caV(\scrT)  }  (C \ep)^{-\str \caE(\scrT) \str}  \,     
\prod_{  \{\tau,\tau' \} \in \caE(\scrT) }  \hat e_{\al}(\tau,\tau') 
  \label{eqsum over trees}
\eaq
where  $\caE(\scrT), \caV(\scrT)$ are the edge, resp.\ vertex set of the tree $\scrT$.
To obtain the first inequality, we used that $\prod_{\{ \tau, \tau'\}  \in \caE(\scrT)} (1 + \str \tau' -\tau\str)^{\al}  >  \dist(A)^{\al}  $ for any spanning tree on $A$.
The sum over trees in \eqref{eqsum over trees} is estimated with the help of Lemma \ref{lemma: combi trick abstract} by choosing
\ben
\item the polymers $S$ as unordered pairs $\{ \tau, \tau' \}$ and $a(S)=1$. 
\item the adjacency relation $S \sim S'$ iff.\  $S \cap S' \neq \emptyset$.
\item the polymer weights $w(\{ \tau, \tau' \})= (C \ep)^{-1} \hat e_\al(\tau,\tau')$ 
\een 
The bound  \eqref{eqkotecky preiss for pairings} plays the role of the Kotecky-Preiss criterion \eqref{eqkotecky preiss abstract} and  Lemma \ref{lemma: combi trick abstract} yields
\beq
 \sup_{\tau_0}  \sum_{\textrm{trees} \, \scrT:  \tau_0 \in \caV(\scrT)  }  (C\ep)^{-\str \scrE(\scrT) \str}   \,  
\prod_{  \{\tau,\tau' \} \in \caE(\scrT) }  \hat e_{\al}(\tau,\tau')   \leq   1  \label{eq: sum over trees bounded one}
\eeq
For example, replace the restriction $\tau_0 \in \caV(\scrT)$ by the weaker restriction  $\scrE(\scrT) \sim \{\tau_0,\tau_1 \} $ for some arbitrary $\tau_1$ and use
 \eqref{eqbound on clusters touching something abstract}.
Upon combining this with \eqref{eqsum over trees}, the bound \eqref{eqbound operator polymers} of Lemma \ref{lem: bound on primary polymers} is proven. 
\subsection{Analysis of $T$: dissipativity and weak coupling limit}  \label{sec: analysis of t}
Recall that operator $T$, acting on $ \scrB_1(\scrH_\sys)$, that was introduced in 
Section \ref{sec: splitting of the dynamics}
 \beq
T \rho_{\sys}  =  \Tr_\res [  \e^{- \i \la^{-2} L_{\ka}} (\rho_{\sys} \otimes P_{\Om}) ]
 \eeq
We exhibit the dissipative properties  of $T$  and we  establish Statement 1) of Lemma \ref{lem: bound on primary polymers}.

\subsubsection{Construction of Lindblad generator $M$ } \label{sec: construction of lindblad generator}
To ease the presentation, we introduce some shorthand notation. 
Recall the one-dimensional spectral projections $P_e$ corresponding to  $H_\sys$ (see the discussion  preceding Assumption \ref{ass: Fermi Golden Rule}).  The set  of differences of eigenvalues will be called $\caE= \{ \ve= e-e',  e,e' \in \si(H_\sys) \}$. Note that it is the set of eigenvalues of the Liouvillian $L_\sys$. 
Let us for the moment assume that for any $\ve \in \caE, \ve \neq 0$, there is a unique pair $(e,e'), e,e' \in \si(H_\sys)$ such that  $\ve=e-e'$.  This assumption allows for a more explicit treatment, we indicate at the end of this section how to relax it. 
Relying on this assumption, we  abbreviate $D_\ve \equiv P_e D P_{e'}$ and we note that $D_{\ve}^*=D_{-\ve}$ since $D=D^*$.
We define the operator  $M $, acting on $\scrB_1(\scrH_\sys)$
\beq  \label{def: lindblad}
M \rho =   - \i [H_{\mathrm{Lamb}}, \rho] + \sum_{\ve \in \caE: \ve <0} \hat h(-\ve)   \left( \e^{\ka} D_{\varepsilon} \rho D^*_{\varepsilon} - \frac{1}{2}  \{  D^*_{\varepsilon}  D_{\varepsilon} , \rho \} \right)
\eeq
with $\hat h(\ve)$ as in \eqref{eqexpression jump rates}, $\{ A,B\}=AB+BA$,  and the energy-shift (`Lamb-shift') operator 
\beq
H_{\mathrm{Lamb}}=  \sum_{\ve \in \caE: \ve \neq 0}  \left( \Im \int_0^{\infty} \d s \,  \e^{\i s \ve} h(s) \right)   D^*_{\varepsilon} D_{\varepsilon},
% \qquad   v(\ve) := \Im \int_{\bbR_+}\d s   \int_{\bbR^d} \d q  \,  \e^{\i s (\str q \str-\ve)}   \str \phi(q)\str^2.
\eeq
which is 
well-defined by  Assumption \ref{ass: decay correlation functions}).

Because of our strong non-degeneracy assumptions, we can greatly simplify the form of $M$. Let us choose a basis for $\scrB(\scrH_\sys)$ consisting of the one-dimensional projectors $P_e$ and the operators  $\str \psi_e \rangle \langle \psi_{e'} \str$, with $e \neq e'$, where $\psi_{e \in \sp (H_\sys)}$ are eigenfunctions of $H_\sys$.  Then, inspection of \eqref{def: lindblad} and the definition of the jump rates $j(\cdot,\cdot)$  in \eqref{eqexpression jump rates} yields the following claims:\\
\textbf{Off-diagonal elements} 
 The operators $\str \psi_e \rangle \langle \psi_{e'} \str$  are eigenvectors of $M$ with corresponding eigenvalue
\beq
- \i  \langle \psi_e,  H_{\mathrm{Lamb}}\psi_e \rangle  - \i  \langle \psi_{e'},  H_{\mathrm{Lamb}}\psi_{e'} \rangle -  (1/2) \sum_{e''} \left( j(e',e'')+  j(e,e'')\right)
\eeq
Assumption \ref{ass: Fermi Golden Rule} implies that the second term is bounded away from $0$ and hence one has
$$\sup_{e \neq e'} \norm \e^{t M} (\str \psi_e \rangle \langle \psi_{e'} \str) \norm  \leq \e^{-c t} $$ 
 with $c>0$ . 
This corresponds physically to decoherence.  \\
\textbf{Diagonal elements} 
The space spanned by the projectors $P_e$ is mapped into itself and we can identify the action of $M$ on this space with that of a $ d_\sys \times d_\sys$-matrix $\caM$ (with $d_\sys= \dim \scrH_\sys$) by setting $M(\sum_e \mu(e) P_e)= \sum_e ((\caM\mu)(e)) P_e$ for $\mu \in \bbC^{\d_\sys}$. Then
\beq
(\caM \mu)(e) = \sum_{e'} ( \e^{\ka} j(e',e) \mu(e')-    j(e,e') \mu(e) )   \label{eq: cal m}
\eeq
It is clear that, for $\ka=0$,  $\caM$ is the forward generator of a Markov process with state space $\si(H_\sys)$.  We determine the spectrum of $\caM$.   Recall  from the discussion following Assumption \ref{ass: decay correlation functions} that $j(e',e)  = 0$ whenever $e \geq e'$, hence $\caM$ is a triangle matrix and its singular values are the diagonal elements.  By inspection of \eqref{eq: cal m} and  Assumption  \ref{ass: Fermi Golden Rule}, it follows that all singular values lie in the region $ \{ z: \Re z < -g_{\caM} \}$ with $g_\caM = \min_{e: e \neq e_0} \sum_{e'} j(e,e')  $ save for a simple eigenvalue $0$, corresponding to the eigenvector $\mu(e)= \delta_{e,e_0}$. 
By the spectral mapping theorem for semigroups with bounded generator, we then get that  $\si (\e^{t \caM})  \subset \{ 0\} \cup \{z:  \str z \str \leq \e^{-t g_\caM} \}$

  Note that the right eigenvector does not depend on $\ka$, but the left eigenvector  does. 
    The Markov process at $\ka=0$ is called absorbing, with absorbing state $e_0$. Physically, this property is due to the fact that the field $\res$ is in the vacuum state and cannot excite the atom $\sys$.    
We summarize the discussion:
\begin{lemma} \label{lem: spectral analysis of m} Assume that for any $\ve \in \caE, \ve \neq 0$, there is a unique pair $(e,e'), e,e' \in \si(H_\sys)$ such that  $\ve=e-e'$. Then, 
The ground-state projection $P_{e_0}$ is an eigenvector of $\e^{tM}$ with simple eigenvalue $1$.  The rest of the spectrum of $\e^{tM}$ lies in the region $\{z:  \str z \str \leq \e^{-t g_M} \}$ for some $g_M >0$ which does not depend on $\ka$. 
\end{lemma}

Finally, if the extra non-degeneracy condition that we introduced at the beginning of this section does not hold, i.e. if there are several pairs $(e,e'), e,e' \in \si(H_\sys)$ such that  $\ve=e-e'$ for $\ve \neq 0$, then  we still define the operator $M$ by \eqref{def: lindblad} and we note that $M$  still preserves the subspace of diagonal and off-diagonal density operators,  i.e.\ we can write $M=\caM \oplus M_{\mathrm{o}}$ and for any density matrix $\rho$, define
\beq
\mu(e):= \Tr P_e \rho, \qquad \rho_{\mathrm{o}} := \rho -\sum_{e} \mu(e) P_e,
\eeq
such that $\rho = \mu \oplus \rho_{\mathrm{o}}$, then
\beq
\e^{t M}\rho =   \e^{t \caM} \mu   \oplus  \e^{t M_{\mathrm{o}} } \rho_{\mathrm{o}}
\eeq
where $ \e^{t M_{\mathrm{o}}} \rho_{\mathrm{o}}$ is off-diagonal, i.e.\ $P_e  (\e^{t M_{\mathrm{o}}} \rho_{\mathrm{o}}) P_e=0$ for any $e \in \si(H_\sys)$. 
The spectral analysis of $\caM$ given above remains unchanged and hence it follows that $\e^{t \caM} \mu \to P_{e_0}$ exponentially fast. For $\ka \in \bbR$, the semigroup $\e^{t M}$ is positivity preserving: from  the positivity of $\e^{t M}\rho$,  combined with the fact that $ P_{e_0}$ is one-dimensional, we deduce that $\e^{t M_{\mathrm{o}} } \rho_{\mathrm{o}}  \to 0$ exponentially fast, independently of $\ka \in \bbR$.  Therefore, for $\ka \in \bbR$, Lemma \ref{lem: spectral analysis of m} remains true without the extra non-degeneracy assumption and for sufficiently small $\ka\in \bbC$, it follows  by perturbation theory of isolated eigenvalues. 

\subsubsection{The Lindblad generator $M$ and the microscopic model} \label{sec: lindblad generator and micro}

The relation between the the semigroup generated by the Lindblad generator $M$ and the microscopic model is a classical result in mathematical physics.  It was discussed already in Section  \ref{sec: plan of the proof}.
\begin{proposition}\label{prop: weak coupling}
For any $t >0$, 
\beq
\norm Q_{t}  -  \e^{-\i t L_{\sys} + \la^2 t M }  \norm  \leq  C \e^{C \la^2 t}     \str\la\str^{2\al_*}
\eeq
where $\al_*=\min(1,\al) $ for $\al \neq 1$, and $  \str\la\str^{2\al_*}=\la^2 \str\log \str \la \str \str $ for $\al=1$. Using that $ L_{\sys}$ commutes with $M$, we get
\beq
\lim_{\la\to 0} \e^{\i \la^{-2}\frt L_{\sys} }  Q_{\la^{-2}\frt}   =   \e^{ \frt M } 
 \eeq
for any $\frt>0$.
\end{proposition}
For completeness, we review the simple proof of this convergence in Appendix \ref{secweak coupling limit}.     We now turn to the

\begin{proof}[Proof of Statement 1) of Lemma \ref{lem: bound on primary polymers}] Note
\begin{itemize}
\item[$i.)$]   $L_\sys$ commutes with and $M$. 
\item[$ii.)$]   $\e^{\i t L_\sys}$ is an isometry on $\scrB_1{(\scrH_\sys)}$.
\item[$iii.)$] $L_\sys(P_e)=0$ for any $e \in \si(H_\sys)$, in particular $e_0$.
\end{itemize}
Therefore,  the spectral analysis of $M$ (Lemma \ref{lem: spectral analysis of m}) and the discussion following it implies that 
$\e^{-\i \la^{-2} L_\sys + M}$ has a simple eigenvalue $1$ and all other spectrum lies inside a circle of radius $ \e^{-g_M} $. 
On the other hand, recall that $T= Q_{\la^{-2}}$ and hence by Proposition \ref{prop: weak coupling}, $T- \e^{-\i \la^{-2} L_{\sys} +  M }  $ vanishes as $\la \to 0$. We apply spectral perturbation theory of isolated eigenvalues to conclude that, for sufficiently small $\str\la\str,\str\ka\str$, $T$ has an isolated eigenvalue $\e^{\theta_T}$ with corresponding spectral projector $R$ such that statement 1) of Lemma \ref{lem: bound on primary polymers} holds.   Since $T$ conserves the trace for $\ka=0$, it follows that  $\theta_T(\la,\ka)$ indeed vanishes for $\ka=0$; as $\la \to 0$, it reduces to the eigenvalue of $M$, hence to $0$, as well.  Since $T$ preserves positivity and the trace
 for $\ka=0$, we can choose  $\tilde \eta=\lone$ and $\eta$ a density matrix, where $\eta, \tilde\eta$ were defined as the right, resp.\ left  eigenvectors of $T$ (see beginning of Section \ref{sec: scalar polymer model}).  For $\ka \in \bbR$, the operator $T$ remains positivity-preserving and therefore $\eta,\tilde\eta$ can be chosen to be positive operators.
\end{proof}
For later use, we note that this perturbation argument also yields the bounds, for sufficiently small $\str \la \str, \str\ka\str$, 
\beq   \label{eq: bound eigenvectors}
 \norm \tilde\eta -\lone \norm \leq \caO(\str\ka\str), \qquad   \norm \eta - P_{e_0} \norm \leq C   \str\la\str^{2\al_*}.
\eeq

\subsection{Boundary polymers} \label{sec: boundary polymers}

\subsubsection{Dyson expansion}
We aim to write the analogue of the expansion \eqref{eq: third duhamel series} for the expression of $ \breve{Q}_{n}  \rho_{\sys}  $ given in \eqref{modifieddynamics}. Recalling the definition \eqref{Phikappa}, we 
 may write $U_\realinitial$  in the same form as \eqref{eq: duhamel on superspace}:
 \baq
 U_\realinitial  (\rho_\sys\otimes P_{\Om}) &=& \rho_\sys\otimes \e^{\i\Phi_{\ka}(\psi_\realinitial,0)}P_{\Om}\e^{-\i\Phi_{-\ka}(\psi_\realinitial,0)}\nonumber\\
 &=& \sum_{m \in \bbN}(-\i )^m  \mathop{\int}\limits_{-1\leq t_1 < \ldots < t_m\leq 0} \d t_1 \ldots \d t_m \,      L_{\inter, \ka} (t_m)  \ldots  L_{\inter, \ka} (t_2) L_{\inter, \ka} (t_1)  \rho_\sys\otimes P_{\Om}
  \eaq  
  with the definition
\baq\label{Linit}
  L_{\inter,\ka}(s) =      - \caL ( \lone  \otimes \Phi_{\ka}(\psi_\realinitial,0) )
  +  \caR ( \lone  \otimes\Phi_{-\ka}(\psi_\realinitial,0) ), \qquad s\in [-1,0).  
\eaq
In the same way we get 
 \beq
 U_\realfinal   =  \sum_{m \in \bbN}(-\i )^m  \mathop{\int}\limits_{n/\la^2\leq t_1 < \ldots < t_m\leq n/\la^2+1} \d t_1 \ldots \d t_m \,      L_{\inter,\ka} (t_m)  \ldots  L_{\inter, \ka} (t_2) L_{\inter, \ka} (t_1) 
  \eeq  
where
  \baq\label{Lfinal}
  L_{\inter,\ka}(s) =      - \caL ( \lone  \otimes \Phi(\e^{\i n \str q \str/\la^2}\psi_\realfinal,0) ), \qquad s\in (n/\la^2,n/\la^2+1)].  
\eaq
Thus we end up with a similar series as in \eqref{eq: first duhamel series} and in \eqref{eq: second duhamel series}:
\baq% \label{eq: first duhamel series hat}
\e^{\i n/\la^2 L_\sys} \breve Q_{n} \rho_\sys&=&     \sum_{m \in \bbN}  (-1
)^{m}\mathop{\int}\limits_{-1\leq t_1 < \ldots < t_{2m} <n/\la^2+1} \d t_1 \ldots \d t_{2m} \, \,  \Tr_{\res}\left[  L_{\inter,\ka}(t_{2m})  \ldots  L_{\inter,\ka}({t_2})  L_{\inter,\ka}({t_1})  (\rho_\sys\otimes P_{\Om})\right] \nonumber\\
&=&     \sum_{m \in \bbN}  \mathop{\int}\limits_{-1\leq t_1 < \ldots < t_{2m} \leq n/\la^2+1} d t_1 \ldots \d t_{2m} \, \,  \sum_{\pi \in \textrm{Pair}(t_1,\ldots,t_{2m})} \caT\left[ \mathop{\opprod}\limits_{\{u,v \} \in \pi}  K_{u,v}  \right] 
\nonumber\\
&=&  
\mathop{\int}\limits_{\Si_{[-1,n/\la^2+1]}} \mu(\d \uw)   \caT\left[  \mathop{\opprod}\limits_{w \in \uw} K_{w}    \right]
 \label{eq: 4th duhamel series}
\eaq
where $K_{u,v}$  is defined by the formula \eqref{Kdefi}, now with $u,v$ in $[-1,n/\la^2+1]$ instead of $[0,n/\la^2]$.
Note that the choice of  time-intervals $[-1,0]$ and $[n/\la^2, n/\la^2+1]$ is somewhat arbitrary; we do this to display the boundary terms in a way that resembles the bulk terms as close as possible. 

\subsubsection{Correlation functions}

First, we extend  the definition of the graph $\caG(\uw)$ by allowing $A \subset \breve{I}_n$ and setting
\beq
\Dom (0)  = [-1, 0], \qquad  \Dom (n+1)  = [n/\la^2, n/\la^2+1]  
\eeq
%We define the factors 
%\beq
%\tilde k_A =  \left(\indicator_{[\realinitial \notin A]}+ \indicator_{[\realinitial \in A]} T_\realinitial\right) \left(\indicator_{[\realfinal \notin A]}+ \indicator_{[\realfinal \in A]} T_\realfinal \right)
%\eeq
%and we
We generalize \eqref{eq: y operators} to $\tau =0, n+1$ by putting
\beq
Y_{\initial} =\widetilde Y_{\initial} = \bsI_{\initial} [\lone], \qquad   Y_{\final} = \bsI_{\final}  [\e^{\i (n/\la^2) L_\sys}], \qquad  \widetilde Y_{\final} =  \bsI_{\final}  [\e^{-\i (n/\la^2) L_\sys}] 
\eeq
and then defining $Y_A, \widetilde Y_A$  with $A \subset \breve I_n$ as in  \eqref{eq: y operators}.
With these definitions, the formulae \eqref{eq: correlation micro}, \eqref{eq: connected micro} hold generally for $A \subset \breve I_n$;
 \beq
  \widetilde Y_A G_A Y_A  =      \mathop{\int}\limits_{{\Si}_{\Dom(A)}}   \mu(\d \uw)     \indicator_{[\supp\caG_A(\uw) =A] } \caT_A \left[   \mathop{\opprod}\limits_{w \in \uw} K_{w}    \right] 
 \eeq
 \beq
 \widetilde Y_A  G^c_A  Y_A  =       \mathop{\int}\limits_{{\Si}_{\Dom(A)}}     \mu(\d \uw)     \indicator_{[ \caG_A(\uw)  \, \textrm{connected}  ] }       \caT_A \left[   \mathop{\opprod}\limits_{w \in \uw} K_{w}    \right]
\eeq
\subsection{Bounds on boundary polymers } \label{sec: bounds on boundary polymers}

In this section, we prove the bound \eqref{eqbound operator polymers boundary} in Lemma  \ref{lem: bound on primary polymers}. 
We first  generalize Lemma \ref{lem: a priori} to read
\begin{lemma} For $A \subset \breve I_n$
\label{lem: a priori generalized}
\baq  \label{eqminimal spanning bound boundary}
 \norm G^c_A  \normw  & \leq  &   \breve C \e^{C \str A \str }  \mathop{\int}\limits_{\Si_{ \Dom  (A)} } \mu(\d \uw)  \,      \left( \mathop{\prod}\limits_{w \in \uw }   \norm K_{w} \normw \right) \,   \indicator_{[ \underline{w} \, \textrm{spans} \, A\, \textrm{minimally}]}
\eaq
\end{lemma}
\begin{proof}
To obtain bounds on the boundary pairings, we first note that  (recall $ \Dom (I_n)=[0, n/\la^2]$)
\beq
\norm K_{u,v} \normw \leq  \left\{  \begin{array}{ll}   C  \str \la h_{\realinitial}(v)\str   & u\in \Dom(0),  v\in \Dom(I_n)   \\[4mm]
  C  \str \la h_{\realfinal}(n/\la^2-u)\str   &    u\in \Dom(I_n),  v\in \Dom(n+1)     \\[4mm]
   C  \str  h_{\Join}(n/\la^2)\str  & u \in \Dom(0) , v \in \Dom(I_n)   \\[4mm]
   C   \norm \psi_{\realinitial} \norm^2  &   u, v \in \Dom(0)   \\[4mm]
      C   \norm \psi_{\realfinal} \norm^2  &   u, v \in \Dom(n+1)   \\[4mm]
 \end{array}\right. 
  \label{eq: bound k as h} 
  \eeq

We proceed as in the proof of Lemma \ref{lem: a priori} and extract from each set of pairs $\uw$ a minimally spanning subset $\uw'$. To bound the $\uw''= \uw \setminus \uw'$-integral in \eqref{eq: minimal split off}, we need an estimate on   
\beq
\int_{\Si_{\Dom(A)}} \mu (\d \uw)    \prod_{w \in \uw} \norm K_w \normw.  \label{eq: apriori boundary}
\eeq  
Denote, for $w=\{ u,v\}$, 
\baq
\chi_{\realinitial}(w)  &:= &   1_{u \in \Dom(0)} 1_{v \notin  \Dom(n+1)}   \nonumber\\
\chi_{\realfinal}(w)  &:= &   1_{u \notin \Dom(0)} 1_{v \in \Dom(n+1)}   \nonumber\\
\chi_{\Join}(w)  &:= &   1_{u \in \Dom(0)} 1_{v \in \Dom(n+1)} 
%\indicator_{\bulk}(w)  &= &   1_{u \in J_{\bulk}} 1_{v \in J_{\bulk}} 
\eaq
and $\chi_{\iota}(\uw) := \prod_{w \in \uw} \chi_{\iota}(w)$ for $\iota=\realinitial,\realfinal,\Join$. 
For any $\uw \in \Dom(A)$, we can write $\uw = \uw' \cup \uw_{\realinitial} \cup \uw_{\realfinal} \cup \uw_{\Join} $ such that $\uw' \in \Dom(A \cap I_n)$ and $\chi_{\iota}(\uw_{\iota})=1$. Hence \eqref{eq: product property} implies 
\beq
\eqref{eq: apriori boundary}
\leq
\left(\mathop{\int}\limits_{\Si_{\Dom(A \cap I_n)}} \mu (\d \uw)   \prod_{w \in \uw}  \norm K_w \normw  \right)\, \left( 
\prod_{\iota\in\{\realinitial,\realfinal,\Join\}}
\mathop{\int}\limits_{\Si_{\Dom(A)}} \mu (\d \uw)   \chi_{\iota}(\uw) \prod_{w \in \uw} \norm K_w \normw \right).  \label{eq: apriori boundary1}
\eeq
The first integral is estimated
in \eqref{eq: unconstrained estimate}, and the three others in the same way using the 
bounds \eqref{eq: bound k as h}. Thus,  the $\iota=\realinitial,\realfinal$ -factors are bounded by  
$$ \sum_{m \in \bbN}   \frac{1}{m!}\left(   C   \norm \psi_{\iota} \norm^2 + C\int_{[0, n/\la^2]} \d s \,  \str \la  h_{\iota}(s) \str \right)^m  \leq \e^{\str \la \str C \norm h_{\iota} \norm_1+  C  \norm \psi_{\iota} \norm^2}$$ 
and the $\iota=\Join$ by $\exp (C \str h_{\Join}(n/\la^2)\str)$. 
Altogether  we conclude that the product \eqref{eq: apriori boundary1}, and hence \eqref{eq: apriori boundary}, is bounded by $\breve C \e^{ C \str A \str}$.
\end{proof}

To establish the bound \eqref{eqbound operator polymers boundary}, we closely follow the proof of \eqref{eqbound operator polymers} given in Section \ref{sec: proof of bound on operator polymers}. 
First, we apply the definition of the edge factors  $\hat e(\tau, \tau')= \hat e(\tau', \tau)$ in \eqref{def: edge factors} to the case where  $\{ \tau,\tau' \} \not\subset I_n$; explicitly,
\baq
 \hat{e}(\tau,\tau') &:=&   \left\{  \begin{array}{ll}   C  \mathop{\int}\limits_{\Dom{(\tau')}}   \d s \,    \str \la h_{\realinitial}(s) \str    &  \tau=\initial, \tau' \in I_n   \\[4mm]
 C \mathop{\int}\limits_{\Dom{(\tau')}}   \d s \,    \str \la h_{\realfinal}(n/\la^2-s) \str    &   \tau=\final, \tau' \in I_n      \\[4mm]
 C  \,\,\,\,  \str  h_{\Join}(n/\la^2) \str    &   \tau=\initial, \tau'= \final   
 \end{array}\right. 
  \eaq
  and then also $ \hat{e}_\al(\tau,\tau') :=   (1+ \str \tau'-\tau \str)^{\al} \hat{e}(\tau,\tau')$.
These boundary  edge factors satisfy 
\beq
 \hat{e}_\al(\initial,\final) \leq   \breve C, \qquad  \text{and} \qquad \sum_{\tau' \in I_n}   \hat{e}_\al(\tau,\tau')   \leq  \str\la\str \breve C, \qquad \textrm{for}\, \tau = \initial,\final   \label{eq: bounds edge factors boundary}
\eeq
as follows immediately from the  properties of $h_{\ltimes}, h_{\rtimes}, h_{\Join}$ in Assumption \ref{ass: initial states and observables}. 
To get \eqref{eqbound operator polymers boundary}, we first restrict ourselves  to $A$ with $\initial \in A$ but $\final \not \in A$ (in particular, $\tau=0$ in \eqref{eqbound operator polymers boundary}).  Starting from \eqref{eqminimal spanning bound boundary} and employing the definitions of $ \hat{e}_\al(\tau,\tau')$ above, we proceed as in \eqref{eqsum over trees} to derive
  \baq
\frac{1}{\breve C}  \sum_{A \in \breve I_n: A \ni \initial, \final \not \in A} (C\breve\ep)^{-\str A \cap I_n \str }\dist(A)^{\al} \norm G_A^c \normw & \leq &     \sum_{\textrm{trees} \, \scrT:  \initial \in \caV(\scrT), \final \not \in \caV(\scrT)   }  (C\breve \ep)^{-\str \caE(\scrT) \str}  \,     
\prod_{ (\tau,\tau') \in \caE(\scrT) }  \hat e_{\al}(\tau,\tau') \nonumber
  \label{eqsum over trees boundary}
\eaq
where we used $\str A\cap I_n \str \geq \str \scrE(\scrT)\str$.
For every tree $\scrT$ in this sum, we identify vertices $\tau_1 <\tau_2< \ldots < \tau_m$ such that $\{ \initial, \tau_j\}  \in \scrE(\scrT)$ and $m\geq 1$.  Then the sum over $\scrT$ is recast as a sum over $m$, choice  of edges $\tau_1 <\tau_2< \ldots < \tau_m$,  and trees growing out of them (including the possibility of no tree). This yields the bound
\baq
&& \sum_{m \geq 1} \sum_{\tau_1 < \ldots <\tau_m}   \prod_{j=1}^m   ( C\breve\ep)^{-1}\hat e_{\al}(0,\tau_j) \left(1+  \sum_{\textrm{trees} \, \scrT:  \tau_j \in \caV(\scrT) \subset I_n   }  (C \breve\ep)^{-\str \caE(\scrT) \str}  \,     
\prod_{ (\tau,\tau') \in \caE(\scrT) }  \hat e_{\al}(\tau,\tau') \right)  \\[2mm]
&&\leq  \sum_{m \geq 1} \sum_{\tau_1 < \ldots <\tau_m}   \prod_{j=1}^m   2 (C \breve \ep)^{-1}\hat e_{\al}(0,\tau_j) \leq  \sum_{m \geq 1} \frac{1}{m!} \left( \sum_{\tau \in I_n}   2 (C \breve \ep)^{-1} \hat e_{\al}(0,\tau) \right)^m  \leq   \breve C 
\eaq
The first inequality follows by  \eqref{eq: sum over trees bounded one}, since the trees are in the bulk, the last inequality uses \eqref{eq: bounds edge factors boundary}.
%  \baq
%  \sum_{A \in \breve I_n: A \ni \tau_0} (C\ep')^{-(\str A \cap I_n \str-1) }\dist(A)^{\al} \norm G_A^c \normw & \leq &   \breve C   \sum_{\textrm{trees} \, \scrT:  \tau_0 \in \caV(\scrT)  }  (C \ep)^{-\str \caE_{\mathrm{bulk}}(\scrT) \str}  \,     
%\prod_{ (\tau,\tau') \in \caE(\scrT) }  \hat e_{\al}(\tau,\tau') 
%  \label{eqsum over trees boundary}
%\eaq
%where $\caE_{\mathrm{bulk}}(\scrT)$ is the set of  bulk edges, (i.e.\ between vertices in $I_n$ hence excluding $\initial,\final$). To perform the sum on the RHS, we 
% first fix the non-bulk edges. The bulk edges form then a collection of trees rooted in the bulk vertices of the non-bulk edges. Those trees  can be summed exactly as in the proof of \eqref{eqbound operator polymers}; they are bounded by $1$.  Finally, we sum over the non-bulk edges by using \eqref{eq: bounds edge factors boundary}. 
% 
The bound over $A$ that do contain $\final$ but not $\initial$ and  $A$ that contain both $\final$ and $\initial$,  is analogous. Note however that for  $A$ that contain both $\final$ and $\initial$ we cannot extract an $\breve \ep$ factor for every edge since some trees contain the edge $\{\initial,\final \}$ (see in \eqref{eq: bounds edge factors boundary}). Therefore, the power of $(\breve\ep)^{-1}$ in  \eqref{eqbound operator polymers boundary} is not $\str A\str-1$ but $\str A \cap I_n \str$.

\section{Analysis of polymer model: Proof of main results}  \label{sec: analysis of polymer models}

\subsection{Proof of Theorem  \ref{thm: steady state}: Approach to steady state} \label{sec: approach to steady state general}

To prove Theorem \ref{thm: steady state}, we first exhibit approach to a steady state for discrete times of the form $t= n/{\la^2}$ and for a restricted class of initial states and observables (Section \ref{sec: approach to steady state}). Then we eliminate these restrictions in Sections \ref{secfromdiscretetocontinuoustime}, \ref{sec: from disc to cont} and \ref{eq: real proof of steady state}.  The reasoning in Section \ref{sec: approach to steady state} is based heavily on cluster expansions  and we advise the reader to read Appendix  \ref{appsec: cluster expansions} before continuing. 

\subsubsection{Approach to steady state for discrete times} \label{sec: approach to steady state}

We start from the polymer representation \eqref{eqscalar polymer model boundaries}  with $\ka=0$, hence $\theta_T=0$, 
\beq
Z_{n} =  k_\realinitial k_\realfinal  \sum_{\caA \in \breve \frB^{1}_{n}}       \prod_{A \in \caA}  v(A)  
\eeq
where $Z_n =Z_n(O,\rho_0, \ka=0)$  depends on the initial state and observable via $k_\realinitial,  k_\realfinal$ and the weights $v(A)$ for $A \cap \{ 0, n+1\} \neq \emptyset$.

We write $A \sim A'$ whenever $ \distance (A,A')  \leq 1$, and for a collection $\caA$, $\caA \sim A'$ whenever there is at least one $A \in \caA$ such that $A \sim A'$ (i.e., $\supp \caA \sim A'$). These definitions coincide with those in Section \ref{appsec: logarithm}  of Appendix  \ref{appsec: cluster expansions}. 

We separate each collection $\caA$ into its boundary and bulk polymers by writing
\beq
Z_{n} =   k_\realinitial k_\realfinal   \sum_{\scriptsize{\left.\begin{array}{c}   
\breve\caA \in \breve{\frB}^1_{n} 
\\
\forall A \in \breve\caA:  A \cap \{ \initial, \final\} \neq \emptyset
\end{array} \right. }}
\left(\prod_{A \in \breve\caA} v(A) \right )   Z_{n , \supp\breve\caA}  \label{eqpolymer with boundary isolated}
\eeq
where
\beq   Z_{n, A'}  :=
  \sum_{\scriptsize{\left.\begin{array}{c}   \caA \in {\frB}^1_{n}     \\   \caA \nsim  A' 
  \end{array} \right. }}    \prod_{A \in \caA}  v(A)      \label{eq: def z restricted}
\eeq
Note that $Z_{n, A'}$ depends only on bulk polymer weights and that $\breve\caA$ consist maximally of two sets, one containing the element $\initial$ and one containing $\final$.  By the identity \eqref{eq: z empty} with $\ka=0$, 
\beq  \label{eq: definition empty zet}
Z_{n , \emptyset} =  Z_n( \lone, \eta \otimes P_{\Om}, \ka=0)
\eeq  
%To see this, note first that the absence of Weyl operators in the initial state/observable implies that  $T_{\initial}=T_{\final}= \lone$.  
%Furthermore,  $\rho_{\sys,0}=\eta$ and $\tilde\eta=\lone= O_{\sys}$ (because $\ka=0$)  imply that
%$k_\realinitial=k_\realfinal=1$ and that $v(A)=0$ for $A \cap \{0, n+1 \} \neq \emptyset$ (inspection in \eqref{def: boundary values}).  Hence \eqref{eq: definition empty zet} follows. 
By unitarity and the fact that $\Tr_\sys \eta=1$, we have
\beq
Z_{n , \emptyset}=  \Tr [\e^{-\i (n/ \la^{2})  H } (\eta \otimes P_{\Om})  \e^{\i  (n/ \la^{2})   H }] = \Tr ( \eta \otimes P_{\Om} )=1   \label{eq: z emptyset is one}
\eeq
The quantity $Z_{n, A'}$ can be viewed as the partition function $\Upsilon_n$ of a polymer gas with  polymer weights  $
w (A) \equiv  v(A)   1_{[A \nsim  A']} $ (i.e.\ it is of the form \eqref{eq: abstract polymer gas} In Section \ref{appsec: logarithm}).
For $\ep$ small enough, the  Kotecky-Preiss criterion   \eqref{eqkotecky preiss abstract} is satisfied with $\delta  \equiv C \ep$ because of \eqref{eqbound scalar polymers} with $\al=0$ ($\al \neq 0$ will be used below), and hence Proposition \ref{prop: basic cluster expansion result} applies and yields
\beq \label{eq: rep z n a}
\log Z_{n, A'}   = \sum_{\caA \in \frB_n}   w^T(\caA)   =  \sum_{\caA \in \frB_n}   v^T(\caA)      \indicator_{[ \caA \nsim A' ]} 
\eeq
where the truncated weights $v^T(\cdot), w^T(\cdot)$ are related to $v(\cdot), w(\cdot)$, respectively,  through the formula \eqref{def: truncated cluster weights} in Section \ref{appsec: logarithm}.  Comparing to the expansion of $\log Z_{n, \emptyset}=0$, we get
\beq \label{eqgeneral expression for excluded z}
\log Z_{n,A'}  =   \log \frac{Z_{n,A'} }{Z_{n, \emptyset} }   =   -\sum_{\caA \in \frB_n}  v^T(\caA)       \indicator_{[ \caA \sim A' ]} 
\eeq
By the bound \eqref{eqbound on clusters touching something} in Proposition \ref{prop: basic cluster expansion result}, applied with $A_0=A'$, we get immediately 
\beq  \label{eq: boundzcon}
\str \log Z_{n, A'} \str \leq  C\ep \str A' \str 
\eeq
Next, we state

\begin{lemma}   \label{lem: existence z limits}    For $\str\la\str$ small enough and $\ka=0$, the following limits exist
\baq
z_{\realinitial}   &:=&   \lim_{n \to \infty}     \sum_{A \subset \breve I_n:  A \ni \initial, \final \notin A}  v(A)   Z_{n, A}    \label{eq: limit zinitial} \\
z_{\realfinal}   &:=&   \lim_{n \to \infty}     \sum_{A  \subset \breve I_n: A \ni \final, \initial \notin A}  v(A)   Z_{n, A}   \label{eq: limit zfinal}
\eaq
\end{lemma}
Note that the number 
 $z_{\realinitial}=z_{\realinitial}(\rho_0)$ does not depend on $O$ and $z_\realfinal= z_\realfinal(O)$ does not depend on $\rho_0$. 
\begin{proof}
To deal with the first limit, we recall the (partial) $n$- independence of polymer weights \eqref{eq: independence macro}, which enables, for $n'>n$, 
\baq
 && \sum^{(n')}_A  v^{(n')}(A)   Z_{n', A} -  \sum^{(n)}_A v(A)   Z_{n, A} =  \sum^{(n')}_{A: \max A >n}  v^{(n')}(A)  Z_{n',A} + \sum^{(n)}_Av(A) (\frac{Z_{n', A}}{Z_{n, A}}-1) Z_{n,A}   \label{eq: cauchy n and n prime}
\eaq
where we abbreviated $ \sum_{A \subset \breve I_n:  A \ni \initial, \final \notin A}$ by $\mathop{\sum}\limits^{(n)}_A$ (in particular, the constraint $0 \in A$ is implicit in all terms) and we indicated the $n$-dependence of polymer weights whenever its omission would cause an  ambiguity. 
We use \eqref{eq: boundzcon} and the bound $$\sum^{(n)}_{A: \dist(A)\geq m}  \str v(A) \str \e^{a_v \str A \str}    \leq  m^{-\al}  \breve C \breve \ep,$$ which follows from  (\ref{eqbound scalar boundary polymers}) in Lemma \ref{lem: bound on scalar polymers}, to conclude that the first term in \eqref{eq: cauchy n and n prime} is bounded by $\breve C \breve \ep n^{-\al}$.  
To bound the second term,  we write, for $\max A \leq n$, 
\beq
\str \log\frac{Z_{n', A}}{Z_{n, A}}\str \leq  \sum_{\caA \in \frB_{n'}:  \caA \sim A} 1_{\max \supp \caA > n} \str v^T(\caA)  \str \leq  C  \ep (1+n-\max A)^{-\al}
\eeq
The first inequality follows by inspecting \eqref{eqgeneral expression for excluded z}  and the second follows by the bound \ref{eqbound scalar polymers} in Lemma \ref{lem: bound on scalar polymers}, which implies decay of cluster weights, as explained in Section \ref{appsec: decay of cluster weights}. 
To conclude, we estimate the second term of \eqref{eq: cauchy n and n prime} as 
\beq
 \sum^{(n)}_A \str v(A) \str  \left\str \frac{Z_{n', A}}{Z_{n, A}}-1 \right\str \str Z_{n,A}  \str \leq  \sum^{(n)}_{A: \max A \leq n/2}   \frac{ \ep C   \e^{\ep C \str A \str}  \str v(A) \str }{(1+n-\max A)^{-\al}}   +   \sum^{(n)}_{A: \max A > n/2}    \str v(A) \str     \e^{\ep C \str A \str}  
\eeq
Using again Lemma \ref{lem: bound on scalar polymers}, this is bounded by $\breve C \breve \ep n^{-\al} $.

The convergence in  \eqref{eq: limit zfinal} is similar upon recalling  the comparison \eqref{eq: independence final} between weights of polymers in $\breve I_n$ and $\breve I_{n'}$.
\end{proof}
%To estimate these terms, we use the decay estimates
%\beq \label{eq: decay non truncated}
%, \qquad    \sum_{A:  A \ni \tau, \dist(A)\geq m}  v^T(A)  \leq  \breve C \breve m^{-\al}
%\eeq
%which follows immediately from Lemma 
%Since $A$ in the first term in \eqref{eq: cauchy n and n prime} is contrained to contain $0$, it has $\dist(A) \geq n$. Combining this with the bound and \eqref{eq: decay non truncated}, we get that this term is bounded by $ \breve \ep \breve C n^{-\al}$.
%
%and this is again bounded by $ \breve \ep \breve C n^{-\al}$, by the second estimate of \eqref{eq: decay non truncated}. This proves the convergence in \eqref{eq: limit zinitial}.   
In what follows, we now longer trace everywhere the factors $\ep,\breve \ep$, thus making the bounds less sharp then possible.

\begin{lemma}  \label{lem: convergence invariant state}    For $\str\la\str$ small enough and $\ka=0$,
\beq \label{eq: convergence to k z}
\left\str Z_n(O, \rho_0, \ka=0)  -   k_{\realinitial} k_{\realfinal} (1+ z_{\realinitial})(1+z_{\realfinal}) \right\str    \leq   \breve C n^{-\al}
\eeq
and we identify $ k_{\realinitial}  (1+ z_{\realinitial})= \Tr \rho_0$. 
%and $ k_{\realfinal}  (1+ z_{\realfinal})= \langle O \rangle_{\infty}$.
\end{lemma}
\begin{proof}
We start from \eqref{eqpolymer with boundary isolated}   and we split the sum over $\breve \caA$  in five parts, of which the first corresponds to  $\breve \caA = \emptyset$. To describe the other parts, let $A_{\realinitial}  $ in general stand for  subsets of $\breve I_n$ such that $\initial \in A, \final \not \in A$, let  $A_{\realfinal}$ stand for  subsets such that  $\initial \not\in A, \final \in A$ and, finally, write  $A_{\realinitial,\realfinal}$ for  subsets such that  $\{\initial,\final \} \subset A_{\realinitial,\realfinal}$. The splitting is
\beq
( k_\realinitial k_\realfinal  )^{-1}Z_{n}= Z_{n, \emptyset} +   \sum_{A_\realinitial}  v(A_{\realinitial}) Z_{n, A_{\realinitial}} + \sum_{A_\realfinal} v(A_{\realfinal}) Z_{n, A_{\realfinal}} +   \sum_{A_\realinitial, A_\realfinal: A_{\realinitial} \nsim A_{\realfinal}}  v(A_{\realinitial})  v(A_{\realfinal})  Z_{n, A_{\realinitial}\cup A_{\realfinal}}  +   \sum_{A_{\realinitial,\realfinal} }  v(A_{\realinitial,\realfinal}) Z_{n, A_{\realinitial,\realfinal}}   \label{eq: splitting in five sums}
\eeq
We have already argued that the first term on the RHS  is $1$ and the second and third term converge to $z_{\realinitial},z_{\realfinal}$, respectively.
We split the fourth term on the RHS as  $\sum_{A_\realinitial, A_\realfinal: A_{\realinitial} \nsim A_{\realfinal}} =\sum_{A_\realinitial, A_\realfinal} -\sum_{A_\realinitial, A_\realfinal: A_{\realinitial} \sim A_{\realfinal}}  $.
Let us concentrate on the first term in this splitting, i.e.\ the unconstrained sum over $A_{\realinitial}, A_{\realfinal}$. Clearly, this expression should tend to $z_{\realinitial} z_{\realfinal}$. To see this we rewrite using  \eqref{eqgeneral expression for excluded z}  
\baq
 \sum_{A_\realinitial, A_\realfinal} v(A_{\realinitial}) v(A_{\realfinal}) Z_{n, A_{\realinitial} \cup A_{\realfinal}}    -  \sum_{A_\realinitial, A_\realfinal} v(A_{\realinitial}) v(A_{\realfinal}) Z_{n, A_{\realinitial}} Z_{n, A_{\realfinal}} \nonumber  
 = \sum_{A_\realinitial, A_\realfinal}   v(A_{\realinitial})  v(A_{\realfinal}) Z_{n, A_{\realinitial}} Z_{n, A_{\realfinal}} 
 q_n(A_{\realinitial}, A_{\realfinal})  \label{eqboth ainitial and afinal}
 % \underbrace{\left( \e^{ -\sum_{\caA \in \frB_n}   v^T(\caA)       \indicator_{[\supp \caA \sim A_{\realinitial} ]}  \indicator_{[\supp \caA \sim A_{\realfinal} ]}               }  -1    \right)}\limits_{=: q_n(A_{\realinitial}, A_{\realfinal})}      \label{eqboth ainitial and afinal}
\eaq
with
\beq\label{qndef}
q_n(A_{\realinitial}, A_{\realfinal})= \e^{ -\sum_{\caA \in \frB_n}   v^T(\caA)       \indicator_{[ \caA \sim A_{\realinitial} ]}  \indicator_{[\caA \sim A_{\realfinal} ]}               }  -1.
\eeq
To bound $q_n$ note first that (set $\str a \str_+:=\max(a,0)$)
\beq
\distance(A_{\realinitial}, A_{\realfinal})  \geq   \str n- \dist(A_{\realinitial}) -\dist( A_{\realfinal}) \str_+
\eeq
so that the bound \eqref{eqdecay of cluster weights} in Section \ref{appsec: decay of cluster weights} together with Lemma \ref{lem: bound on scalar polymers} yields
\beq
 \sum_{\caA \in \frB_n}   \str v^T(\caA) \str      \indicator_{[ \caA \sim A_{\realinitial} ]}  \indicator_{[\caA \sim A_{\realfinal} ]}  \\
 \leq  
C(1+  \str n- \dist(A_{\realinitial}) -\dist( A_{\realfinal})\str_+)^{-\al} \ep \str A_{\realinitial}  \str  . \label{eq: bound exponent q}
\eeq 
 Using $ \str \e^{\ga x}-1  \str \leq   \ga \e^{\str x\str}$ for $ 0\leq \ga \leq 1, x\in \bbR$ we get from  \eqref{qndef}
and \eqref{eq: bound exponent q}
\beq
\str q_n(A_{\realinitial}, A_{\realfinal}) \str  \leq   (1+  \str n- \dist(A_{\realinitial}) -\dist( A_{\realfinal})\str_+)^{-\al}  \e^{ \ep C \str A_{\realinitial}  \str }.
\eeq
Furthermore, since for all $A_{\realinitial}, A_{\realfinal}$:
\beq
(1+  \str n- \dist(A_{\realinitial}) -\dist( A_{\realfinal})\str_+) \dist(A_{\realinitial})\dist(A_{\realfinal}) > cn,
\eeq
we conclude that 
\beq
\left\str\textrm{RHS of \eqref{eqboth ainitial and afinal}} \right\str  \leq C n^{-\al}    \left(\sum_{A_{\realinitial}} \dist(A_{\realinitial})^{\al} \str v(A_{\realinitial}) \str  \str Z_{n, A_{\realinitial}}\str  \e^{ \ep C  \str (A_{\realinitial} \str}  \right)   \left(\sum_{A_{\realfinal}} \dist(A_{\realfinal})^{\al} \str v(A_{\realfinal})\str  \str Z_{n, A_{\realfinal}}\str   \right) 
\eeq
The sums between brackets are estimated by $\breve C \breve\ep$ by first using \eqref{eq: boundzcon} and then Lemma  \ref{lem: bound on scalar polymers}.

Hence, we have obtained the asymptotics of the first part of the fourth term in \eqref{eq: splitting in five sums}:
\beq
 \left\str \sum_{A_\realinitial, A_\realfinal}  v(A_{\realinitial})  v(A_{\realfinal})  Z_{n, A_{\realinitial}\cup A_{\realfinal}}  - z_{\realinitial} z_{\realfinal} \right\str \leq  \breve C n^{-\al}.
  \eeq 
  By similar, but rather simpler reasoning, we can now treat the other terms in \eqref{eq: splitting in five sums}:
  \baq
  \left\str \sum_{A_\realinitial}  v(A_{\realinitial}) Z_{n, A_{\realinitial}} - z_{\realinitial}  \right\str  \leq  \breve C n^{-\al}  
   &&  \left\str  \sum_{A_\realfinal} v(A_{\realfinal}) Z_{n, A_{\realfinal}}  - z_{\realfinal}    \right\str  \leq  \breve C n^{-\al}  \\[3mm]   
      \left\str  \sum_{A_\realinitial, A_\realfinal: A_{\realinitial} \sim A_{\realfinal}}  v(A_{\realinitial})  v(A_{\realfinal})  Z_{n, A_{\realinitial}\cup A_{\realfinal}}      \right\str  \leq  \breve C n^{-\al}  
       &&    \left\str   \sum_{A_{\realinitial,\realfinal} }  v(A_{\realinitial,\realfinal}) Z_{n, A_{\realinitial,\realfinal}}     \right\str   \leq  \breve C n^{-\al} 
  \eaq
  In fact, the upper line has already been obtained in the proof of Lemma \ref{lem: existence z limits}.
%For example, for the last of these terms, we first bound  $\str Z_{n, A_{\realinitial,\realfinal}}  \str \leq \e^{ \ep C  \str A_{\realinitial,\realfinal} \str } $ and then, since $\dist(A_{\realinitial,\realfinal}) \geq n$,  
%\beq
%\sum_{A_{\realinitial,\realfinal} } \str v(A_{\realinitial,\realfinal}) \str   \e^{ \ep C   \str A_{\realinitial,\realfinal} \str }  \leq      n^{-\al} \sum_{A_{\realinitial,\realfinal} }  \dist(A_{\realinitial,\realfinal})^{\al} v(A_{\realinitial,\realfinal})  \e^{ \ep C   \str A_{\realinitial,\realfinal} \str }    \leq   \breve C n^{-\al}  
%\eeq
This proves \eqref{eq: convergence to k z}. 

 Take now $O= \lone$, then $Z_n = \Tr \rho_0$ and, by the third remark following \eqref{eqscalar polymer model boundaries} and using that $\tilde \eta=\lone$ for $\ka=0$, we have $k_{\realfinal}=1$ and $z_\realfinal=0$. Then, by the convergence established above, $Z_n \to k_{\realinitial} (1+ z_{\realinitial}) $ and hence  $\Tr \rho_0 =k_{\realinitial} (1+ z_{\realinitial}) $.\end{proof}

As promised, we comment on the case in which either $k_\realinitial$ or $k_\realfinal$ vanishes, thus invalidating our expression for the polymer weights given in \eqref{def: boundary values}. Let us assume for concreteness $k_\realinitial=0$. Then all non-vanishing contributions to \eqref{eqscalar polymer model boundaries} have $0 \in \supp \caA$. All polymers weights $v(A), A \ni 0$ are obtained by dividing a certain expression by $k_\realinitial$.  Therefore, it is possible to redefine $v(A)$ without the division and to omit  $k_\realinitial$ in the RHS of 
\eqref{eqscalar polymer model boundaries}.  It is straightforward to check that then the conclusions of the above lemmata still applies when $k_\realinitial (1+z_\realinitial)$ is replaced by $z'_\realinitial$ where  $z'_\realinitial$ is defined by removing the division by  $k_\realinitial$ in all its terms.
The same remark applies to $k_\realfinal$.

\subsubsection{Continuity  of the polymer representation} \label{secfromdiscretetocontinuoustime}

We envisage the situation that the model Hamiltonian and/or initial state and observable depend on a parameter $\nu \in \caD_\nu \subset \bbC$ such that our assumptions are satisfied for any $\nu \in \caD_\nu$.
Note first that the numbers $Z_{n}(\rho_0, O, \ka)$ are determined by $a)$ the correlation functions $ h,  h_{\realinitial},  h_{\realfinal}, h_{\Join}$ (we write $\tilde h$ to denote any of the four correlation functions),  $b)$ the operators  $D, H_\sys, \rho_{\sys,0}, O_{\sys}$ in $\scrB(\scrH_\sys)$  (we write $X_\sys$ to denote any of these operators), and $c)$ the parameter $\ka$.   Indeed, the dependence on, for example, the Weyl operators, is hidden in the correlation functions. 

Below, we indicate the dependence on the parameter $\nu$ by writing $\tilde h^{(\nu)}$ and $X_\sys^{(\nu)}$. 

\begin{lemma} \label{lem: continuity implies continuity of polymers}
Assume that
\ben
\item The bounds of Assumptions \ref{ass: decay correlation functions} and \ref{ass: initial states and observables} hold with the functions $\tilde h(s)$ replaced by
$\sup_{\nu \in \caD_\nu}\str \tilde h^{(\nu)}(s)\str$. 
\item Assumption \ref{ass: Fermi Golden Rule} is satisfied uniformly in $\nu \in \caD_\nu$, more precisely,  the gap $\gap_T$ and the constant $C_T$ on the RHS of \eqref{eq: iterates of t} can be chosen uniformly in $\nu$. 
\item  The  operators $X_\sys^{(\nu)}$ are continuous in $\nu$ and the  functions $\tilde h^{(\nu)}$ are pointwise continuous in $\nu$.
\een
Then,  for sufficiently small $\str \la \str, \str \ka\str$, the bounds on scalar polymer weights $v(A)= v^{(\nu)}(A) , A \in \breve I_n$ stated in Lemma \ref{lem: bound on scalar polymers} hold uniformly in $\nu$ and $v^{(\nu)}(A)$ is continuous in $\nu$, for any $A \in \breve I_n$. Also the factors $k^{(\nu)}_{\realinitial}, k^{(\nu)}_{\realfinal}$ are continuous in $\nu$. 
\end{lemma}

\begin{proof}
By the dominated convergence theorem, one first establishes that the bounds on $\norm G^c_A \normw$ hold uniformly in $\nu$, and that the  operators $T_\tau, G^c_A $ are continuous in $\nu$ (and the gap of $T$ can be chosen uniform in $\nu$)  and then one passes from the operators $T_\tau, G^c_A $ to the weights  $v^{(\nu)}$.
\end{proof}
Note that we do not claim the continuity to be uniform in $\la$ as $\la \to 0$. 
This can indeed not be deduced because of the rapidly oscillating factors $\e^{\i\la^{-2} t H_\sys}$ in the definition of $G_A, G_A^c$.

\subsubsection{From discrete to continuous time} \label{sec: from disc to cont}

We extend the result of Lemma \ref{lem: convergence invariant state} to all times $t$, rather than times of the form $n \la^{-2}$. 
Choose $\ell>0$ and imagine that we had started our analysis from Hamiltonian $ \ell^{-1} H$ instead of $H$. 
This would modify the operators $X_{\sys}$ and correlation functions $\tilde h$ (notation as above) as
\beq
D^{(\ell)} = \ell^{-1} D, \qquad  H_\sys^{(\ell)} = \ell^{-1} H_\sys, \qquad  \tilde h^{(\ell)}(s)  =   \tilde h(\ell^{-1} s) 
\eeq 
It is clear that all of our assumptions are satisfied for, say, $\ell \in [1,2]$. 
Hence all conclusions of our analysis hold for this modified model as well. We focus on  $ Z_n=Z_n( O,\rho_0,\ka=0)$ and let us  indicate explicitly the dependence on $\ell$ by writing $ Z_n^{(\ell)}$.  By the result of Section \ref{sec: approach to steady state}, for sufficiently small $\str \la \str$, the limit $ Z^{(\ell)}_{\infty} := \lim_{n \to \infty} Z^{(\ell)}_{n}$ exists for any $\ell \in [1,2]$.  Choose two values $\ell_1, \ell_2$ such that $\ell_1/\ell_2 \in \bbQ$, then there is a subsequence of  $Z^{(\ell_1)}_{\bullet}$ that coincides with a subsequence of $Z^{(\ell_2)}_{\bullet}$ and hence  $ Z^{(\ell_1)}_{\infty}=  Z^{(\ell_2)}_{\infty}$. 

%
%\beq
% Z^{(\ell)}_{\infty} =  \lim_{n \nearrow \infty}   \Tr [O \e^{-\i t H}\rho_0 \e^{-\i t H}  ] , \qquad t =n \ell/ \la^2  \label{eq: ell dependent limits}
%\eeq
Since the correlation functions $\tilde h^{(\ell)}$ are Fourier transforms of  $L^1$-functions, they are pointwise continuous in $\ell$. Hence,  Lemma \ref{lem: continuity implies continuity of polymers} applies with $\nu \equiv \ell$ and the polymer weights are continuous in $\ell \in [1,2]$. Since the bounds of Section \ref{sec: approach to steady state} can be chosen uniform in $\ell$, we deduce  that $Z^{(\ell)}_{\infty}$ is continuous in $\ell$, as well. 
Combining this with the arguments above, it follows that $Z^{(\ell)}_{\infty}$ is in fact independent of $\ell$, and we can hence define
\beq
  \langle O \rangle_{\infty} :=  \left( 1/ \Tr \rho_0 \right) \lim_{t \to \infty}  \Tr [O \e^{-\i t H}\rho_0 \e^{\i t H}  ]   \label{eq: independence of ell}
\eeq
which enables us to proceed to the
 
 \subsubsection{Proof of Theorem \ref{thm: steady state}}  \label{eq: real proof of steady state}

The convergence \eqref{eq: independence of ell} is proven up to now 
for observables $O$ and initial states $\rho_0$ satisfying Assumption  \ref{ass: initial states and observables}.  
To get point 1) of Theorem \ref{thm: steady state}, we show that the convergence holds for any initial density matrix $\rho_0 \in \scrB_1(\scrH)$ and observable $O \in \frW_{\al}$. First, we note that vectors $\psi_{\realinitial}$ satisfying  Assumption \ref{ass: initial states and observables} are dense in $L^2({\bbR^d})$. By the irreducibility of the Fock representation of the canonical commutation relations, this implies that vectors of the type $\psi_\sys  \otimes {\caW}(\psi_{\realinitial}) \Omega$ are dense in $\scrH$ (see e.g.\ \cite{petzcanonicalcommutationrelations}). Hence, linear combinations of initial density matrices $\rho_0$ for which \eqref{eq: independence of ell} holds, are dense in $\scrB_1(\scrH_\sys)$. On the other hand, by the unitarity of the dynamics, we have
\beq
\str \Tr (O \rho_t) -\Tr (O \rho'_t)    \str   \leq   \norm O \norm_{\scrB(\scrH)}  \norm \rho_0 -\rho'_0 \norm_{\scrB_{1}(\scrH)}   
\eeq
and hence by a simple density argument the convergence \eqref{eq: independence of ell} holds for any $\rho_0 \in \scrB_{1}(\scrH) $ and $O \in \frW_{\al}$.

Point 2) of Theorem \ref{thm: steady state}.  Recall that here $O$ is of the form \eqref{eq: choice o and rho}  and $\psi_{\realfinal}$ satisfying \eqref{ass: regularity final}. Without loss, we take   $\norm O_\sys \norm=1$.  Then 
\baq
 \str \langle O \rangle_{\infty} -  \Tr \left[ (P_{e_0} \otimes P_\Om) O \right]  \str &\leq &  \str \langle O \rangle_{\infty} - k_{\realfinal} \str   +  \str k_\realfinal-  \Tr \left[ (P_{e_0} \otimes P_\Om) O \right]  \str   \\
 &\leq&   \str k_{\realfinal}  z_{\realfinal} \str  +  \norm \eta - P_{e_0} \norm_{\scrB_1(\scrH_\sys)}   \left\str \langle \Om, {\caW}(\psi_\realfinal) \Om \rangle\right\str \\
 &\leq&   \breve C \breve\ep   \str k_{\realfinal} \str   +    C \str\la\str^{2 \al_*}       \left\str  \langle \Om, {\caW}(\psi_\realfinal) \Om \rangle  \right\str
\eaq
To obtain the second inequality, we used that $\langle O \rangle_{\infty}= k_\realfinal (1+ z_\realfinal)$, as obtained in Section \ref{sec: approach to steady state} and the definition \eqref{kfinal}.  The third inequality was established in the proof of Lemma \ref{lem: existence z limits} (bound on $z_\realfinal$) and in \eqref{eq: bound eigenvectors} (bound on $\eta - P_{e_0}$).  

Point 3) of the Theorem was established in Section \ref{sec: approach to steady state} for discrete times and is immediately extended to continuous times by the reasoning in Section \ref{sec: from disc to cont}.   Hence the full Theorem \ref{thm: steady state} is  now proven. 

\subsection{Photon number bound} \label{sec: photon number bound}

% we let $\Psi_t = \e^{-\i t  H}\Psi_0$ with   $\Psi_0=\Psi_{\mathrm{gs}}= \psi_{0} \otimes \Om$ and we investigate $\langle \Psi_t,  (1 \otimes \e^{\ka N}) \Psi_t \rangle$.
% It is natural to call 

To derive a photon number bound,  we are interested in $\Tr[\rho_t \e^{\ka N}]$, that is, $Z_n$ with $O= \lone$ and general $\rho_0$. However, it is more convenient to investigate first a slightly different quantity;  namely
 \beq
 p(\ka) := \lim_{n \to \infty}  \frac{1}{n}  \log   \Tr [(\tilde \eta \otimes \e^{\ka N})  \rho_{n/\la^2} ], \qquad \rho_0 =  \eta \otimes P_{\Om}   \label{def: pressure fake}
 \eeq
with $\eta, \tilde \eta$ as in \eqref{eqr explicitly}.
We call this quantity a 'pressure', since it is the logarithm of $Z_{n}$, which we view as a partition function.   The advantage of the definition \eqref{def: pressure fake} is that we can establish right away
\begin{lemma} \label{lem: analyticity of pressure in kappa}
 For $\str \ka \str$ small enough, the pressure $p(\ka)$ exists and is analytic in $\la$
\end{lemma}
\begin{proof}
Let $ \rho_{n/\la^2}$ be as above, then by  \eqref{eq: toy polymer} and \eqref{eq: z empty}, 
\beq  \label{eq: photon bound bulk only}
 \Tr [(\tilde \eta \otimes \e^{\ka N}) \rho_{n/\la^2} ] =  Z_n(\tilde\eta\otimes \lone,\eta \otimes P_{\Om},  \ka)= \Tr_\sys R Q_{n/\la^2}R =   \e^{n \theta_T}  \sum_{\scriptsize{\left.\begin{array}{c}   \caA \in {\frB}^1_{n}   
  \end{array} \right. }}    \prod_{A \in \caA}  v(A)  
\eeq 
Starting from \eqref{eq: photon bound bulk only} and using the cluster expansion results in the Appendix section \ref{appsec: pressure}, we get that $ p(\ka)$ exists for $\ep$ (hence $\la$) sufficiently small.     To verify that Section  \ref{appsec: pressure} is applicable, we need Lemma \ref{lem: bound on scalar polymers} and the translation invariance of polymer weights, see below \eqref{eqscalar polymer model boundaries}.

  We argue that $p(\ka)$ is analytic in $\ka$ by proceeding as in the proof of Lemma  \ref{lem: continuity implies continuity of polymers}. 
  Note first that the parameter $\ka$ enters only (at least for bulk polymers, which are the only ones concerning us here) by multiplying two of the four terms in $K_{u,v}$ by $\e^\ka$, see \eqref{eq: explicit for k}. Hence, by the Vitali convergence theorem, the operators $G^c_A, T$, and also the scalar polymer weights $v(A)$ are analytic in $\ka$. 
Since analytic polymer weights imply analyticity of the pressure by the Vitali convergence theorem (see Section \ref{appsec: continuity and analyticity of the pressure})  the lemma follows. 
\end{proof}

 Of course,  given sufficient infrared regularity, we expect $p(\ka)=0$.  This will be established in Section \ref{sec: pressure vanishes} by introducing and removing an infrared cutoff.  However, we would like to draw attention to the fact that there is in principle a more straightforward way. One could establish that  $\lim_{n \to \infty}\frac{1}{n}\log\Tr [ \e^{\ka N}  \rho_{n/\la^2} ]$ is independent of $\rho_0$, and then \emph{use the existence of a ground state} $\Psi_{\mathrm{gs}}$ with the property that $\langle\Psi_{\mathrm{gs}}, \e^{\ka N}  \Psi_{\mathrm{gs}} \rangle < \infty $ to argue that $\lim_{n \to \infty}\frac{1}{n}\log\Tr [ \e^{\ka N}  \rho_{n/\la^2} ]$ must vanish in general, since it vanishes for $\rho_0 = \str \Psi_{\mathrm{gs}} \rangle \langle \Psi_{\mathrm{gs}} \str$. 
 We choose not to exploit this approach so as to keep our presentation as self-contained as possible. 
 
Once we have $p(\ka)=0$, we still have to exclude that $\Tr [ \e^{\ka N}  \rho_{t} ]$ grows subexponentially in $t$, and this is done in Section \ref{sec: finite size}.

\subsubsection{Infrared cutoff and spectral translations} \label{secinfrared cutoff}
In this section, we introduce an infrared cutoff $\ga$ and we argue that the pressure $p(\ka,\ga)$ is continuous in $\ga$.

We implement the cutoff by translating the form factor in the spectral parameter of the one-photon Hamiltonian, i.e.\  the multiplication with $\str q \str$ on $L^{2}(\bbR, \d q)$. We write $q = \om \hat q$ with $\om =q$ and $\hat q \in \bbS^{d-1}$, the sphere equipped with the surface measure $\d \hat q$.
Consider the isomorphism of  Hilbert spaces 
\beq
J: L^2(\bbR^d,  \d q)\to L^2(\bbR_{+}, L^2(\bbS^{d-1}, \d \hat q),  \d \om   ), 
\eeq such that   $ (J\psi)(\om)$ is the function $\hat q \to \om^{\frac{d-1}{2}} \psi(\om \hat q)$ in $L^2(\bbS^{d-1}, \d \hat q)$.  We  define the '$\ga$-translated'  form factor $ \phi^{(\ga)}$ by specifying $J\phi^{(\ga)}$:
\beq
(J\phi^{(\ga)}) (\om)  :=    \left\{ \begin{array}{ll}  J \phi( \om-\ga ) & \om \geq \ga  \\[1mm]
 0   &     \om < \ga
   \end{array} \right.  \qquad \textrm{for} \, \ga \geq 0 
\eeq
In general, the correlation function $h(\cdot)$ can be written as
\beq
h(t) =  \mathop{\int}\limits_{\bbR^d} \d q \e^{-\i \str q \str t}  \str \phi(q) \str^2 =  \mathop{\int}\limits_{\bbR^+} \d \om \e^{-\i \om t}    \norm  J\phi(\om) \norm_{L^2(\bbS^{d-1})} 
\eeq
It is immediate that $h^{(\ga)}(t) =\e^{-\i \ga t} h(t)$ and hence in particular 
\beq
\left\str h^{(\ga_1)}(t)  -  h^{(\ga_2)}(t)      \right\str   \leq   \left\str 1- \e^{\i (\ga_1- \ga_2) t}    \right\str  \str h(t) \str  \label{eq: uniform boundedness of f correlation functions}
 \eeq

We want to apply Lemma \ref{lem: continuity implies continuity of polymers} to conclude that the  polymer weights $v^{(\ga)}(A)$ are continuous in $\ga$, and consequently (see Section \ref{appsec: continuity and analyticity of the pressure}) also $\ga \mapsto p(\ka,\ga)$ is continuous.    In fact,  the conditions of Lemma \ref{lem: continuity implies continuity of polymers} demand that we also check continuity of the functions $h_{\ltimes}, h_{\rtimes}, h_{\Join} $, but this is not necessary as long as we only use bulk polymers (since those do not depend on $h_{\ltimes}, h_{\rtimes}, h_{\Join} $). Indeed, the pressure $p(\ka,\ga)$ above depends only on the bulk polymers and hence we conclude that  $p(\ka,\ga)$ is continuous in $\ga \in [0,1]$.

\subsubsection{Pressure vanishes} \label{sec: pressure vanishes}
The advantage of having a sharp infrared cutoff is that we can bound the expectation of the number operator in  the state $\Psi_t$ in terms of the expectation in the state $\Psi_0$.
Indeed, let us decompose the number operator $N$ as
\beq
N^{}= N^{\geq \ga}+ N^{< \ga}, \qquad  N^{\geq \ga} :=   \int_{\str q \str \geq \ga} \d q  \,   a^*_q a_q, \qquad  N^{< \ga} :=    \int_{\str q \str < \ga} \d q \,    a^*_q a_q  
\eeq
Then because of the cutoff,  $[N^{< \ga}, H]=0$, and we can bound $N^{ \geq \ga} \leq \ga^{-1} H_\res$.  Choosing $\Psi_0$ such that $\langle\Psi_0, N \Psi_0 \rangle <C$ and $ \langle\Psi_0, H \Psi_0 \rangle < C$ and $\norm \Psi_{0} \norm=1$, we estimate
\baq
 \ga \langle \Psi_t, N \Psi_t \rangle &\leq&  \ga \langle \Psi_0, N^{< \ga}  \Psi_0 \rangle  +  \langle \Psi_t, H_\res \Psi_t \rangle   \label{eqbound number}
\eaq
To get an estimate on $  \langle \Psi_t, H_\res \Psi_t \rangle$, we write $H_\res= H- H_\sys- \la H_{\inter}$, we use that $\langle \Psi_{t}, H \Psi_{t} \rangle= \langle \Psi_{0}, H \Psi_{0} \rangle $, the boundedness of $H_{\sys}$ and  the  estimate \eqref{eq: infinitesimal perturbation} to obtain
\beq
\langle \Psi_t, H_\res \Psi_t \rangle \leq C +   \str \la \str C'   \langle \Psi_t, H_\res \Psi_t \rangle  \quad  \Rightarrow  \quad  \langle \Psi_t, H_\res \Psi_t \rangle \leq C''
\eeq
for $\la$ sufficiently small.
%
%We use the spectral theorem to define the probability measure $\mu=\mu_{\Psi_t}$ on $\bbN$
%\beq
%\mu(m) := \langle \Psi_t,  \indicator_{[N=m]}  \Psi_t \rangle 
%\eeq
% By Jensen inequality and \eqref{eqbound number}, 
%\beq
% 0 \geq  \log \sum_m \e^{\ka m} \mu(m)  \geq     \ka \sum_m \mu(m) m   \geq   \ka C \ga^{-1}, \qquad      (\Im \ka =0, \Re \ka \leq 0) 
% \eeq
By using the spectral theorem and Jensen's inequality, we can now bound $ \langle \Psi_t, \e^{\ka N} \Psi_t \rangle$  with $\ka \leq 0$ from below and obtain
 \beq
 \lim_{n \to \infty}\frac{1}{n} \log  \langle \Psi_{n\over\la^2} , \e^{\ka N} \Psi_{n\over\la^2} \rangle =0, \qquad \ka \leq 0  \label{eq: physical pressure zero}
 \eeq
 We want to deduce that also  $p$ as defined in \eqref{def: pressure fake} vanishes. 
 This indeed follows immediately. Firstly, since the operator $\tilde\eta$ in the definition of $p(\ka,\ga)$ satisfies $\norm \lone- \tilde\eta \norm=\caO(\str \ka \str)$ and  for real $\ka$, it is a positive-definite operator, we can find constants $c,C$ such that  $ c \lone \leq  \tilde \eta \leq C\lone $ for sufficiently small, real $\ka $.  Secondly, the density matrix $\rho_0 = \eta \otimes P_{\Om}$ is a finite sum  of terms $\str \Psi_0 \rangle \langle \Psi_0 \str$ with $\Psi_0$ satisfying all conditions necessary for the estimate \eqref{eqbound number}. 
Hence we have shown
\beq
p(\ka,\ga) =0, \qquad   \ka \leq 0, \ga >0. 
\eeq
 However, by analyticity in $\ka$ (Lemma \ref{lem: analyticity of pressure in kappa}), this holds for sufficiently small  $ \str\ka \str$ and by continuity  of the pressure in $\ga$ (established in Section \ref{secinfrared cutoff}), it holds for  $\ga=0$, as well.    

\subsubsection{Proof of Theorem \ref{thm: photon bound}: Finite-size corrections}  \label{sec: finite size}

The fact that $p=0$ tells us merely that the number of emitted photons grows slower than linearly in time, a result that has been established with other techniques as well. We aim to prove that it does not grow at all.  Let us denote the quantitiy in \eqref{eq: photon bound bulk only} by 
  $Z_{n,\emptyset}$, in analogy to the definition in Section \ref{sec: approach to steady state} (the only difference is that now $\ka \neq 0$).  We invoke the discussion in Section \ref{appsec: pressure}, in particular the formula \eqref{eqfinite size correction to pressure} which, applied to the case at hand, reads
\baq
\log Z_{n,\emptyset}  &= &  - \sum_{\caA \in\frB_{\infty}: \min \supp \caA =1 }  \min \left(\dist(\caA)-1, n \right)   v^T(\caA)    \label{eqfinite size correction to zero pressure}
\eaq
where we also used that $p=0$.  We bound \eqref{eqfinite size correction to zero pressure} as
\baq
\str \log Z_{n,\emptyset} \str  
& \leq  &    C n^{1-\min(\al,1)} \sum_{\caA \in\frB_{\infty}: \min \supp \caA =1 }   \dist(\caA)^{\al}\str v^T(\caA) \str < C  n^{1-\min(\al,1)}  
\eaq
by Proposition \ref{prop: basic cluster expansion result}.
Next, we show that this bound remains valid for arbitrary initial state satisfying Assumption \ref{ass: initial states and observables}. We proceed as in Section \ref{sec: approach to steady state}, obtaining
\beq  \label{eq: ration of z with kappa}
\frac{Z_n}{Z_{n,\emptyset}} =   k_\realinitial k_\realfinal   \sum_{\scriptsize{\left.\begin{array}{c}   
\breve\caA \in \breve{\frB}^1_{n} 
\\
\forall A \in \breve\caA:  A \cap \{0, n+1\} \neq \emptyset
\end{array} \right. }}
\left(\prod_{A \in \breve\caA} v(A) \right )   \frac{Z_{n , \supp\breve\caA}}{Z_{n,\emptyset}} 
\eeq
In contrast to the situation in Section \ref{sec: approach to steady state}, $Z_n \neq 1$ since $\ka\neq 0$. However, the second equality in  \eqref{eqgeneral expression for excluded z} still applies and, in analogy to \eqref{eq: boundzcon}  we get $\str \log \frac{Z_{n , \supp\breve\caA}}{Z_{n, \emptyset}}  \str \leq \ep C \str \supp\breve\caA \str$ and then also 
\beq
\str \frac{Z_n}{Z_{n,\emptyset}} \str \leq \breve C,
\eeq
by combining  \eqref{eq: ration of z with kappa} and \eqref{eqbound scalar boundary polymers} and making $\ep$ sufficiently small.  

Hence we have in general
\beq \label{eq: photon bound discrete}
\str Z_n\str  =   \str \frac{Z_n}{Z_{n,\emptyset}} \str \str Z_{n,\emptyset} \str   \leq    \breve C   \e^{ C n^{1-\min(\al,1)}  }
\eeq

Finally, we mimic the reasoning in Section \ref{sec: from disc to cont} to extend the bound \eqref{eq: photon bound discrete} to continuous time. This finishes the proof of Theorem \ref{thm: photon bound}.

\appendix
\renewcommand{\theequation}{\Alph{section}\arabic{equation}}
\setcounter{equation}{0}

\section{Cluster Expansions} \label{appsec: cluster expansions}

We present  some standard results on cluster expansions. We opted not to present a very general exposition, but instead a framework that is maximally adapted to our problem. A slight exception is Section \ref{sec: combinatorial lemma} where we state a combinatorial lemma that is used in several places in our paper.   In that section, the 'polymers' can be anything, but in the remaining sections of this appendix, we assume them to be subsets $A$ of $\bbN$, just as in our construction, and although we call the polymer weights $w(A)$, the reader can safely confuse them with our concrete polymer weights $v(A)$ defined in Section \ref{sec: scalar polymer model}.  In the sections other than Section \ref{sec: combinatorial lemma}, we also freely use the notation introduced in Section \ref{sec: notation for combinatorics}. 
This appendix follows very closely the presentation in  \cite{ueltschi}, which uses the convergence criterion by \cite{koteckypreiss}.

\subsection{A combinatorial lemma} \label{sec: combinatorial lemma}

Let $\scrS$ be a countable set. We call its elements  $S\in\scrS$  {\it polymers}.  A function
$w:\scrS\to \bbC$ will be called  a \emph{polymer weight}.  Furthermore, let   $\sim$ be
an adjacency relation on  $\scrS$  i.e.  $\sim$ is  a symmetric and irreflexive relation.  Given
a finite collection of polymers
$\caS\subset\scrS$ we define
\beq
k(\caS) :=    \sum_{\scrG \in \frG^c(\caS) }  
 \prod_{\{S_i, S_j\} \in  \scrE(\scrG)} 1_{[S_i \sim S_j]} 
\eeq
where the sum runs over $\frG^c(\caS)$, the set of connected graphs with vertex set $\caS$, and 
 $\scrE(\scrG) $ is the edge set of the graph $\scrG$.  
Finally, for a collection $\caS$ and a polymer $S_0$, we write  $ \caS \sim S_0$ whenever there is a $S \in \caS$ such that $S \sim S_0$.
Then, we have the following combinatorial result

\begin{lemma}\label{lemma: combi trick abstract}
Assume that there is a positive function $S \mapsto a(S)$ such that the so-called "Kotecky-Preiss criterion"
\beq  \label{eq: kp extra abstract}
\sum_{S: S \sim S'}  \str w (S) \str \e^{a(S)} \leq   a(S')
\eeq
holds. 
Then 
\beq
    \sum_{\caS:    \caS \sim S_0 }    k(\caS)   \prod_{S \in \caS} \str w(S) \str  \leq  a(  S_0 )   \label{eqbound on clusters touching something abstract}
\eeq
and
\beq
    \sum_{\caS }    k(\caS\cup \{S_0\})   \prod_{S \in \caS} \str w(S) \str  \leq  \e^{a(  S_0 )}   \label{eqbound on clusters containing something abstract}
\eeq
where the term corresponding to $\caS= \emptyset$ is understood to equal $1$.
\end{lemma}
An easy way to prove \eqref{eqbound on clusters containing something abstract} is via induction in the number of elements of the collections $\caS$.  The claim \eqref{eqbound on clusters touching something abstract} then follows by choosing $S_1 \in \caS$ such that $S_0 \sim S_1$ and applying \eqref{eqbound on clusters containing something abstract} with $S_1$ in the role of $S_0$.

\subsection{Logarithm of the partition function} \label{appsec: logarithm}

We make the setup introduced in Section \ref{sec: combinatorial lemma} more concrete
\ben
\item The polymers  are now subsets of $I_n=\{1, \ldots, n \}$. They are denoted by $A,A', \ldots$. 
\item The adjacency relation $\sim$ is defined to be  $A \sim A' \Leftrightarrow      \distance(A, A' )  \leq 1$ with $\distance(A, A' ) = \inf_{\tau \in A, \tau' \in A'} \str \tau -\tau'\str$.
\item Collections of polymers are denoted by $\caA$. We use the sets of collections $\frB_n,\frB^{1}_n$ introduced in Section \ref{sec: notation for combinatorics}.
\een
%
%
%We consdider the adjacency relation $\sim$ on $2^{I_n}$, i.e. subsets of $I_n$, given by 
%\beq \label{def: compatibility relation}
%A \sim A' \Leftrightarrow      \distance(A, A' )  \leq 1
%\eeq
%and we assume that for every $A \subset I_n$, there is a weight $w(A) \in \bbC$. For a collection of sets $\caA$, we write $\caA \sim A'$ as a shorthand notation for $\supp \caA \sim A'$ (coinciding with the definition in Lemma \ref{lemma: combi trick abstract})
%
%
We start from a polymer-gas representation of some \emph{partition function} $\Upsilon_n$:
\beq
\Upsilon_n  =  \sum_{\caA \in \frB^1_n }   \prod_{A \in \caA} w(A)  =  \sum_{\caA \in \frB_n }   \prod_{A \in \caA} w(A)    \prod_{\{A_{i}, A_j \} \subset \caA  }   \indicator_{A_{i} \nsim A_j}  \label{eq: abstract polymer gas}
   \eeq
where the last product runs over pairs in the collection $\caA$, and the summand is defined to be $1$ if $\caA =\emptyset$.

We define the \emph{truncated weights} $w^T$ of $\caA \in \frB_n$ as follows
  \beq  \label{def: truncated cluster weights}
w^T(\caA) :=    \sum_{\scrG \in \frG^c(\caA) }  (-1)^{\str  \scrE(\scrG) \str}    \prod_{\{A_{i}, A_j\} \in  \scrE(\scrG)} 1_{[A_{i} \sim A_j]}  \prod_{A_{i} \in \caA} w(A)
\eeq
where the sum runs over $\frG^c(\caA)$, the set of connected graphs with vertex set $\caA$, and 
 $\scrE(\scrG) $ is the edge set of the graph $\scrG$.   We call $\caA$ a cluster whenever the graph on $\caA$ with edge set $\{ \{A_{i}, A_j\},    A_{i} \sim A_j \}$ is connected.  Hence, if $\caA$ is not a cluster, then $w^T(\caA)=0$. 
% 
% We also define $k(\caA)$ as  the RHS of \eqref{def: truncated cluster weights}  with absolute values inside the sum over graphs $\scrG $, that is, 
%\beq
%k(\caA) :=    \sum_{\scrG \in \frG^c(\caA) }  
% \prod_{\{A_{i}, A_j\} \in  \scrE(\scrG)} 1_{[A_{i} \sim A_j]}  \prod_{A_{i} \in \caA} \str w(A) \str
%\eeq

Next, we state the basic result of cluster expansions (cfr.\ (eq. 4) in \cite{ueltschi}).

\begin{proposition}   \label{prop: basic cluster expansion result}
Assume there is  $\delta >0$ such that the 'Kotecky-Preiss criterion'
\beq
\sum_{A \subset I_n: A \sim A'}   \e^{\delta \str A \str}  w(A)   \leq \delta \str A' \str     \label{eqkotecky preiss abstract}
\eeq
holds.
Then $\Upsilon_n \neq 0$ and
 \beq
\log \Upsilon_n=  \sum_{\caA \in \frB_n}    w^T(\caA),
 \eeq
and, for any $A_0 \subset I_n $ 
\beq
 \sum_{\caA \in \frB_n:   \caA \sim A_0  }    \left \str w^T(\caA) \right\str    \leq    \delta \str  A_0 \str   \label{eqbound on clusters touching something}
\eeq
%\beq
% \sum_{\caA \in \frB_n:   \caA \sim A_0  }    \left \str w^T(\caA) \right\str    \leq    \sum_{\caA \in \frB_n,    \caA \sim A_0 }    k(\caA)   \leq   \delta \str  A_0 \str   \label{eqbound on clusters touching something}
%\eeq
\end{proposition}
This proposition is a direct consequence of Lemma \ref{lemma: combi trick abstract}.
It is clear how it applies to the present paper. The polymer weights $w(A)$ are $v(A)$ as constructed in Section \ref{sec: scalar polymer model}, and the Kotecky-Preiss criterion is  \eqref{eqbound scalar polymers} where $\delta$ should be chosen not smaller than $C \ep$ and not larger than $c$ ($C,c$ as in  \eqref{eqbound scalar polymers}). 

\subsection{Decay of cluster weights}\label{appsec: decay of cluster weights}

We assume now that the Kotecky-Preiss criterion is satisfied in a stronger sense, namely: 
\beq
\sum_{A \subset I_n: A \sim A'}   \e^{\delta \str A \str}    w_\al(A)   \leq \delta \str A' \str ,   \label{eqkotecky preiss distance abstract}
\eeq
where 
$$ \, w_\al(A)   :=\dist(A)^{\al}  w(A)$$
 and we show that this yields some decay in the cluster weights.
Proposition \ref{prop: basic cluster expansion result}  gives
\beq
 \sum_{\caA \in \frB_n:   \caA \sim A_0  }    \left \str w_\al^T(\caA) \right\str   =   \sum_{\caA \in \frB_n:   \caA \sim A_0  }  \prod_{A \in \caA}  \dist(A)^{\al}  w^T(\caA)     \leq   \delta\str A_0 \str \label{eqdecay of cluster weights with alpha}
\eeq
If $\caA$ is a cluster  then
$\sum_{A \in \caA} \dist(A) \geq  \dist(\caA)$,   and  if additionally all polymers $A$ with nonzero weight have $\dist(A) >1$ (this is our case),
 then 
\beq
 \prod_{A \in \caA}  \dist(A)^{\al} \geq         \dist(\caA)^{\al}  \label{eq: distances on clusters}
\eeq
Hence, we use \eqref{eqdecay of cluster weights with alpha} to derive 
\baq
%\sum_{\scriptsize{ \left.\begin{array}{c}   \caA \in \frB_n      \\      \caA \sim \{\tau_1\}, \caA \sim \{\tau_2\}    \end{array}\right.   }}
 \sum_{\caA \in \frB_n:   \caA \sim A_0  } 
1_{\dist(\caA) \geq m}
   \left \str w^T(\caA) \right\str   
  & \leq &    C   \delta   \str A_0 \str (1+m)^{-\al} 
 \label{eqdecay of cluster weights}
\eaq
%
%\baq
%%\sum_{\scriptsize{ \left.\begin{array}{c}   \caA \in \frB_n      \\      \caA \sim \{\tau_1\}, \caA \sim \{\tau_2\}    \end{array}\right.   }}
%\sum_{ \caA \in \frB_n }
%1_{ \caA \sim A_1}1_{ \caA \sim A_2}
%   \left \str w^T(\caA) \right\str    & \leq &   
%   \sum_{ \caA \in \frB_n }1_{ \caA \sim A_1}1_{ \caA \sim A_2}  
%    \frac{ \left \str w_\al^T(\caA) \right\str }{  \dist(\caA)^{\al} }   \nonumber \\
% & \leq &  
%    \frac{ C }{  (1+\distance(A_1,A_2))^{\al} }    \sum_{ \caA \in \frB_n }1_{ \caA \sim \{A_1 \cup A_2\}} 
% \left \str w_\al^T(\caA) \right\str  \nonumber \\
%  & \leq &   \frac{   C   \delta   \str A_1 \cup A_2 \str }{(1+\distance(A_1,A_2))^{\al}} 
% \label{eqdecay of cluster weights}
%\eaq

\subsection{Pressure}  \label{appsec: pressure}
In the previous sections, $n$ was a fixed parameter. Now, we consider different values of $n$. If $A \subset I_n \subset I_{n'}$ ($n' >n$), we assume that $w(A)$ does not depend on whether we view  $A$ as polymer in $I_n$ or $I_{n'}$. Moreover, we assume that the weights are translation-invariant in the sense that  
\beq
w(A) = w(A+ a)
\eeq 
where $A+a :=  \{ \tau, \tau -a \in A\}$.
An immediate consequence is that the pressure
\beq
p : = \lim_{n\to\infty}   n^{-1} \log \Upsilon_n 
\eeq
exists and is given by
\beq
p  =  \sum_{\caA \in \frB_{\infty}, \min \supp\caA = 1}  w^T(\caA)  \label{eqexpression pressure}
\eeq
where $\frB_{\infty} = \cup_n \frB_n$  is the set of finite collections of finite subsets of $\bbN$. 
Indeed, the absolute summability of the RHS of \eqref{eqexpression pressure} follows from Proposition \ref{prop: basic cluster expansion result}, whereas (we set $\str a \str_+:= \max(a,0)$)
\baq
p -   n^{-1}\log \Upsilon_n  &= &   \sum_{\caA \in\frB_{\infty}: \min \supp \caA =1 }   \left(\frac{n    -  \left\str n+1- \dist( \caA) \right\str_+  }{n} \right)   w^T(\caA)    \label{eqfinite size correction to pressure}
\eaq
which vanishes as $n \to \infty$ be the dominated convergence theorem, since the factor between brackets converges to $0$ for any value of $\dist(\caA)$.

\subsubsection{Continuity and analyticity of the pressure} \label{appsec: continuity and analyticity of the pressure}

Often, one encounters the situation where the weights $w(A)$ depends on a parameter $\nu \in \caD_\nu \subset \bbC$, $w(A)= w^{(\nu)}(A)$ such that the Kotecky-Preiss criterion  
\eqref{eqkotecky preiss abstract}  is satisfied even when $w(A)$ is replaced by $\sup_{\nu \in \caD_\nu} \str w^{(\nu)}(A) \str$.  Then 
\ben
\item If the weights $w^{(\nu)}(A) $ are continuous in $\nu \in \caD_\nu$, for any $A \subset \bbN$, then the pressure is continuous in $\nu \in \caD_\nu$.
\item If the weights $w^{(\nu)}(A) $ are analytic in $\nu \in \caD_\nu$, for any $A \subset \bbN$, then the pressure is analytic in $\nu \in \caD_\nu$.
\een
Starting from the absolute summability of the RHS of \eqref{eqexpression pressure}, the first claim
 follows from the dominated convergence theorem, and the second from the
 Vitali convergence theorem.

\section{Van Hove limit}  \label{secweak coupling limit}

In Section \ref{sec: analysis of t}, we introduced the Lindblad generator $M$ and we stated and exploited the claim (Proposition \ref{prop: weak coupling}) that 
\beq
\norm Q_{t}  -  \e^{-\i t L_{\sys} + \la^2 t M }  \norm  \leq  \e^{C \la^2 t}  \str \la\str^{2\al_*}  \label{eq: wc app}
\eeq
Below we derive this claim. The derivation is independent of the rest of the paper, except for the representation of the effective dynamics $Q_t$: 
\beq   \label{eq_ dyson for app}
 Q_t   =  \int_{\Si_{[0, t]}}   \,   \mu(\d\uw)   \,     \caT\left\{  \mathop{\otimes}\limits_{w \in \uw} K_w    \right\}
\eeq
and the bound $\norm K_{u,v} \normw \leq \la^2 \str h(v-u) \str $.  Moreover, we freely use Assumption \ref{ass: decay correlation functions}, i.e.\ $\int_{0}^{\infty} \d t (1+t )^{\al} \str h(t)\str \leq C$.  We set $\ka=0$ in this Appendix. Including a nonzero value for $\ka$ does not affect the reasoning.

\subsection{Alternative construction of   $M$} \label{sec: alternative construction of m}

  We define the operator $F_{s} \in \scrR = \scrB(\scrB_1(\scrH_\sys))$ using  the notation of Section \ref{secdiscretization}; 
\baq
\la^2 F_{v-u}   &:=  &  \e^{-\i v L_\sys} \caT[K_{u,v}]  \e^{\i u L_\sys} 
\eaq
For the sake of clarity, we  write out $F_s$ explicitly (recall that $\ka=0$ here)
\baq
F_{s}  &= &   - h(s)  \caR (D)  \e^{-\i s L_\sys} \caL ( D)  -  h(-s)  \caL (D)   \e^{-\i s L_\sys} \caR ( D)     \nonumber \\[1mm]
&+&    h(s)  \caL (D)  \e^{-\i s L_\sys}  \caL ( D)    + h(-s)  \caR (D)    \e^{-\i s L_\sys}\caR( D)    \label{eq: f four terms}
\eaq
Introduce  spectral projections $P_{\ve \in \caE}$ of the Liouvillian $L_\sys$. The range of $P_{\ve=0}$ has dimension $\dim \scrH_\sys$ and the $P_{\ve \neq 0}$ are one-dimensional if for any $\ve \in \caE, \ve \neq 0$, there is a unique pair $(e,e'), e,e' \in \si(H_\sys)$ such that  $\ve=e-e'$ (recall the discussion in Section \ref{sec: construction of lindblad generator}), but this plays no role in what follows. 
The connection between the operators $F_s$ and the generator $M$ given in  Section \ref{sec: analysis of t}  is 
\baq
M  &= &   \sum_{\varepsilon \in \caE}   \int_0^{\infty}  \d s \,   \e^{\i s \ve}  P_{\varepsilon} F_{s}  P_{\varepsilon}    \label{def: m as function of f}
\eaq
This can be checked by direct computation. The four terms in \eqref{eq: f four terms}  are matched with  the expression \eqref{def: lindblad} as follows:  the first two terms   give the first term under the integral (with $\kappa$). The two other terms yield the commutator with $H_{lamb}$ and the second term under the integral. 
Note also that the part of $M$ corresponding to the $\ve=0$ term  on the RHS of  \eqref{def: m as function of f} is isomorphic to the operator $\caM$ discussed in Section \ref{sec: construction of lindblad generator}, and the parts corresponding to $\ve\neq 0$ were discussed under the heading 'off-diagonal elements' in that same section. 

%\begin{remark}
%From the above expression for $F_{v, u} $, it is easy to recover an explicit expression for $K_{u,v}$: it suffices to replace in the four terms the rightmost, respectively leftmost $\adjoint_{\#}(D)$ by $\bsI_{u}[\adjoint_{\#}(D)]$, respectively $\bsI_{v}[\adjoint_{\#}(D)]$, and to strip away the oscillating factor  $ \e^{\i (v \ve'-\ve u)}$. 
%\end{remark}

\subsection{Emergence from the Dyson expansion}
We start from the Dyson series \eqref{eq_ dyson for app} and
we split the set of collections of  pairs $\Si_{[0, t]}$ into leading  $\Si_{[0, t]}^{\mathrm{l}}$ and non-leading  $\Si_{[0, t]}^{\mathrm{nl}}$ subsets, as follows
\beq
\Si_{[0, t]}^{\mathrm{l}}   := \{  \uw \in \Si_{[0, t]}, \,   \forall i=1,\ldots, \str \uw \str-1:  v_i < u_{i+1} \}, \qquad  \Si_{[0, t]}^{\mathrm{nl}}:= \Si_{[0, t]} \setminus  \Si_{[0, t]}^{\mathrm{l}}
\eeq
(note that  $\emptyset \in \Si_{[0, t]}^{\mathrm{l}} $). 
Define
 \baq
Q_{t, \mathrm{l}}  & := &    \mathop{\int}\limits_{\Si_{[0, t]}^{\mathrm{l}}}   \,   \mu(\d\uw)   \,      \caT\left\{  \otimes_{w \in \uw} K_w    \right\} \nonumber \\
&=&\sum_{m=0}^{\infty}    \,   \la^{2m}  \mathop{\int}\limits_{0<\ldots <u_i < v_i < u_{i+1}< \ldots< t} \d \uu  \d \uv  \,  \e^{-\i (t-v_m) L_\sys}    F_{v_m-u_m}\ldots       F_{v_2-u_2}  \e^{-\i (u_2-v_1) L_\sys} F_{v_1-u_1} \e^{-\i u_1 L_\sys} \nonumber
\eaq
Its Laplace transform is
\beq
\hat Q_{z,  \mathrm{l} }:=    \int_{\bbR_{+}} \d t\,  \e^{-t z} Q_{t, \mathrm{l}}  =  \left(  z +\i L_\sys -  \la^2 \hat F_z  \right)^{-1}
\eeq
where $\hat F_z=   \int_{\bbR_{+}} \d t \, \e^{-t z} F_t$.
By inverting the Laplace transform we get
\begin{lemma}\label{lemma: inverse laplace of lead}
\beq
\left\norm Q_{t, \mathrm{l}} - \e^{ - \i t L_{\sys} +  \la^2 t M}  \right\norm   \leq    \e^{\la^2 tC}   C  \str\la\str^{2\al_*}
\eeq
where $\al_*=\min(1,\al) $ for $\al \neq 1$, and $  \str\la\str^{2\al_*}=\la^2 \str\log \str \la \str \str $ for $\al=1$.
\end{lemma}
\begin{proof}
We start by writing
\beq   \label{eq: first integral laplace inversion}
 Q_{t, \mathrm{l}}  - \e^{ - \i t L_{\sys} +  \la^2 t M} =   \int_{\caC}  \d z (\hat Q_{z,  \mathrm{l} } -\hat Q_{z, 0 } )  \e^{t z}
\eeq
 where  $\caC$ is a contour  in the complex plane, of the form $\caC = \i \bbR + A\la^2$ with
 $A$ large enough and
 $$
 \hat Q_{z, 0 }=  (  z +\i L_\sys-\la^2M)^{-1}.
 $$  Indeed,
 since the spectrum  $\caE$ of $L_\sys$ is real and $ \norm \hat F_z\norm$  
 is uniformly bounded for $\Re z\geq 0$,  the function $\hat Q_{z, l } $ is analytic in the region $\Re z\geq A\la^2$ if $A$ is large enough and so \eqref{eq: first integral laplace inversion} holds. 

Next, simple algebra yields
\beq
 \hat Q_{z,  \mathrm{l} } - \hat Q_{z, 0 } =\la^2\hat Q_{z,  \mathrm{l}}(\hat F_z -M)  \hat Q_{z, 0 }
\label{resolventid}
\eeq
 and from \eqref{def: m as function of f}
\beq
 M = \sum_{\ve} P_{\ve} \hat F_{-\i\ve}  P_{\ve}.
 \label{resolventid1}
\eeq
We decompose  $\int_\caC = \sum_{\ve \in \caE} \int_{\caC_\ve}$ where  $\caC_\ve:= \{ z\in \caC,  \str z+\i \ve \str= \min_{\ve' \in \caE}\str z+\i \ve' \str  \}$.  For $z\in\caC_\ve$ we write
\beq
\hat F_z -  M = \left(\hat F_{-\i\ve}   - M \right)+\left(\hat F_z - \hat F_{-\i\ve}   \right)
:=V_\ve+W_\ve.
\eeq
The lemma follows if we show
\beq
\la^2 \sum_{\ve\in\caE} \int_{\caC_\ve}  \d z\| \hat Q_{z,  \mathrm{l} } (V_\ve+W_\ve)\hat Q_{z, 0 }  \|\leq C\la^{2\al_*}.
\label{decomposed}
\eeq
 We need the
bounds, for $z \in \caC$,
\baq  \label{eqa priori}
\norm \hat Q_{z, j }  \norm & \leq&  C(\distance(z, -\i\caE))^{-1} \label{eqa priori1} \\[2mm]
 \norm   \hat Q_{z, j} P_\ve  \norm  &\leq&  C \str z+\i\ve \str^{-1}  \label{eqa priori2} 
\eaq
for $j= \mathrm{l},0$.  The bound for $j=1$ follow by straightforward estimates on the Neumann series
\beq
\left(  z+  \i L_{\sys} -  \la^2 \hat  F_z  \right)^{-1}  =    \left(  z +  \i L_{\sys}   \right)^{-1} \sum_{n=0}^{\infty} \left(  \la^2  \hat  F_z     \left(  z + \i L_{\sys}   \right)^{-1} \right)^n
\eeq
using in particular $\norm (z + \i L_{\sys})^{-1} \norm \leq C (\distance(z,-\i \caE))^{-1} \leq  C\str\la\str^{-2}$. The $j=0$ case follows similarly. Using \eqref{resolventid1} we have
$$
V_\ve=\sum_{(\ve',\ve'')\neq(\ve,\ve)}P_{\ve'} \hat F_{-\i\ve}   P_{\ve''}
$$ 
so combining with \eqref{eqa priori1}  and  \eqref{eqa priori2} 
$$
\| \hat Q_{z,  \mathrm{l}} V_\ve\hat Q_{z, 0 }  \|\leq C\sum_{\ve'\neq\ve}  \frac{1}{\str z + \i  \ve \str \str z + \i   \ve' \str}, \qquad \textrm{for} \, z \in \caC_{\ve}. 
$$
 The  $V_\ve$ contribution to \eqref{decomposed} is then bounded by
 \beq\label{vve}
C\la^2 \sum_{\ve' \neq \ve} \int_{\caC_\ve}  \d z    \,  \frac{1}{\str z + \i  \ve \str \str z + \i   \ve' \str}    <  C   \la^{2} |\log \str \la \str |   .
\eeq
As for the $W_\ve$ we  have
\baq
 &&     \norm \hat F_z - \hat F_{-i\ve} \norm  \leq     \int_0^\infty \d t \,  h(t) \left\str \e^{-z t}- \e^{i\ve t}\right\str.
\eaq
Using  $\int_0^{\infty} \d t \, h(t) (1+t)^\al \leq C$, we obtain for $  \Re z \geq 0 $
\beq
 \norm \hat F_z - \hat F_{-i\ve} \norm \leq  C\min\{ \str z+i\ve\str^{\min(\al,1)}  ,1\}.
\eeq
Hence the $W_\ve$ contribution to \eqref{decomposed} is bounded by
\beq\label{wve}
C \la^2 \ \sum_\ve   \int_{\caC_\ve}  \d z  \,   \frac{1}{\str z +\i  \ve \str^2}   \,   \min(1, \str z+i\ve\str^{\min(\al,1)}  )    \leq C    \str \la \str^{2\al_*} .
\eeq
\eqref{decomposed} follows now from \eqref{vve} and \eqref{wve}.
\end{proof}

It remains to estimate the contribution of the non-leading pairs:
\begin{lemma}  \label{lem: bound on nonlead}
For any $\frt >0$;
\beq
\sup_{0 < t < \la^{-2}\frt}  \mathop{\int}\limits_{\Si_{[0, t]}^{\mathrm{nl}}}   \,  \d \uw   \,      \norm  \caT\left\{  \otimes_{w \in \uw} K_{w}    \right\} \norm   \leq\e^{\frt C}    \str\la\str^{2\min(1,\al)}
\eeq
\end{lemma}
\begin{proof}
We start with a bound analogous to \eqref{eqbasic bound on correlations bulk}:
\beq
 \mathop{\int}\limits_{\Si_{[0, t]}^{\mathrm{nl}}}   \,  \d \uw   \,      \norm  \caT\left\{  \otimes_{w \in \uw} K_{w}    \right\} \norm  \leq   \mathop{\int}\limits_{\Si_{[0, t]}^{\mathrm{nl}}}   \,  \d \uw   \prod_{w \in \uw}  \la^2C \str h(v-u)\str
\eeq
By definition, every  $\uw \in \Si_{[0, t]}^{\mathrm{nl}}$ has to contain at least two pairs $(u,v), (u',v')$ such that $u <u' < v$. Choose the first two such pairs and integrate over the coordinates of all other pairs, proceeding as in \eqref{eq: unconstrained estimate}.  This yields
\beq
 \mathop{\int}\limits_{\Si_{[0, t]}^{\mathrm{nl}}}   \,  \d \uw   \prod_{w \in \uw} C\la^2\str h(v-u)  \str \leq C \e^{ \la^2C t  \norm h \norm_1 }  \la^4 q(t)
\eeq 
with
\beq
   q(t) =   \mathop{\int}\limits_{0 <u <u'<v<t ,  u' <v' <t  }  \d u  \d v   \d u'  \d v' \,           \str  h(v-u) h(v'-u')   \str.
\eeq 
$q(t)$ is estimated by first performing the integral over $v'$, which gives a factor $\norm h \norm_1$, and then the one over $u'$, giving $\str v-u\str$:
\baq
q(t)   &\leq&    \norm h \norm_1    \mathop{\int}\limits_{0 <u <v<t}  \d u  \d v      \str  h(v-u)   \str \str v-u \str  \nonumber \\[2mm]
 & \leq &     \norm h \norm_1   \int_0^t       \d u   \left( \max_{0< s< t-u} (1+\str s \str)^{1-\al} \right)     \int_0^{\infty} \d s \, (1+\str s \str)^{\al}  \str h(s) \str  \nonumber  \\[2mm]
  & \leq &    \norm h \norm_1 C     t^{2-2\min(1,\al)} 
\eaq 
Taking $0< t \leq \str\la\str^{-2}\frt$, this yields the lemma.
\end{proof}
The bound \eqref{eq: wc app} (and hence Proposition \ref{prop: weak coupling}) 
 follows by combining  Lemmata \ref{lem: bound on nonlead} and \ref{lemma: inverse laplace of lead}.

\section{Infrared Regularity} \label{sec: ifnrared regularity}

We clarify how our Assumption \ref{ass: decay correlation functions}, formulated in the time-domain, relates   to infrared regularity of the form factor $\phi$. First, we state the auxiliary
\begin{lemma} \label{lem: decay}
	Let $f : (0,\infty) \to \bbR$ be a  measurable function of compact support  and such  that  for some  $0<\ga<1$,
\beq  \label{eq: ga decay}
\str \partial^{n} f(\om) \str  \leq  C_0  \om^{\ga-1-n}, \qquad n=0,1 \eeq 
Then, for any $\be<\ga$,  
$$
\left |\int_{0}^{\infty} f(\om)\e^{\i t \om } \d \om \right |\leq   C(\be) \str t \str^{-\be}
$$
\end{lemma}
\begin{proof}
Define, for $0<\be \leq 1$,
\beq   \label{def: holdernorm}
|f|_\be:= \int_{0}^{\infty} \d\om \,  \sup_{ 0<\epsilon \leq 1}\epsilon^{-\be}|f(\om+\epsilon)-f(\om)| <\infty
\eeq
Then, we claim that
\beq
\left | \int_{0}^{\infty} f(\om)e^{it \om } \d \om \right |\leq   \frac{1}{\str 1-\e^{-\i}\str}  \left(  |f|_\be \str t\str^{-\be} + C\ga^{-1} \str t\str^{-\ga} \right)
\eeq
Indeed, writing (say $t>0$, the case $t<0$ follows by replacing $f$ by $\bar f$)
$$
\int_{0}^{\infty} f(\om )\e^{\i t\om }\d\om =\frac{ t^{-\be}}{\e^{-\i}-1} \int_{0}^{\infty} \frac{f(\om +1/t)-f(\om )}{t^{-\be} }\e^{\i t\om }d\om + \frac{ \e^{-\i}}{\e^{-\i}-1}  \int_{0}^{1/t} f(\om )\e^{\i t\om } \d \om 
$$
the claim follows.    To verify that  $|f|_\be < \infty$ for any $\be<\ga$, we estimate
$$
\epsilon^{-\be}|f(\om+\epsilon)-f(\om)| \leq   \begin{cases}   \om^{-\be}\left( \str  f(\om)  \str+\str  f(\om+\ep)  \str\right)  &  \om \leq \ep \\[2mm]
 \om^{1-\be} \sup_{\om' \in [\om, \om+\ep]} \str \partial f(\om')\str   &  \om > \ep
 \end{cases}
$$
and then use \eqref{eq: ga decay} to bound this by $C \om^{-\be+\ga-1}$, which is integrable for any $\be<\ga$.  
\end{proof}
Using the above Lemma, we can state a sufficient condition on the form factor $\phi$ for Assumption \ref{ass: decay correlation functions} to hold. 
\begin{lemma} 
Fix $\ga>0$ (not necessarily $\ga <1$) and let $\lfloor \ga \rfloor$ be the highest integer not larger than $\ga$. 
Let $f_{\hat q}(\om) :=   \om^{d-1} \str\phi(\om\hat q )\str^2  $, with $\hat q \in \bbS^{d-1}$ and assume that 
\beq \label{eq:  stat phase bound on high momenta}
 \int_1^{\infty} \str \partial^n_{\om} f_{\hat q}(\om) \str   \leq C, \qquad  \textrm{and} \qquad   \forall \om \in (0,2]: \,\, \str \partial^n_{\om} f_{\hat q}(\om) \str   \leq C  \om^{\ga-n},  
\eeq
 for any  $n=0,1, \ldots, \lfloor \ga \rfloor+2$ and uniformly in $\hat q$.  Then, Assumption \ref{ass: decay correlation functions} holds with $\al <\ga$. 
\end{lemma}
\begin{proof}
Let $\theta: \bbR^+ \to [0,1]$ be a $C^{\infty}$ function such that 
\beq
\theta(\om)= \begin{cases} 1  &     \om\leq 1 \\  0  &     \om \geq  2
\end{cases}
\eeq
and write $f= f^{(1)}+f^{(2)}$ with $ f^{(1)}=  \theta f$ and $ f^{(2)}=  (1-\theta) f$ where we omit the parameter $\hat q$ (all bounds are uniform in $\hat q$).
Then  $\int \d \om f^{(2)}(\om)  \e^{-\i t \om}$ decays as $t^{-\lfloor \ga \rfloor+2}$ by the first bound in \eqref{eq:  stat phase bound on high momenta} and partial integration.  To estimate $\int \d \om f^{(1)}(\om)  \e^{-\i t \om}$, we use the second bound in \eqref{eq:  stat phase bound on high momenta} and partial integration $\lfloor \ga \rfloor+1$ times  to extract a factor  $t^{-(\lfloor \ga \rfloor+1)}$. Then we use  Lemma \ref{lem: decay} but with $\ga-{\lfloor \ga \rfloor} \leq 1$ playing the role of $\ga$ in Lemma \ref{lem: decay}, to extract the additional decay  $t^{-\be}$ with $\be<\ga-{\lfloor \ga \rfloor}  $. 
\end{proof}
%
%Therefore,  if we assume that 
%$$\str \partial^{n}_{\om} \str\phi(\om\hat q )\str^2 \str  \leq  C \str \om   \str^{\ga-n-(d-1)}, \qquad   \int_1^{\infty} \d \om  \str \partial^{n} \str\phi\str(\om\hat q ) \str  \str \leq C, \qquad n=0,1,2 $$
%uniformly in $\hat q \in \bbS^{d-1}$,  for some $0<\ga<1$
%then Assumption \ref{ass: decay correlation functions} holds for any $\al<\ga$.  To see this, we first used a stationary phase argument to get the decay $t^{-1}$, and then Lemma \ref{lem: decay} to get additional decay $t^{-\al}$.   Similary, if we assume 
%$$\str \partial^{n}_{\om} \str\phi(\om\hat q )\str^2 \str  \leq  C \str \om   \str^{1+\ga-n-(d-1)}, \qquad   \int_1^{\infty} \d \om  \str \partial^{n} \str\phi\str(\om\hat q ) \str  \str \leq C, \qquad n=0,1,2,3, $$
%uniformly in $\hat q \in \bbS^{d-1}$,  for some $0<\ga<1$, then Assumption \ref{ass: decay correlation functions} holds for any $\al<1+\ga$. Here we used stationary phase to get $t^{-2}$. 
Likewise,  one can derive from Lemma \ref{lem: decay} sufficient conditions on $\phi, \psi_\realinitial, \psi_\realfinal$ to satisfy Assumption \ref{ass: initial states and observables}.

\bibliographystyle{abbrv}
\bibliography{mylibrary11}

\end{document}